\documentclass[acmsmall]{acmart}

\acmJournal{PACMPL}
\acmVolume{1}
\acmNumber{CONF} 
\acmArticle{1}
\acmYear{2018}
\acmMonth{1}
\acmDOI{} 
\startPage{1}

\setcopyright{none}

\usepackage{graphicx}
\usepackage{enumitem}
\usepackage[noend]{algpseudocode}
\usepackage{algorithmicx}
\usepackage[plain,Algorithm]{algorithm}
\usepackage{bbm}
\usepackage{ragged2e}
\usepackage{float}
\usepackage{xcolor}
\usepackage{color}
\usepackage{amsmath}
\usepackage{array}
\usepackage[f]{esvect}
\usepackage{todonotes}
\usepackage{relsize}
\usepackage{xspace}
\usepackage{caption}
\usepackage{amsthm}
\hypersetup{
    colorlinks=true,
    linkcolor=blue,
    filecolor=blue,      
    urlcolor=blue,
    citecolor=blue,
}
\setlist{  
  listparindent=\parindent,
}

\usepackage{listings}
\usepackage{wrapfig}
\usepackage{booktabs}

\newcommand{\foldop}{\textnormal{fold}}
\newcommand{\addop}{\textnormal{add}}
\newcommand{\mulop}{\textnormal{mult}}
\newcommand{\exprogsymbol}{\textsf{P}}
\newcommand{\exlistsymbol}{\textsf{L}}
\newcommand{\exEsymbol}{\textsf{E}}

\newcommand{\revision}[1]{{{{#1}}}}

\newcommand{\sslash}{\mkern-0mu/\mkern-5mu/\mkern-0mu}
\newcommand{\mydots}{\cdots}
\newcommand{\mycdots}{\cdots}
\newcommand{\assign}{:=}
\newcommand{\doubleprime}{{\prime\mkern-2mu\prime\mkern-3mu}}

\newcommand{\tool}{\xspace\textsc{Arborist}\xspace}
\newcommand{\webrobot}{\xspace\textsc{WebRobot}\xspace}
\newcommand{\helena}{\xspace\text{Helena}\xspace}

\newcommand{\newpara}[1]{{{\vspace{3pt} \noindent \emph{\textbf{#1}}}}}

\newcommand{\finalized}[1]{{{{#1}}}}

\newcommand{\mybox}[1]{
\vspace{3pt}
\begin{center}
\fcolorbox{black}{white}{\parbox{.94\linewidth}{
{#1}
}}
\vspace{3pt}
\end{center}
}

\newcommand{\evalmybox}[1]{
\vspace{3pt}
\begin{center}
\fcolorbox{black}{white}{\parbox{.95\linewidth}{
{#1}
}}
\vspace{3pt}
\end{center}
}

\algrenewcommand\algorithmicindent{1em}

\newcommand{\footprint}{\Omega}

\newenvironment{centermath}{\begin{center}$\displaystyle}{$\end{center}}

\newcommand{\emptylist}{\texttt{[]}}

\algnewcommand\algorithmicswitch{\textbf{switch}}
\algnewcommand\algorithmiccase{\textbf{case}}
\algdef{SE}[SWITCH]{Switch}{EndSwitch}[1]{\algorithmicswitch\ #1\ \algorithmicdo}{\algorithmicend\ \algorithmicswitch}%
\algdef{SE}[CASE]{Case}{EndCase}[1]{\algorithmiccase\ #1}{\algorithmicend\ \algorithmiccase}%
\algtext*{EndSwitch}%
\algtext*{EndCase}%

\newcommand{\myprocedure}[2]{
\Statex\hspace{-12pt} {\bf procedure} \textsc{#1}({#2})
\vspace{1pt}
}
\algnewcommand\Input{
\hspace{-10pt}
\textbf{input:}
\vspace{0pt}
}
\algnewcommand\Output{
\hspace{-10pt}
\textbf{output:}
\vspace{0pt}
}

\newcommand{\inputcontexts}{\textit{GetContexts}}
\newcommand{\antiunifyfunc}{\textsc{AntiUnify}}
\newcommand{\parametrizefunc}{\textsc{Parametrize}}

\newcommand{\getannot}{\textit{Footprint}}
\newcommand{\newstate}{\textit{MkState}}

\newcommand{\evalftastate}[3]{#1 \vdash #2 \rightsquigarrow (#3)}
\newcommand{\evalftatransition}[3]{#1 \vdash #2 \rightsquigarrow (#3)}

\newcommand{\absolutexpath}{{\raisebox{1pt}{$\chi$}}}
\newcommand{\fullxpath}{\absolutexpath}

\newcommand{\traceconcat}{\listconcat}
\newcommand{\emptytrace}{\emptylist}

\newcommand{\irule}[2]{\mkern-2mu\displaystyle\frac{#1}{\vphantom{,}#2}\mkern-2mu}

\newcommand{\inputcontext}{C}
\newcommand{\outputval}{O}

\newcommand{\antiunify}[3]{#1 \vdash #2 \twoheadrightarrow #3}

\newcommand{\grammarsymbol}{\textsf{s}}
\newcommand{\grammareq}{::=}

\newcommand{\program}{P}
\newcommand{\programsymbol}{\textsf{P}}
\newcommand{\skipprog}{\emph{skip}}

\newcommand{\seqprog}{\emph{Seq}}

\newcommand{\expression}{E}
\newcommand{\statement}{\expression}
\newcommand{\expressionsymbol}{\textsf{E}}
\newcommand{\statementsymbol}{\expressionsymbol}

\newcommand{\clickexpression}{\emph{Click}}
\newcommand{\clickstatement}{\clickexpression}

\newcommand{\scrapetextexpression}{\emph{ScrapeText}}
\newcommand{\scrapetextstatement}{\scrapetextexpression}

\newcommand{\scrapelinkexpression}{\emph{ScrapeLink}}
\newcommand{\scrapelinkstatement}{\scrapelinkexpression}

\newcommand{\downloadexpression}{\emph{Download}}
\newcommand{\downloadstatement}{\downloadexpression}

\newcommand{\sendkeysexpression}{\emph{SendKeys}}

\newcommand{\senddataexpression}{\emph{EnterData}}
\newcommand{\enterdataexpression}{\senddataexpression}
\newcommand{\senddatastatement}{\senddataexpression}
\newcommand{\enterdatastatement}{\senddatastatement}

\newcommand{\gobackexpression}{\emph{GoBack}}

\newcommand{\extracturlexpression}{\emph{ExtractURL}}

\newcommand{\forselectorloop}{\emph{ForSelectors}}
\newcommand{\forselectorsloop}{\emph{ForSelectors}}
\newcommand{\fordataloop}{\emph{ForData}}
\newcommand{\whileloop}{\emph{While}}

\newcommand{\inputvar}{x}

\newcommand{\selectorexpr}{se}
\newcommand{\selectorexprsymbol}{\textsf{se}}
\newcommand{\selectorvar}{y}
\newcommand{\concreteselector}{\selectorexpr}

\newcommand{\dataexpr}{de}
\newcommand{\dataexprsymbol}{\textsf{de}}
\newcommand{\datavar}{z}
\newcommand{\datakey}{\emph{key}}
\newcommand{\dataexprlist}{lst}

\newcommand{\predicate}{\psi}
\newcommand{\domtag}{t}
\newcommand{\domattribute}{\tau}

\newcommand{\constantstring}{str}

\newcommand{\clickexpr}{\emph{Click}}

\newcommand{\cfg}{G}
\newcommand{\fta}{\mathcal{A}}
\newcommand{\ftastatespeculated}{p}
\newcommand{\ftastate}{q}
\newcommand{\ftastatep}{p}
\newcommand{\ftastater}{r}
\newcommand{\ftastates}{Q}
\newcommand{\transitions}{\Delta}
\newcommand{\ftatransitions}{\transitions}
\newcommand{\transition}{\delta}
\newcommand{\ftatransition}{\transition}
\newcommand{\transitionoperator}{f}
\newcommand{\ftatransitionoperator}{\transitionoperator}
\newcommand{\transitionarrow}{\rightarrow}
\newcommand{\ftatransitionarrow}{\transitionarrow}
\newcommand{\ftalang}{L}
\newcommand{\ftafinalstates}{\ftastates_{f}}
\newcommand{\ftaalphabet}{F}

\newcommand{\annot}{\Omega}

\newcommand{\inputcontexttooutput}{\mapsto}

\newcommand{\inputtooutput}{\mapsto}

\newcommand{\subfta}{\textit{SubFTA}}

\newcommand{\eval}[3]{#1 \vdash #2 : #3}

\newcommand{\inputdata}{I}

\newcommand{\bindsto}{\mapsto}

\newcommand{\env}{\Gamma}

\newcommand{\domtrace}{\Pi}
\newcommand{\dom}{\pi}

\newcommand{\listconcat}{\mkern-0mu \texttt{+} \mkern-1mu \texttt{+} \mkern1mu}

\newcommand{\actiontrace}{A}
\newcommand{\action}{a}

\newcommand{\datavalue}{v}

\begin{document}

\title{Efficient Bottom-Up Synthesis for Programs with Local Variables}


\author{Xiang Li*}
\affiliation{
  University of Michigan, Ann Arbor, USA
}
\email{xkevli@umich.edu}          

\author{Xiangyu Zhou*}
\affiliation{
    University of Michigan, Ann Arbor, USA
}
\email{xiangyz@umich.edu}         

\author{Rui Dong}
\affiliation{
  University of Michigan, Ann Arbor, USA
}

\author{Yihong Zhang}
\affiliation{
  University of Washington, USA
}

\author{Xinyu Wang}
\affiliation{
  University of Michigan, Ann Arbor, USA
}

\begin{abstract}

\revision{We propose a new synthesis algorithm that can \emph{efficiently} search programs with \emph{local} variables (e.g., those introduced by lambdas). 
Prior bottom-up synthesis algorithms are not able to evaluate programs with \emph{free local variables}, and therefore cannot effectively reduce the search space of such programs (e.g., using standard observational equivalence reduction techniques), making synthesis slow. 
Our algorithm can reduce the space of programs with local variables. 
The key idea, dubbed \emph{lifted interpretation}, is to lift up the program interpretation process, from evaluating one program at a time to simultaneously evaluating all programs from a grammar. 
Lifted interpretation provides a mechanism to systematically enumerate all binding contexts for local variables, thereby enabling us to evaluate and reduce the space of programs with local variables.  
Our ideas are instantiated in the domain of web automation. The resulting tool, $\tool$, can automate a significantly broader range of challenging tasks more efficiently than state-of-the-art techniques including $\webrobot$ and $\helena$.}

\end{abstract}



\def\thefootnote{*}\footnotetext{Xiang Li and Xiangyu Zhou contributed equally to this work.}\def\thefootnote{\arabic{footnote}}


\maketitle

\section{Introduction}\label{sec:intro}

Web automation can automate web-related tasks such as scraping data and filling web forms. 
While an increasing number of populations have found it useful~\cite{uipath-studiox-webinar-slides,chasins2019thesis,katongo2021towards}, it is notoriously difficult to create web automation programs~\cite{krosnick2021understanding}. 
Let us consider the following example, which we also use as a running example throughout the paper.

\begin{example}
{\small\url{https://haveibeenpwned.com/}} is a website where one could check whether or not an email address has been compromised (i.e., ``pwned''). 
Figures~\ref{fig:intro:not-pwned} and~\ref{fig:intro:pwned} show the webpage DOMs for an email with no pwnage detected and for a pwned email respectively; both figures are simplified from the original DOMs solely for presentation purposes.
Consider the task of scraping the pwnage text for each email from a list of emails.\footnote{This is a real-life task from the iMacros forum: \url{https://forum.imacros.net/viewtopic.php?f=7&t=26683}.} 
Figure~\ref{fig:pwnage-program} shows an automation program $P$ for this task. 
While the high-level logic of $P$ is rather simple, implementing it turns out to be very difficult. 
First, one must implement the right control-flow structure, such as the loop in $P$ with three instructions. 
Second, each instruction must use a \emph{generalizable} selector to locate the desired DOM element \emph{for all emails} across all iterations. 
For example, line 4 from Figure~\ref{fig:pwnage-program} uses one generalizable selector that utilizes the {\small\texttt{aria-expanded}} attribute. This attribute is necessary for generalization, since both messages (pwned and not pwned) are \emph{always} in the DOM but which one to render is determined by this attribute's value. 
On the other hand, a full XPath expression {\small\texttt{/html/body/div/div/div/div/h2}} locates the element with ``{\textsf{Good news -- no pwnage found!}}'' in \emph{both} Figure~\ref{fig:intro:not-pwned} and Figure~\ref{fig:intro:pwned}, which is \emph{not} desired. 
In general, one has to try many \emph{candidate selectors} before finding a generalizable one --- in other words, this is fundamentally a search problem.
\label{ex:intro:example}
\end{example}

\begin{figure}
\centering
\begin{minipage}[b]{.49\linewidth}
\hspace{3pt}
\includegraphics[height=4cm]{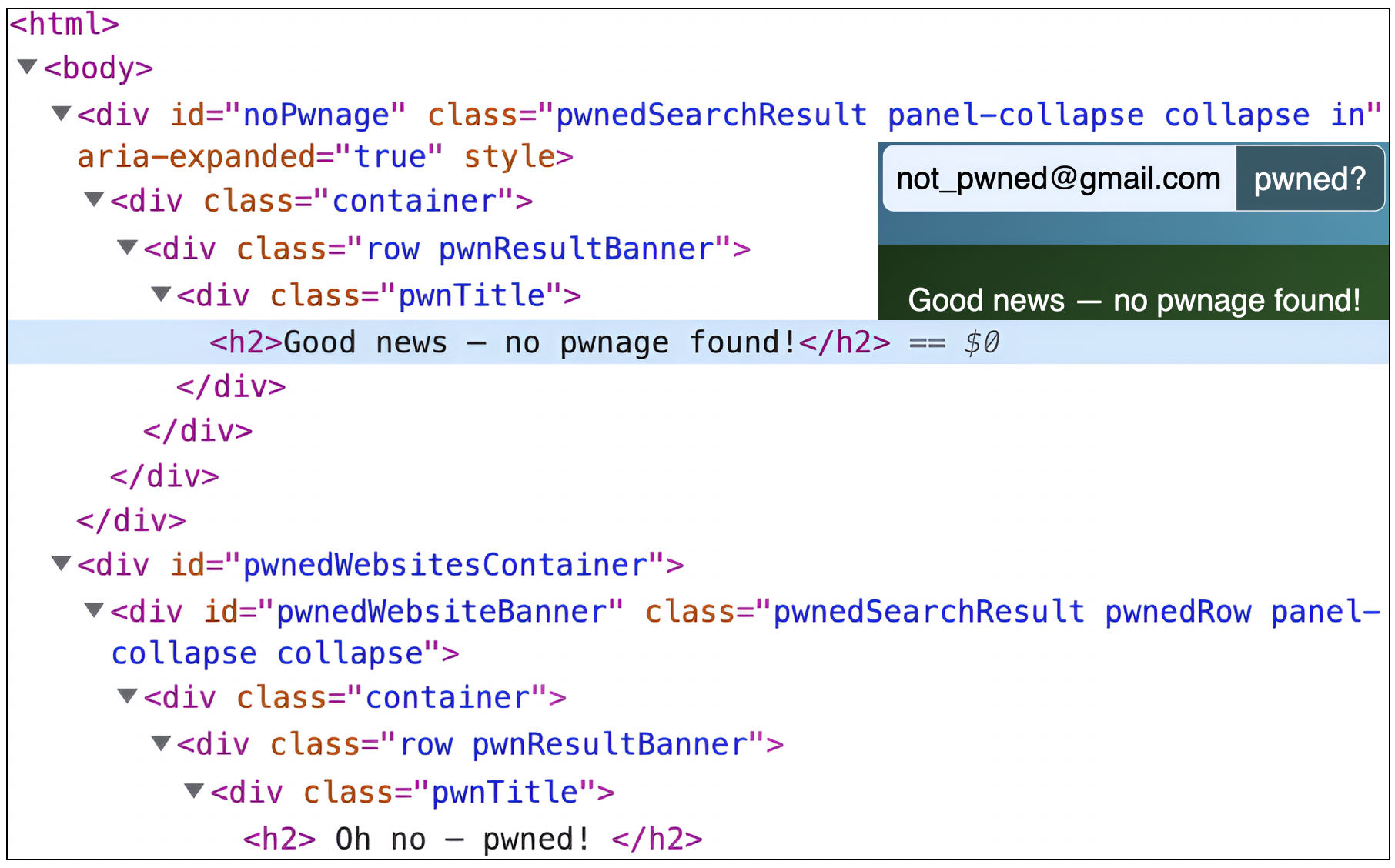}
\vspace{-5pt}
\caption{DOM when no pwnage detected.}
\label{fig:intro:not-pwned}
\end{minipage}
\hspace{2pt}
\begin{minipage}[b]{.49\linewidth}
\includegraphics[height=4cm]{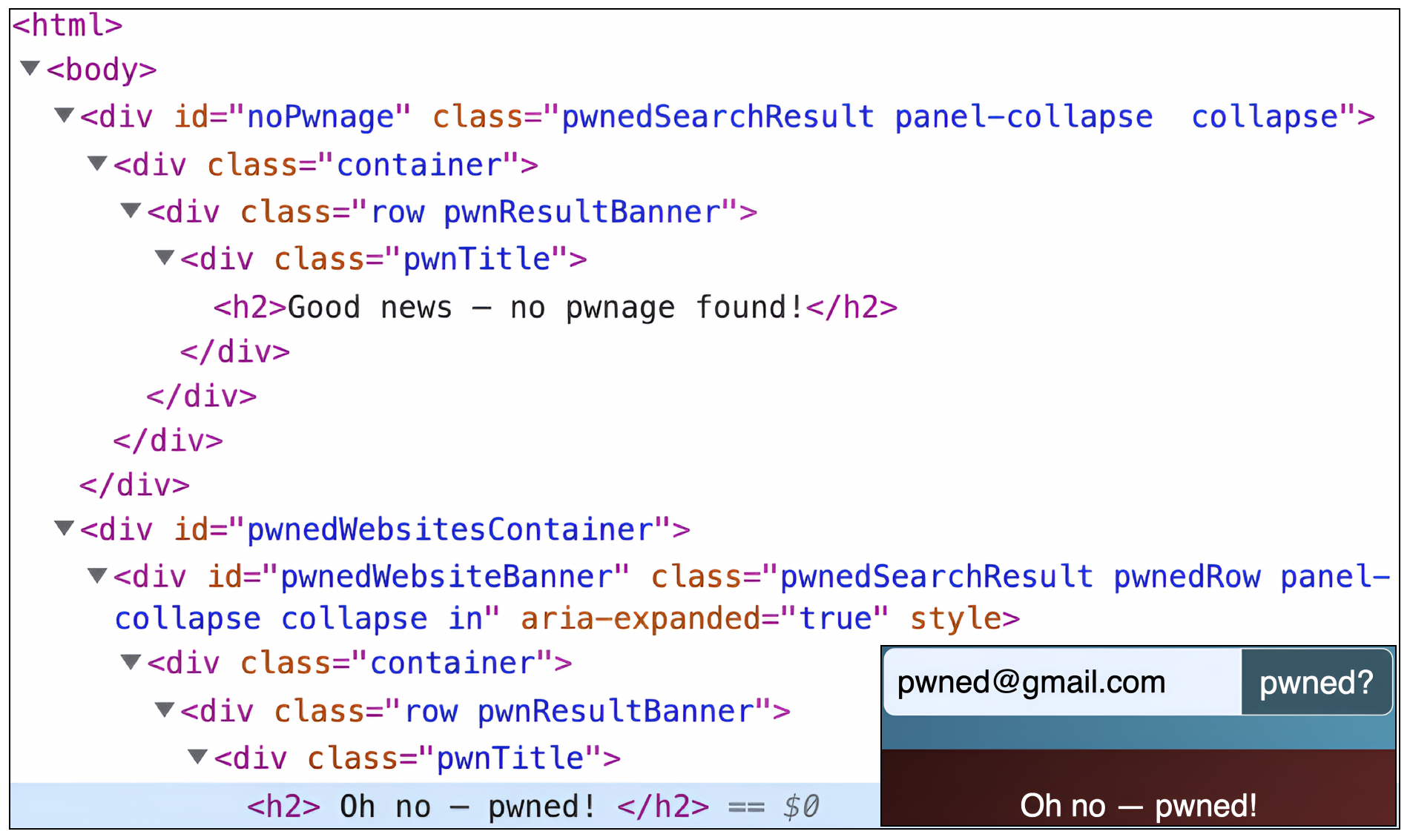}
\vspace{-5pt}
\caption{DOM when pwnage detected.}
\label{fig:intro:pwned}
\end{minipage}
\\[5pt] 
\lstset{%
  basicstyle=\ttfamily\footnotesize,
  keywordstyle=\bfseries,
  breaklines=true,
  showstringspaces=false,
  columns=flexible,
  numbers=left,
  numberstyle=\tiny\color{gray},
  commentstyle=\color{gray},
  framexleftmargin=0pt,
  xleftmargin=10pt,
  tabsize=2, 
  morekeywords={foreach,in}, 
  emph={}, 
}
\begin{minipage}{.98\linewidth}
\begin{lstlisting}
foreach email in list_of_emails: 
    EnterData( email, //input[@id='account'] )               # enter each email to search box
    Click( //button[@id='searchPwnage'] )                    # click search button         
    ScrapeText( //div[@aria-expanded='true']//h2 )           # scrape pwnage message 
\end{lstlisting}
\end{minipage}
\vspace{-10pt}
\caption{A web automation program that scrapes pwnage result for each email address from a list.}
\vspace{-10pt}
\label{fig:pwnage-program}
\end{figure}

\newpara{Synthesizing web automation programs.}
In general, implementing web automation programs requires creating both the desired control-flow structure (with arbitrarily nested loops potentially) and identifying generalizable selectors for all of its instructions; this is very hard and time-consuming. 
$\webrobot$~\cite{dong2022webrobot} is a state-of-the-art technique that allows non-experts to create web automation programs from a short trace $\actiontrace$ of user actions (e.g., clicking a button, scraping text). 
A key underpinning idea is its \emph{trace semantics}: given a program $\program$, trace semantics outputs a sequence $\actiontrace'$ of actions that $P$ executes, by resolving any free variables (such as the \emph{local} loop variable from Figure~\ref{fig:pwnage-program}).  
Then, one can check if $P$ satisfies $\actiontrace$ by checking if $\program$'s output trace $\actiontrace'$ matches $\actiontrace$. 

\vspace{5pt}
\newpara{Key challenge: synthesizing programs with local variables.}
While trace semantics significantly bridges the gap between programming-by-demonstration (PBD) and programming-by-example (PBE) by enabling a ``guess-and-check'' style synthesis approach for web automation~\cite{dong2022webrobot,chen2022semanticon,miwa,dilogics}, state-of-the-art synthesis algorithms unfortunately fail to scale to challenging web automation tasks due to the heavy use of local variables in those tasks.
Consider the enormous space of potential loop bodies to be searched. All these bodies use loop variables and therefore must be evaluated under an \emph{extended} context that \emph{also} binds such local variables to values. 
Conventional observational equivalence (OE) from the program synthesis literature~\cite{udupa2013transit,albarghouthi2013recursive} {fundamentally} cannot reduce this space: they track bindings for only input variables, but not for local variables. 
As a result, they cannot build equivalence classes for loop bodies, and hence fall back to enumeration. 
The only existing work (to our best knowledge) that pushed the boundary of OE is RESL~\cite{peleg2020oopsla}. 
Briefly, its idea is to ``infer'' bindings for local variables given a higher-order sketch, which enables applying OE over the space of lambda bodies. 
A fundamental problem, however, is the data dependency across iterations (for functions like fold). 
Its implication to synthesis is succinctly summarized by RESL as a ``chicken-and-egg'' problem: 
we need output values of programs in order to apply OE for more efficient synthesis, but we need the programs first in order to obtain their output values. 
This \emph{cyclic dependency} fundamentally limits all prior work (such as RESL, among others~\cite{feser2015synthesizing,smith2016mapreduce}) to sketch-based approaches with crafted binding-inference rules, which still resort to enumeration in many cases. 
The general problem of how to reduce the space of programs with local variables remains open~\cite{peleg2020oopsla}.

\vspace{5pt}
\newpara{Our idea: lifted interpretation.}
The same problem occurs in our domain: as we will show shortly, loops in web automation use local variables and exhibit data dependency across iterations. 
\revision{In this work, we propose a new algorithm that can apply OE-based reduction for \emph{any programs}, without requiring binding inference. 
We build upon the OE definition from RESL~\cite{peleg2020oopsla}: two programs belong to the same equivalence class if they share the same \emph{context} and yield the same output. Notably, ``context'' here is a binding context with \emph{all free variables} including both input and local variables. 
Our key insight can be summarized as follows.}

\mybox{We can compute an equivalence relation of programs based on OE under \emph{all reachable} contexts, by creating equivalence classes \emph{simultaneously while} evaluating \emph{all programs} from a given tree grammar (with respect to a given input). 
Furthermore, we can use (a generalized form of) finite tree automata to compactly store the equivalence relation.}
Here, a \emph{reachable context} is one that emerges during the execution of at least one program from the grammar, for a given input. 
Furthermore, programs {rooted at the same grammar symbol} share reachable contexts. 
For instance, the loop from Figure~\ref{fig:pwnage-program} introduces the same binding context for loop bodies that can be put inside it.
However, computing reachable contexts requires program evaluation which again requires reachable contexts --- this is the aforementioned ``chicken-and-egg'' problem described in RESL~\cite{peleg2020oopsla}. 
To break this cycle, our key insight is to \emph{simultaneously} evaluate \emph{all} programs \emph{top-down} from the grammar, \emph{during} which we construct equivalence classes of programs \emph{bottom-up} based on their outputs under their reachable contexts. 
\revision{This idea essentially \emph{lifts up} an interpreter from evaluating one single program at a time to \emph{simultaneously} evaluating \emph{all} programs from a grammar, with respect to a given input. 
This \emph{lifted interpretation} process allows us to systematically enumerate all reachable contexts, and hence build equivalence classes for all programs including those with local variables.}

\vspace{5pt}
\newpara{General idea, instantiation, and evaluation.}
\revision{In the rest of this paper, we first illustrate how our idea works in general in Section~\ref{sec:overview} using a small functional language. 
Then, in Section~\ref{sec:alg}, we present an instantiation of our approach in the domain of web automation.}
We implement this instantiation in a tool called $\tool$\footnote{$\tool$ is a specialist that can manage a lot of trees (i.e., programs), even if they have local variables.}. 
Our evaluation results show that $\tool$ can solve more challenging benchmarks using much less time, significantly advancing the state-of-the-art for web automation. 

\revision{
\vspace{5pt}
\newpara{Contributions.}
This paper makes the following contributions. 
\begin{itemize}[leftmargin=*]
\setlength\itemsep{1pt}
\item 
We propose a new synthesis algorithm, based on {lifted interpretation} and finite tree automata, that can reduce the search space of programs with local variables. 
\item 
We instantiate this approach in the domain of web automation and develop a new programming-by-demonstration algorithm. 
\item 
We implement our technique in a tool called $\tool$ and evaluate it on \finalized{131} benchmarks. 
Our results highlight that the idea of lifted interpretation yields a significantly faster synthesizer. 
\end{itemize}
}
\section{Preliminaries}\label{sec:prelim}

In this section, we review the standard concepts of observational equivalence (OE) and finite tree automata (FTAs) from the literature, focusing on their application to program synthesis.

\subsection{Synthesis using Observational Equivalence}\label{sec:prelim:oe}

Observational equivalence (OE) was originally proposed by \citet{hennessy1980observing} to define the semantics of concurrent programs, and has been used widely within the programming languages community. 
Intuitively, two terms are observationally equivalent whenever they are interchangeable in \emph{all observable contexts.} 
In the field of programming-by-example (PBE), OE has been utilized to reduce the search space of programs, typically in \emph{bottom-up} synthesis algorithms.  

\newpara{Bottom-up synthesis.}
Bottom-up algorithms synthesize programs by first constructing smaller programs which are later used as building blocks to create bigger ones. 
Specifically, the algorithm begins with an initial set $W$ containing all atomic programs of size 1 (e.g., input variables, constants), and then iteratively grows $W$ by adding new programs of larger sizes that are composed of those already in $W$.  
The algorithm terminates when $W$ has a program $P$ that meets the given specification. 
For instance, for a language that includes variable $x$ and integer constants $1$ and $2$, $W$ is initially $\{ x, 1, 2 \}$ but later will contain more terms such as $x+1$ and $x+2$, assuming $+$ operator is allowed by the language. If the specification is given as an input-output example pair $(1, 3)$ meaning ``return $3$ when $x = 1$'', then $x+2$ is a correct program whereas $x+1$ is not. 

\newpara{Observational equivalence reduction.}
Bottom-up synthesis often uses observational equivalence to reduce the program space in order to improve the search efficiency. 
The key idea is to \emph{not} add a new program $P$ to $W$, if there already exists some $P' \in W$ that behaves the same as $P$ \emph{observationally}. 
In particular, existing PBE work~\cite{albarghouthi2013recursive,udupa2013transit} defines two programs $P_1, P_2$ to be observationally equivalent if they yield the same output \emph{on each input example}. 
This idea keeps only programs that are observationally distinct (given input examples), thereby reducing the size of $W$ and accelerating the search. 
\revision{RESL~\cite{peleg2020oopsla} further generalizes OE to consider an \emph{extended} context that also includes local variables: two programs are observationally equivalent if they yield the same output given a shared context (which may include local variables). However, conventional bottom-up synthesis algorithms no longer work under this OE definition, as they cannot evaluate programs with free local variables before knowing their binding context.}

\subsection{Synthesis using Finite Tree Automata}\label{sec:prelim:fta}

OE essentially defines an equivalence relation of programs, which can be stored using finite tree automata (FTAs)~\cite{wang2017program}.

\newpara{Finite tree automata.}
Finite tree automata (FTAs)~\cite{comon2008tree} deal with tree-structured data: they generalize standard finite (word) automata by accepting trees rather than words/strings. 

\begin{definition}[Finite Tree Automata]\label{def:FTAs}
A (bottom-up) finite tree automaton (FTA) over alphabet $\ftaalphabet$ is a tuple $\fta = (\ftastates, \ftaalphabet, \ftafinalstates, \ftatransitions)$, where $\ftastates$ is a set of states, $\ftafinalstates \subseteq \ftastates$ is a set of final states, and $\transitions$ is a set of transitions of the form $\transitionoperator(\ftastate_1, \mydots, \ftastate_n) \transitionarrow \ftastate$ where $\ftastate_1, \mydots, \ftastate_n, \ftastate \in \ftastates$ and $\transitionoperator \in \ftaalphabet$. 
A term $t$ is \emph{accepted} by $\fta$ if $t$ can be rewritten to a final state according to the transitions (i.e., rewrite rules).
The language of $\fta$, denoted $\ftalang(\fta)$, is the set of terms accepted by $\fta$. 
\end{definition}

\vspace{-3pt}
\newpara{Notations.}
We  also use $\fta = (\ftafinalstates, \transitions)$ as a simpler notation, since  $\ftastates$ and $\ftaalphabet$ can be  determined by $\transitions$. 
We use $\subfta(\ftastate, \fta)$ to mean the sub-FTA of $\fta$ that is rooted at state $\ftastate$.

\newpara{Program synthesis using FTAs.}
Given a tree grammar $\cfg$ defining the syntax of a language and given an input-output example $\mathcal{E} = (\mathcal{E}_{\emph{in}}, \mathcal{E}_{\emph{out}})$, we can construct an FTA $\fta = (\ftastates, \ftaalphabet, \ftafinalstates, \ftatransitions)$, such that $\ftalang(\fta)$ contains all programs from $\cfg$ (up to a finite size) that satisfy $\mathcal{E}$. 
In particular, the alphabet $\ftaalphabet$ consists of all operators from $\cfg$. 
We have a state $\ftastate^{c}_{s} \in \ftastates$ if there exists a program rooted at symbol $s$ from $\cfg$ that outputs $c$ given input $\mathcal{E}_{\emph{in}}$. 
We have a transition $\ftatransitionoperator(\ftastate^{c_1}_{s_1}, \mydots, \ftastate^{c_n}_{s_n}) \ftatransitionarrow \ftastate^{c_0}_{s_0} \in \ftatransitions$ if applying function $\ftatransitionoperator$ on $c_1, \mydots, c_n$ yields $c_0$. 
A state $\ftastate^{c}_{s}$ is marked final if $c = \mathcal{E}_{\emph{out}}$ and $s$ is a start symbol of $\cfg$. 
Once $\fta$ is constructed, one can extract a program $P$ from $\ftalang(\fta)$ heuristically (e.g., smallest in size) and return $P$ as the final synthesized program~\cite{wang2017program}.

\newpara{Remarks.}
Every state $\ftastate^{c}_{s} \in \ftastates$ represents an equivalence class of all programs rooted at grammar symbol $s$ that produce the same value $c$ on $\mathcal{E}_{\emph{in}}$. 
In other words, $\ftastate^{c}_{s}$ stores all \emph{observationally equivalent} programs. 
\revision{To our best knowledge, all existing FTA-based synthesis techniques~\cite{wang2017program,wang2017synthesis,wang2018relational,yaghmazadeh2018automated,miltner2022bottom} are based on the notion of OE that considers only input variables. In other words, all existing techniques resort to enumeration of programs with free local variables~\cite{peleg2020oopsla}.}

\revision{\section{Lifted Interpretation}\label{sec:overview}

This section illustrates how the general idea of lifted interpretation works on a simple functional language. 
Section~\ref{sec:alg} will later describe a full-fledged instantiation to the domain of web automation.

\subsection{A Simple Programming-by-Example Task}\label{sec:overview:task}

\begin{example}
Given the simple functional language from Figure~\ref{fig:overview:syntax}, let us consider the following programming-by-example (PBE) task: synthesize a program that returns $7$ given input list $[1,2,4]$. 
Suppose the intended program is: 
\[
P_1: 
\foldop \big( x, (acc, elem) \Rightarrow \addop( acc, elem ) \big)
\]
which calculates the sum of all elements from the input list $x$. 
Consider another program: 
\[
P_2: 
\foldop\big( x, (acc, elem) \Rightarrow \addop( \mulop( acc,  2 ) , 1 ) \big) 
\]
which returns the same output $7$ as $P_1$, given the example input $[1, 2, 4]$. 
Notably, while the lambda bodies in $P_1$ and $P_2$ are different, they share the same context-output behaviors (or \emph{footprint}), for the given input list. 
Please see Table~\ref{tab:overview:footprints} which shows their local variable bindings and corresponding output values across all iterations. 
In what follows, we will illustrate how to synthesize $P_1$ and $P_2$ from the input-output example $[1, 2, 4] \mapsto 7$, using our lifted interpretation idea. 
\label{ex:overview:example}
\end{example}

\begin{figure}[!t]
\centering
\begin{minipage}{.99\linewidth}
\centering
\[
\arraycolsep=3pt
\def\arraystretch{1.1}
\begin{array}{rll}
 
P & 
\grammareq & 
\foldop \big( L, (acc, elem) \Rightarrow E \big) \\ 

E & 
\grammareq & 
acc \ | \ 
elem \ | \ 
1 \ | \ 
2 \ | \ 
\addop(E, E) \ | \ 
\mulop(E, E) \\ 

L & 
\grammareq & 
x \\ 

\end{array}
\]
\end{minipage}
\vspace{-5pt}
\caption{A simple functional language. Here, $x$ is the \emph{input variable}, which is a list of integers. We simplify the standard $\foldop$ operator to use a default seed of $0$ (which is implicit and not shown as an argument). Note that $\foldop$ introduces two \emph{local} variables: $acc$ is the accumulator, and $elem$ will be bound to each element from $L$. $E$ is the lambda body, which may use local variables $acc$ and $elem$. The ``$\addop$'' and ``$\mulop$'' operators are the standard addition and multiplication.}
\label{fig:overview:syntax}
\end{figure}

\begin{table}[!h]
\small 
\centering
\caption{Footprints of lambda bodies from $P_1$ and $P_2$ respectively, across all iterations.}
\label{tab:overview:footprints}
\vspace{-5pt}
\begin{tabular}{ r | c | c | c }
\toprule 
& local variable bindings & $P_1$: $\addop( acc, elem )$ & $P_2$: $\addop( \mulop( acc,  2 ) , 1 )$
\\ \midrule
\textnormal{iteration 1} & $acc \mapsto 0, \ \ elem \mapsto 1$ & 1 & 1 
\\ 
\textnormal{iteration 2} & $acc \mapsto 1, \ \ elem \mapsto 2$ & 3 & 3 
\\ 
\textnormal{iteration 3} & $acc \mapsto 3, \ \ elem \mapsto 4$ & 7 & 7 
\\
\bottomrule
\end{tabular}
\end{table}

\subsection{FTAs based on Observational Equivalence}\label{sec:overview:fta}

Let us first present a new FTA-based data structure that our lifted interpretation approach utilizes to succinctly encode equivalence classes of programs. 
The main ingredient is its generalization of the FTA state definition from prior work~\cite{wang2017program}: our state includes a \emph{context} $\inputcontext$ which contains information (e.g., all variable bindings) to evaluate programs with free local variables. 
In particular, we define an FTA state $\ftastate$ as a pair: 
\[
( 
\grammarsymbol, 
\footprint 
)
\ \ 
\emph{where} 
\ \ 
\footprint = 
\{ \ 
\inputcontext_1 \inputcontexttooutput \outputval_1, \mydots, \inputcontext_l \inputcontexttooutput \outputval_l 
\ \} 
\]
Here, $\grammarsymbol$ is a grammar symbol, and $\footprint$ is a \emph{footprint} which maps a context $\inputcontext_i$ to an output $\outputval_i$. 
Each entry $\inputcontext_i \mapsto \outputval_i$ is called a \emph{behavior}; so a footprint is a set of behaviors. 
Intuitively, if there exists a program $\program$ rooted at $\grammarsymbol$ that evaluates to $\outputval_i$ under $\inputcontext_i$ for all $i \in [1, l]$, then our FTA $\fta$ has a state $( \grammarsymbol, \{  \inputcontext_1 \inputcontexttooutput \outputval_1, \mydots, \inputcontext_l \inputcontexttooutput \outputval_l \} )$, and vice versa. 
A context is \emph{reachable} if it can actually emerge, when executing programs in a given grammar for a given input. 
Given a finite grammar, if all programs terminate, then the number of reachable contexts is finite. 
Our work assumes a finite number of reachable contexts, which we believe is a reasonable assumption for program synthesis.


The remaining definitions are relatively standard. 
Our alphabet $\ftaalphabet$ includes all operators from the programming language. 
A transition $\ftatransition \in \ftatransitions$ is of the form $\ftatransitionoperator(\ftastate_1, \mydots, \ftastate_n) \ftatransitionarrow \ftastate$ which connects multiple states to one state. 
However, because our state definition is more general, the condition under which to include a transition now becomes different from prior work~\cite{wang2017program}.
Specifically, $\fta$ includes a transition $\ftatransition = \ftatransitionoperator(\ftastate_1, \mydots, \ftastate_n) \ftatransitionarrow \ftastate$, if for every behavior $\inputcontext \inputcontexttooutput \outputval$ in $\ftastate$, we have behavior $\inputcontext_{k} \inputcontexttooutput \outputval_{k}$ in $\ftastate_k$ (for all $k \in [1, n]$), such that according to $\ftatransitionoperator$'s semantics and given $\inputcontext_1 \inputcontexttooutput \outputval_1, \mydots, \inputcontext_n \inputcontexttooutput \outputval_n$, evaluating $\ftatransitionoperator$ under context $\inputcontext$ indeed yields output $\outputval$.

\subsection{Illustrating Lifted Interpretation}\label{sec:overview:explanation}

Now we are ready to explain how our lifted interpretation idea works for Example~\ref{ex:overview:example}. 

\newpara{Setup.}
First, we build an FTA $\fta_s = ( \{ \ftastate_1 \}, \ftatransitions_s)$ --- see  Figure~\ref{fig:overview:grammar-fta} --- for the grammar in Figure~\ref{fig:overview:syntax}. 
We will later apply lifted interpretation to $\fta_s$. 
Each state in $\fta_s$ is annotated with a grammar symbol and an \emph{empty} footprint. 
Notice the cyclic transitions around $\ftastate_3$, due to the recursive $\addop$ and $\mulop$ productions. This induces an infinite space of programs --- our approach finitizes the grammar by bounding the size of its programs, as standard in the literature~\cite{wang2017program}.

\newpara{Applying lifted interpretation.}
Then, we use lifted interpretation to ``evaluate'' $\fta_s$ under an initial context, which would eventually produce another FTA $\fta_e = ( \{ \ftastate_{11} \}, \ftatransitions_e)$ (see a part of it in Figure~\ref{fig:overview:final-fta}). 
Different from $\fta_s$, $\fta_e$ will cluster programs (including all sub-programs) into equivalence classes based on OE. 
Figure~\ref{fig:overview:final-fta-annotations} shows the annotations (i.e., grammar symbols and footprints) for $\fta_e$.

At a high-level, lifted interpretation traverses $\fta_s$ systematically, computes reachable contexts on-the-fly during traversal, and most importantly, constructs the equivalence classes simultaneously given these reachable contexts. 
In particular, given an initial context $\inputcontext$ and an FTA $\fta_s = (\ftastates_f, \ftatransitions)$ with final states $\ftastates_f$ and transitions $\ftatransitions$, lifted interpretation returns an FTA $\fta_e = (\ftastates'_f, \ftatransitions')$ with final states $\ftastates'_f$ and transitions $\ftatransitions'$, such that if two programs from $\ftalang(\fta_s)$ yield the same output under $\inputcontext$, then they belong to the same (final) state in $\fta_e$. 
More specifically, we have: 
\[
\begin{array}{rcc}
\ftastates'_f 
& = &
\big\{ 
\ 
\ftastate'_i
\ \ \big| \ \
\ftastate_i \in \ftastates_f, \ \ 
\evalftastate{
\inputcontext
}{
\ftastate_i; \ftatransitions
}{
\ftastate'_i, \ftatransitions'_i
}
\ 
\big\}
\\[5pt]
\ftatransitions' 
& = & 
\bigcup_{\ftastate_i \in \ftastates_f} 
{
\ftatransitions'_i
}
\ \ \ 
\textnormal{ where }
\evalftastate{
\inputcontext
}{
\ftastate_i; \ftatransitions
}{
\ftastate'_i, \ftatransitions'_i
}
\end{array}
\]
That is, if any final state $\ftastate_i$ in $\fta_s$ ``evaluates to'' a state $\ftastate'_i$ under $\inputcontext$, then $\ftastate'_i$ is a final state of $\fta_e$. 
This process also yields a set $\ftatransitions'_i$ of transitions for each $\ftastate_i$, which are added as transitions to $\fta_e$.

\newpara{Key judgment.}
The key judgment that drives the lifted interpretation process is of the following form. Note that this judgment is non-deterministic; that is, a state may evaluate to multiple states. 
\[
\begin{array}{c}
\evalftastate{
\inputcontext
}{
\ftastate; \ftatransitions
}{
\ftastate', \ftatransitions'
}
\end{array}
\]
It reads as follows: given context $\inputcontext$, state $\ftastate$ (with respect to transitions $\ftatransitions$) evaluates to state $\ftastate'$ with transitions $\ftatransitions'$. 
The guarantee is that: for all programs $P$ from $\ftalang(\{ \ftastate \}, \ftatransitions)$ that yield the same output under context $\inputcontext$, we have \emph{one unique} state $\ftastate'$ (and transitions $\ftatransitions'$) such that these programs $P$ are in $\ftalang(\{ \ftastate' \}, \ftatransitions')$. 
For instance, given $\inputcontext = \{ x \mapsto [ 1, 2, 4] \}$, $\ftastate_1$ from $\fta_s$ may evaluate to $\ftastate_{11}$ in $\fta_e$: all programs in $\ftalang( \{ \ftastate_1 \}, \ftatransitions_s )$ that yield $7$ given $\inputcontext$ are merged into state $\ftastate_{11}$.

\begin{figure}[!t]
\centering
\begin{minipage}[b]{.53\linewidth}
\centering
\includegraphics[height=2.3cm]{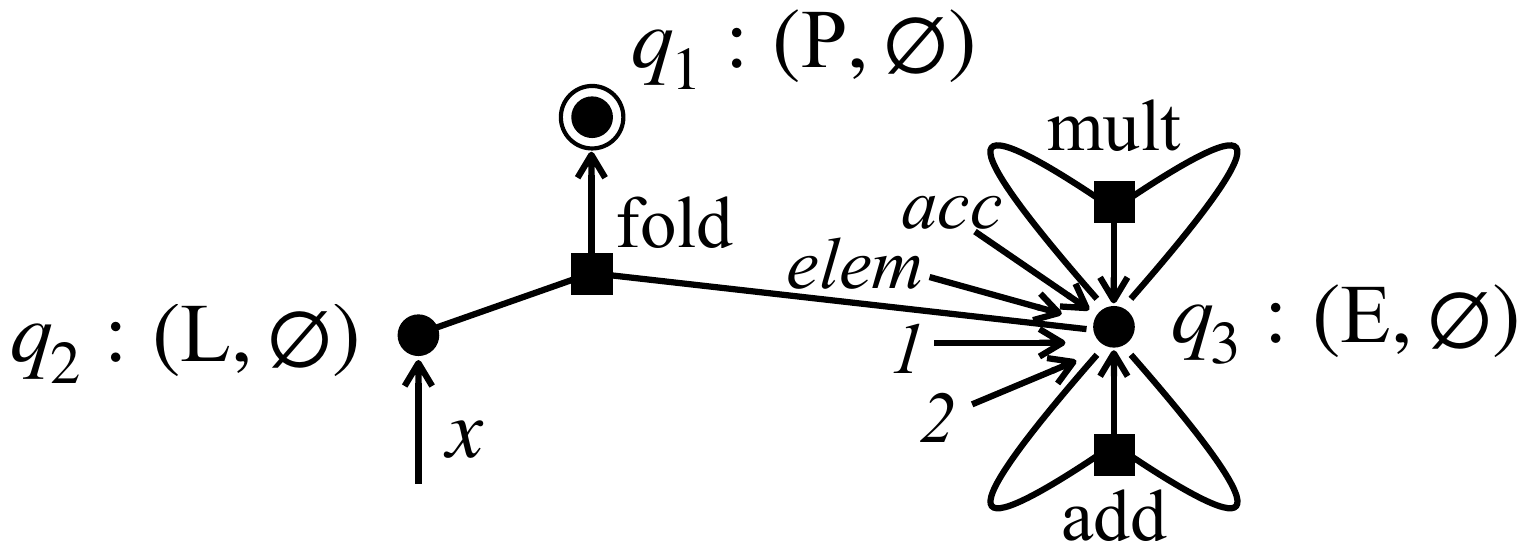}
\vspace{5pt}
\caption{FTA $\fta_s$ constructed for grammar from Figure~\ref{fig:overview:syntax}. Each FTA state is annotated with a grammar symbol and a footprint (which maps a reachable context to a value).}
\label{fig:overview:grammar-fta}
\end{minipage}
\ \ \ \ \ \ \ \ \ 
\begin{minipage}[b]{.4\linewidth}
\centering
\includegraphics[height=3.4cm]{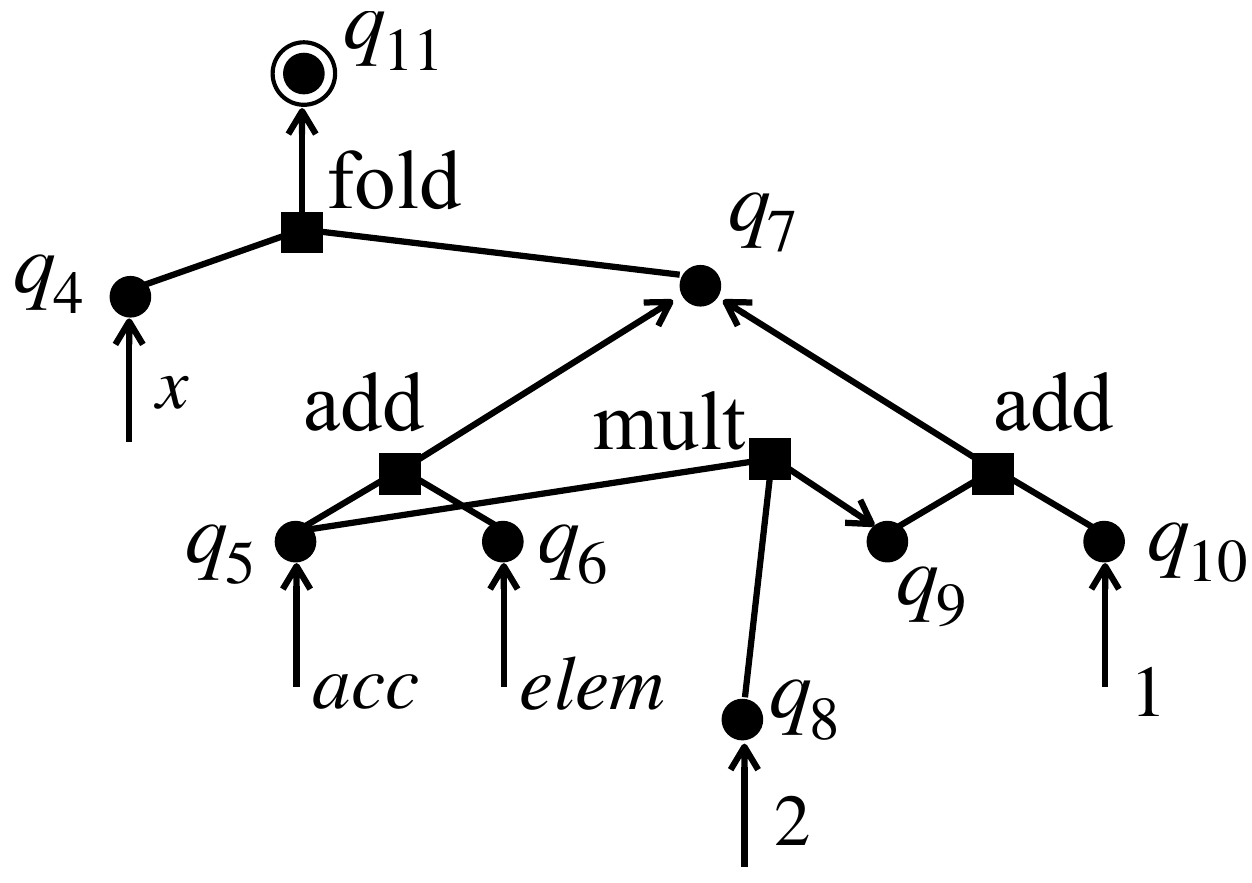}
\vspace{-10pt}
\caption{A part of $\fta_e$ after applying lifted interpretation on $\fta_s$ from Figure~\ref{fig:overview:grammar-fta}. Annotations on FTA states can be found in Figure~\ref{fig:overview:final-fta-annotations}.}
\label{fig:overview:final-fta}
\end{minipage}
\\
\begin{minipage}{.99\linewidth}
\footnotesize 
\centering
\[
\arraycolsep=2pt\def\arraystretch{1}
\begin{array}{ll}

\ftastate_4 :  
\big( \ \ \exlistsymbol, \left\{ 
\begin{array}{l}
\{ x \mapsto [1, 2, 4] \} \mapsto [1, 2, 4]
\end{array}
\right\}  \ \ \big)

\\ \\ 

\ftastate_5 : 
( \exEsymbol, \left\{ 
\begin{array}{l}
\left\{ 
\begin{array}{l}
x \mapsto [1, 2, 4] \\ 
acc \mapsto 0 \\ 
elem \mapsto 1
\end{array}
\right\} 
\mapsto 0 
\\[15pt]
\left\{ 
\begin{array}{l}
x \mapsto [1, 2, 4] \\ 
acc \mapsto 1 \\ 
elem \mapsto 2
\end{array}
\right\} 
\mapsto 1 
\\[15pt]
\left\{ 
\begin{array}{l}
x \mapsto [1, 2, 4] \\ 
acc \mapsto 3 \\ 
elem \mapsto 4
\end{array}
\right\} 
\mapsto 3 
\end{array}
\right\})
\ \  
\ftastate_6 : 
( \exEsymbol, \left\{ 
\begin{array}{l}
\left\{ 
\begin{array}{l}
x \mapsto [1, 2, 4] \\ 
acc \mapsto 0 \\ 
elem \mapsto 1
\end{array}
\right\} 
\mapsto 1 
\\[15pt]
\left\{ 
\begin{array}{l}
x \mapsto [1, 2, 4] \\ 
acc \mapsto 1 \\ 
elem \mapsto 2
\end{array}
\right\} 
\mapsto 2 
\\[15pt]
\left\{ 
\begin{array}{l}
x \mapsto [1, 2, 4] \\ 
acc \mapsto 3 \\ 
elem \mapsto 4
\end{array}
\right\} 
\mapsto 4 
\end{array}
\right\})
\ \ 
\ftastate_7 : 
( \exEsymbol, \left\{ 
\begin{array}{l}
\left\{ 
\begin{array}{l}
x \mapsto [1, 2, 4] \\ 
acc \mapsto 0 \\ 
elem \mapsto 1
\end{array}
\right\} 
\mapsto 1 
\\[15pt]
\left\{ 
\begin{array}{l}
x \mapsto [1, 2, 4] \\ 
acc \mapsto 1 \\ 
elem \mapsto 2
\end{array}
\right\} 
\mapsto 3 
\\[15pt]
\left\{ 
\begin{array}{l}
x \mapsto [1, 2, 4] \\ 
acc \mapsto 3 \\ 
elem \mapsto 4
\end{array}
\right\} 
\mapsto 7 
\end{array}
\right\})

\\ \\ 

\ftastate_8 : 
( \exEsymbol, \left\{ 
\begin{array}{l}
\left\{ 
\begin{array}{l}
x \mapsto [1, 2, 4] \\ 
acc \mapsto 0 \\ 
elem \mapsto 1
\end{array}
\right\} 
\mapsto 2 
\\[15pt]
\left\{ 
\begin{array}{l}
x \mapsto [1, 2, 4] \\ 
acc \mapsto 1 \\ 
elem \mapsto 2
\end{array}
\right\} 
\mapsto 2 
\\[15pt]
\left\{ 
\begin{array}{l}
x \mapsto [1, 2, 4] \\ 
acc \mapsto 3 \\ 
elem \mapsto 4
\end{array}
\right\} 
\mapsto 2  
\end{array}
\right\})
\ \  
\ftastate_9 : 
( \exEsymbol, \left\{ 
\begin{array}{l}
\left\{ 
\begin{array}{l}
x \mapsto [1, 2, 4] \\ 
acc \mapsto 0 \\ 
elem \mapsto 1
\end{array}
\right\} 
\mapsto 0 
\\[15pt]
\left\{ 
\begin{array}{l}
x \mapsto [1, 2, 4] \\ 
acc \mapsto 1 \\ 
elem \mapsto 2
\end{array}
\right\} 
\mapsto 2 
\\[15pt]
\left\{ 
\begin{array}{l}
x \mapsto [1, 2, 4] \\ 
acc \mapsto 3 \\ 
elem \mapsto 4
\end{array}
\right\} 
\mapsto 6 
\end{array}
\right\})
\ \ 
\ftastate_{10} : 
( \exEsymbol, \left\{ 
\begin{array}{l}
\left\{ 
\begin{array}{l}
x \mapsto [1, 2, 4] \\ 
acc \mapsto 0 \\ 
elem \mapsto 1
\end{array}
\right\} 
\mapsto 1 
\\[15pt]
\left\{ 
\begin{array}{l}
x \mapsto [1, 2, 4] \\ 
acc \mapsto 1 \\ 
elem \mapsto 2
\end{array}
\right\} 
\mapsto 1 
\\[15pt]
\left\{ 
\begin{array}{l}
x \mapsto [1, 2, 4] \\ 
acc \mapsto 3 \\ 
elem \mapsto 4
\end{array}
\right\} 
\mapsto 1 
\end{array}
\right\})

\\ \\

\ftastate_{11} :  
\big( \ \  \exprogsymbol, \left\{ 
\begin{array}{l}
\{ x \mapsto [1, 2, 4] \} \mapsto 7
\end{array}
\right\}  \ \ \big)

\end{array}
\]
\vspace{-5pt}
\caption{Annotations (i.e., grammar symbols and footprints) for all states in $\fta_e$ from Figure~\ref{fig:overview:final-fta}.}
\label{fig:overview:final-fta-annotations}
\end{minipage}
\end{figure}

\vspace{3pt}
\newpara{Inference rules.}
The key inference rule that implements this judgment is the \textsc{Transition} rule: 
\[\small 
\centering
\begin{array}{c}
\textsc{(Transition)} 
\ \ 
\irule{
\begin{array}{c}
\ftatransition = \ftatransitionoperator(\ftastate_1, \mydots, \ftastate_n) \ftatransitionarrow \ftastate  \in \ftatransitions \ \ \  \ \ \ 
\evalftatransition{\inputcontext}{\ftatransition; \ftatransitions}{\ftastate', \ftatransitions'}
\end{array}
}{
\evalftastate{\inputcontext}{\ftastate; \ftatransitions}{\ftastate', \ftatransitions'}
}
\end{array}
\]
It says: evaluating a state $\ftastate$ boils down to evaluating each of $\ftastate$'s incoming transitions $\ftatransition$. For instance, to evaluate $\ftastate_1$ in $\fta_s$ (see Figure~\ref{fig:overview:grammar-fta}) under $\inputcontext = \{ x \mapsto [1, 2, 4] \}$, we would evaluate $\foldop(\ftastate_2, \ftastate_3) \ftatransitionarrow \ftastate_1$ under $\inputcontext$. The following inference rules describe how to actually evaluate a $\foldop$ transition. 
\[\small 
\arraycolsep=3pt\def\arraystretch{1.2}
\centering
\begin{array}{lc}
\textsc{(Fold-1)}  & 
\irule{
\begin{array}{c}
\evalftatransition{\inputcontext}{\ftastate_1; \ftatransitions}{\ftastate'_1, \ftatransitions'_1} \ \ \ \ 
\inputcontext \inputtooutput lst \in \getannot(\ftastate'_1) \ \ \ \ 
\evalftatransition{\inputcontext, lst}{\ftastate_2; \ftatransitions}{\ftastate'_2, \ftatransitions'_2} \\ 
v_0 = 0 \ \ \ \ 
\inputcontext \big[ acc \mapsto v_{i-1}, elem \mapsto lst[i] \big] \inputtooutput v_i \in \getannot(\ftastate'_2) \ \ \ \ 
i \in [ 1, |lst| ] \\ 
\annot = \getannot(\ftastate) \ \ \ \ 
\annot' = \annot \cup \{ \inputcontext \inputtooutput v_{|lst|} \} \ \ \ \ 
\ftastate' = \newstate(\exprogsymbol, \annot')
\end{array}
}{
\evalftatransition{\inputcontext}{\foldop(\ftastate_1, \ftastate_2) \ftatransitionarrow \ftastate; \ftatransitions}{\ftastate', \ftatransitions'_1 \cup \ftatransitions'_2 \cup \{ \foldop(\ftastate'_1, \ftastate'_2) \ftatransitionarrow \ftastate' \} }
}
\\  \\ 
\textsc{(Fold-2)}  & 
\irule{
\begin{array}{c}
v_0 = 0 \ \ \ \ 
\ftastate'_0 = \ftastate  \ \ \ \ 
\ftatransitions'_0 = \ftatransitions \ \ \ \ 
i \in [1, |lst| ] \ \ \ \  
\inputcontext_{i-1} = \inputcontext \big[ acc \mapsto v_{i-1}, elem \mapsto lst[i] \big] \\ 
\evalftastate{\inputcontext_{i-1}}{\ftastate'_{i-1}; \ftatransitions'_{i-1}}{\ftastate'_i, \ftatransitions'_i} \ \ \ \ 
\inputcontext_{i-1} \inputtooutput v_i \in \getannot(\ftastate'_i)
\end{array}
}{
\evalftatransition{\inputcontext, lst}{\ftastate; \ftatransitions}{ 
\ftastate'_{|lst|}, \ftatransitions'_{|lst|} }
}
\end{array}
\]
Let us take $\foldop(\ftastate_2, \ftastate_3) \ftatransitionarrow \ftastate_1$ from $\fta_s$ as an example and explain how these two rules work.

The \textsc{Fold-1} rule first evaluates $\ftastate_2$ from $\fta_s$ to $\ftastate_4$ in $\fta_e$, which, as mentioned above, boils down to evaluating $\ftastate_2$'s incoming transition $\inputvar \ftatransitionarrow \ftastate_2$. This is done using the following \textsc{Input-Var} rule.
\[\small 
\arraycolsep=3pt\def\arraystretch{1.2}
\centering
\begin{array}{lc}
\textsc{(Input-Var)}  & 
\irule{
\begin{array}{c}
\inputcontext[x] = lst \ \ \ \ 
\footprint = \getannot(\ftastate)  \ \ \ \ 
\footprint' = \footprint \cup \{ \inputcontext \mapsto lst \} \ \ \ \ 
\ftastate' = \newstate(\exlistsymbol, \footprint')
\end{array}
}{
\evalftastate{\inputcontext}{x \ftatransitionarrow \ftastate; \ftatransitions}{\ftastate', \{ x \ftatransitionarrow \ftastate' \} }
}
\end{array}
\]
Specifically, \textsc{Input-Var} first obtains the value $lst$, which is $[1, 2, 4]$, that $x$ binds to. Then it creates a footprint $\footprint'$ that includes all behaviors from $\ftastate$ and a new binding $\inputcontext\mapsto lst$. Finally, a new state $\ftastate'$ is created, with grammar symbol $\exlistsymbol$ and footprint $\footprint'$. Here, $\newstate$ is simply a state constructor. 
In our example, \textsc{Input-Var} yields $\ftastate_4$, which has footprint $\{ \{ x \mapsto [1, 2, 4] \} \mapsto [1, 2, 4] \}$: it means  all programs in $\ftalang( \{ \ftastate_4\}, \ftatransitions_e )$ produce $[1,2,4]$ given context $\inputcontext = \{ x \mapsto [1, 2, 4] \}$. 

Popping up to \textsc{Fold-1}, given $\ftastate_4$, it retrieves the output $lst$ for $\ftastate_4$ given context $\inputcontext$, notably, by looking up $\ftastate_4$'s footprint. 
It then uses an auxiliary rule, \textsc{Fold-2}, to evaluate $\ftastate_3$ from $\fta_s$ --- which corresponds to lambda bodies --- given $\inputcontext$ and $lst$. This yields $\ftastate_7$ in $\fta_e$, among potentially other states which are not shown in Figure~\ref{fig:overview:final-fta}. 
Again, the guarantee is: all programs in $\ftalang( \{ \ftastate_7 \}, \ftatransitions_e)$ share the same footprint. 
Now given $\ftastate_7$, \textsc{Fold-1} finally creates $\ftastate_{11}$ in $\fta_e$, as well as transition $\foldop(\ftastate_4, \ftastate_7) \ftatransitionarrow \ftastate_{11}$. The key is to compute $\ftastate_{11}$'s footprint $\footprint'$, which inherits everything from $\ftastate_1$'s footprint $\footprint$ but also includes additionally the behavior $\inputcontext \mapsto v_{|lst|} $. Here, $v_{|lst|}$ (which is $v_3$ in our example, as $|lst| = 3$) is the output for the fold operation, under context $\inputcontext=\{ x \mapsto [1, 2, 4] \}$ which is the input example we are concerned with. We note that the computation of $v_{|lst|}$ is based on looking up $\ftastate_7$'s footprint. 

Now let us briefly explain how the \textsc{Fold-2} rule evaluates $\ftastate_3$ to $\ftastate_7$. 
The evaluation is an iterative process that follows the fold semantics. 
It begins with context $\inputcontext_0$ that binds the accumulator $acc$ to the default seed $0$ and binds $elem$ to the first element of $lst$, then recursively evaluates $\ftastate'_0$ (i.e., $\ftastate_3$ in our example) under $\inputcontext_0$ which yields $\ftastate'_1$ (not shown in Figure~\ref{fig:overview:final-fta}), and finally obtains $v_{1}$ (which $acc$ should be bound to in the next iteration) by (again) looking up the footprint of $\ftastate'_1$. 
Note that $\ftastate'_1$ is an intermediate state whose footprint has one behavior, since we have only seen $\inputcontext_0$ so far. 
The second iteration will repeat the same process but for $\ftastate'_1$ and using $\inputcontext_1$, which would yield $\ftastate'_2$ whose footprint has two behaviors. 
This process continues until we reach the end of $lst$, eventually yielding $\ftastate'_{|lst|}$; $\ftastate_7$ in $\fta_e$ is one such state. 
As shown in Figure~\ref{fig:overview:final-fta-annotations}, $\ftastate_7$ has three behaviors in its footprint.

We skip the discussion of the other rules which are used to construct all the other states in $\fta_e$, and refer readers to the appendix 
of our paper for a complete list of rules.
In the end, given $\fta_e$, we will mark states whose footprint satisfies the specification as final states, and extract a program from $\fta_e$. In our example, $\ftastate_{11}$ is final, and both $P_1$ and $P_2$ are in $\ftalang( \{ \ftastate_{11} \}, \ftatransitions_e)$.

}

\revision{\section{Instantiation to Web Automation}\label{sec:alg}

This section presents a full-fledged instantiation of the lifted interpretation idea to web automation. 

\subsection{Web Automation Language}\label{sec:alg:dsl}

Figure~\ref{fig:alg:dsl} presents our web automation language. 
The syntax is slightly different from the one in $\webrobot$~\cite{dong2022webrobot}. 
First, our syntax looks more ``functional'': for example, loop bodies are presented as lambdas. This is solely for the purpose of making it easier to later present our approach. 
Second, the language is also slightly more expressive: $\forselectorloop$ allows starting from the $i$-th child/{descendant} with $i \geq 1$, whereas $\webrobot$'s syntax requires $i=1$. 
This extension is motivated by our observation when curating new benchmarks: many tasks require this more relaxed form of loop. 
We refer interested readers to the $\webrobot$ work for more details about the syntax, but in brief, a web automation program $\program$ is always a sequence of statements. 
It supports different types of statements. 
For example, $\clickexpr$ clicks a DOM element located by selector expression $\selectorexpr$. 
We use an XPath-like syntax for selector expressions: $\selectorexpr\slash\predicate[i]$ gives the $i$-th child of $\selectorexpr$ that satisfies $\predicate$, whereas $\sslash$ considers $\selectorexpr$'s descendants. $\inputvar$ is an input variable that is bound to a user-provided data source (like the list of emails from Example~\ref{ex:intro:example}), whereas $\selectorvar$ and $\datavar$ are \emph{local} variables introduced by and internal to the program. 
$\fordataloop$ is a loopy statement that iterates over a list of data entries (such as emails from Example~\ref{ex:intro:example}) given by $\dataexpr$, binds $\datavar$ to each of them, and executes loop body $\program$. 
$\forselectorloop$ is quite similar, but it loops over a list of selector expressions returned by $\selectorexpr$. 
$\whileloop$ handles pagination, where it repeatedly clicks the ``next page'' button located by $\selectorexpr$ and executes loop body $\program$, until $\selectorexpr$ no longer exists on the webpage.

\begin{figure}
\centering
\begin{minipage}[b]{.46\linewidth}
\small 
\centering
\[
\arraycolsep=3pt
\def\arraystretch{1.1}
\begin{array}{rlll}

\emph{Program} & 
\program & 
\grammareq & 
\seqprog( \statement, \program ) \ | \ 
\skipprog \\ 

\emph{Statement} & 
\statement & 
\grammareq & 
\clickexpr( \selectorexpr ) \\  
& & \ \ \ | & 
\scrapetextexpression( \selectorexpr ) \\ 
& & \ \ \ | & 
\scrapelinkexpression( \selectorexpr ) \\  
& & \ \ \ | & 
\downloadstatement( \selectorexpr )  \\ 
& & \ \ \ | & 
\gobackexpression \ | \ 
\extracturlexpression  \\ 
& & \ \ \ | & 
\sendkeysexpression( \constantstring, \selectorexpr )  \\ 
& & \ \ \ | & 
\enterdataexpression( \dataexpr, \selectorexpr)  \\ 
& & \ \ \ | & 
\fordataloop( \dataexpr, \lambda \datavar. \program )  \\
& & \ \ \ | & 
\forselectorloop( \selectorexpr\slash\predicate[i], \lambda \selectorvar. \program ) \\ 
& & \ \ \ | & 
\forselectorloop( \selectorexpr\sslash\predicate[i], \lambda \selectorvar. \program ) \\ 
& & \ \ \ | & 
\whileloop( \emph{true}, \program, \selectorexpr ) \\

\emph{Selector Expr} & 
\selectorexpr & 
\grammareq & 
\epsilon \ | \ 
\selectorvar \ | \ 
\selectorexpr\slash\predicate[i] \ | \ 
\selectorexpr\sslash\predicate[i]  \\

\emph{Data Expr} & 
\dataexpr & 
\grammareq & 
\inputvar \ | \ 
\datavar \ | \ 
\dataexpr[\datakey] \ | \ 
\dataexpr[i] \\

\emph{Predicate} & 
\predicate & 
\grammareq & 
\domtag \ | \ 
\domtag[@\domattribute=\constantstring] \\

\end{array}
\]
\end{minipage}
\ \ 
\begin{minipage}[b]{.52\linewidth}
\[
\arraycolsep=3pt
\def\arraystretch{1.05}
\small 
\begin{array}{lc}
(\textsc{Seq}) &  
\irule{
\begin{array}{c}
\eval{\domtrace, \env}{\expression}{\actiontrace'_1, \domtrace'_1}  \ \ \ 
\eval{\domtrace'_1, \env}{\program}{\actiontrace'_2, \domtrace'} \\ 
\end{array}
}{
\eval{\domtrace, \env}{\seqprog( \expression, \program )}{\actiontrace'_1 \traceconcat \actiontrace'_2, \domtrace'}
}
\\ \\ 
(\textsc{Click}) & 
\irule{
\begin{array}{c}
\domtrace = [ \dom_1, \mydots, \dom_m ] \ \ \ 
\eval{\dom_1, \env}{\selectorexpr}{\fullxpath} \\ 
\domtrace' = [ \dom_2, \mydots, \dom_m ] 
\end{array}
}{
\eval{\domtrace, \env}{\clickexpr(\selectorexpr)} { [ \clickexpr(\fullxpath) ], \domtrace' }
}
\\ \\ 
(\textsc{ForData-1}) & 
\irule{
\begin{array}{c}
\eval{ \env}{\dataexpr}{\dataexprlist} \\ 
\eval{\domtrace, \env, \dataexprlist}{\program}{\actiontrace', \domtrace'} \\ 
\end{array}
}{
\eval{\domtrace, \env}{\fordataloop( \dataexpr, \lambda \datavar. \program )} { \actiontrace', \domtrace' }
}
\\ \\ 
(\textsc{ForData-2}) & 
\irule{
\begin{array}{c}
\domtrace'_0 = \domtrace \ \ \ 
i \in [ 1, |\dataexprlist| ]  \\ 
\eval{\domtrace'_{i-1}, \env \big[ \datavar \mapsto \dataexprlist[i] \big]}{\program}{\actiontrace'_i, \domtrace'_i} \\ 
\actiontrace' = \actiontrace'_1 \traceconcat \mydots \traceconcat \actiontrace'_{|\dataexprlist|} 
\end{array}
}{
\eval{\domtrace, \env, \dataexprlist}{\fordataloop( \dataexpr, \lambda \datavar. \program )} { \actiontrace', \domtrace'_{|\dataexprlist|} }
}
\end{array}
\]
\end{minipage}
\vspace{-5pt}
\caption{Web automation language. Left is syntax, where $\constantstring$ is a string, $i$ is an integer, $\domtag$ is an HTML tag, and $\domattribute$ is an HTML attribute. Right is a subset of the trace semantics rules; please find the complete set of rules in the appendix.}
\label{fig:alg:dsl}
\end{figure}

Figure~\ref{fig:alg:dsl} also presents a subset of the trace semantics rules, which are cleaner than $\webrobot$'s. 
The evaluation judgment is of the form: 
\[
\eval{\domtrace, \env}{\program}{\actiontrace', \domtrace'}
\]
which reads: given a \emph{context} --- consisting of a DOM trace $\domtrace$ and a binding environment $\env$ (that binds all free variables in scope) --- evaluating program $\program$ yields an action trace $\actiontrace'$ and a DOM trace $\domtrace'$. 
This evaluation does \emph{not} execute $\program$ in browser; instead, it simulates the execution by ``replaying'' $\program$ given $\domtrace$. We refer interested readers to the $\webrobot$ paper~\cite{dong2022webrobot} for the design rationale. 
Here, we briefly explain a few representative rules. 
The \textsc{Seq} rule is standard: it evaluates $\expression$ and $\program$ in sequence, and concatenates the resulting action traces. 
The \textsc{Click} rule is more interesting. 
It first evaluates $\selectorexpr$ (which may use a variable) under $\dom_1$ and $\env$, yielding a selector $\fullxpath$ which is used to form the output action. 
Then, it removes the first DOM from $\domtrace$ to obtain the resulting DOM trace $\domtrace'$. 
In other words, the program under evaluation and the DOM trace are always ``in sync'': the first action to be executed always corresponds to the first DOM. 
The \textsc{ForData} rules are perhaps the most interesting.
\textsc{ForData-1} first evaluates $\dataexpr$ to a list $\dataexprlist$, and then invokes \textsc{ForData-2} which is a helper rule that executes all iterations of the loop until termination. 
The key observation is: similar to the fold function from Example~\ref{ex:overview:example}, $\fordataloop$ also performs \emph{dependent} iterations; that is, an iteration has to be executed under a context computed by its previous iteration. In particular, the input DOM trace carries the dependency. 
This data-dependent feature actually is not specific to $\fordataloop$: all loops in our language are data-dependent, and in fact, $\seqprog$ is too. 
Prior work cannot reduce the space of our web automation programs. 
Our lifted interpretation idea is able to reduce this space, and we will present how it works next.

\subsection{Top-Level Synthesis Algorithm}\label{sec:alg:top-level}

Algorithm~\ref{alg:top-level} shows the top-level algorithm that synthesizes web automation programs from demonstrations. 
It shares the same interface as  $\webrobot$'s algorithm and thus can be directly integrated with $\webrobot$'s front-end UI. 
At a high-level, it takes as input an action trace $\actiontrace$, a DOM trace $\domtrace$, and input data $\inputdata$. 
It returns a program $\program$ that \emph{generalizes} $\actiontrace$ --- i.e., given $\domtrace$ and $\env = \{ \inputvar \mapsto \inputdata\}$, evaluating $\program$ using trace semantics produces $\actiontrace'$ such that $\actiontrace$ is a \emph{strict} prefix of $\actiontrace'$. 
This generalization is possible, because we require $\domtrace$ to have one more DOM than the number of actions in $\actiontrace$. 
To synthesize $\program$, we first find all programs that \emph{reproduce} $\actiontrace$ (i.e., $\actiontrace'$ has $\actiontrace$ as a prefix) --- notably, this process compresses a large number of programs in an FTA $\fta$. 
Then line~\ref{alg:top-level:rank} picks a smallest program $\program$, from $\fta$, that generalizes $\actiontrace$. 
If no such $\program$ exists in $\fta$, the algorithm returns \emph{null}.

\begin{figure}[!t]
\centering
\vspace{-10pt}
\begin{minipage}{.99\linewidth}
\small 
\centering
\begin{algorithm}[H]
\begin{algorithmic}[1]
\myprocedure{Synthesize}{$\actiontrace, \domtrace, \inputdata$}

\Statex\Input{Action trace $\actiontrace = [\action_1, \mydots, \action_m]$, DOM trace $\domtrace = [\dom_1, \mydots, \dom_m, \dom_{m+1}]$, and input data $\inputdata$.}

\Statex\Output{Program $\program$ that generalizes $\actiontrace$, or $\emph{null}$ if no such program can be found.}

\State 
$\fta \assign \textsc{InitFTA}([\action_1, \mydots, \action_m], [\dom_1, \mydots, \dom_m], \inputdata)$ where $\fta = (\ftastates, \ftaalphabet,  \ftastates_f, \ftatransitions)$;
\label{alg:top-level:init}

\While{$\fta$\emph{\_is\_not\_saturated} and \emph{not\_timeout}}
\label{alg:top-level:begin-loop}

\State
$\ftatransitions' \assign \ftatransitions$;

\ForAll{$\ftatransition_1, \mydots, \ftatransition_{2l} \in \ftatransitions'$ 
where $\ftatransition_i = \seqprog(\ftastate'_i, \ftastate_i) \ftatransitionarrow \ftastate_{i-1}, i \in [1, 2l]$}
\label{alg:top-level:matching-transitions}

\Statex\Comment{Find $2l$ ``consecutive'' transitions.}

\State
$\fta_s \assign \textsc{SpeculateFTA}( [ \ftatransition_1, \mydots, \ftatransition_{2l} ], \ftatransitions)$;
\label{alg:top-level:speculate}

\State
$\fta_e \assign \textsc{EvaluateFTA}(\fta_s, \inputcontexts(\ftastate_0))$;
\label{alg:top-level:validate}
\Comment{$\inputcontexts$ gives all contexts in $\ftastate_0$'s footprint.}

\State
$\fta \assign \textsc{MergeFTAs}(\fta, \ftastate_0, \fta_e)$;
\label{alg:top-level:merge}

\EndFor

\EndWhile
\label{alg:top-level:end-loop}

\State
\Return $\textsc{Rank}(\fta, \domtrace, \inputdata)$;
\label{alg:top-level:rank}
\vspace{-5pt}
\caption{Top-level synthesis algorithm.}
\label{alg:top-level}
\end{algorithmic}
\end{algorithm}
\end{minipage}
\end{figure}

Now let us dive into the algorithm a  bit more, though more details will be described in subsequent sections. 
Line~\ref{alg:top-level:init} initializes $\fta$ \emph{based on the input traces}, such that $\fta$ stores all loop-free programs that are guaranteed to reproduce $\actiontrace$. 
The reason that we base our synthesis on the input DOM trace $\domtrace$ is because selector expressions are not given a priori; they are known only when $\domtrace$ becomes available. 
The input action trace $\actiontrace$ is used to further guide synthesis. 
The following example briefly illustrates what an initial FTA looks like; we defer the more detailed explanation to Section~\ref{sec:alg:fta-init}. 

\begin{example}
Consider the following action trace $\actiontrace = [ \action_1, \mydots,  \action_6 ]$ for the task from Example~\ref{ex:intro:example}. 
\[
\arraycolsep=3pt
\def\arraystretch{1.2}
\small 
\begin{array}{lll}

{\color{gray}{\footnotesize\texttt{1}}} 
& \enterdatastatement( x[0], /html/.../div[1]/input ) 
& {\color{gray}{\hspace{0pt}\textsf{\# enter pwned@gmail.com}}}
\\ 

{\color{gray}{\footnotesize\texttt{2}}} 
& \clickstatement( /html/body/.../div[1]/span/button ) \ \ 
& {\color{gray}{\hspace{0pt}\textsf{\# click search button}}}
\\ 

{\color{gray}{\footnotesize\texttt{3}}} 
& \scrapetextstatement( /html/body/div[1]/.../div/h2 ) 
& {\color{gray}{\hspace{0pt}\textsf{\# scrape ``Oh no -- pwned!''}}}
\\ 

{\color{gray}{\footnotesize\texttt{4}}} 
& \enterdatastatement( x[1], /html/.../div[1]/input ) 
& {\color{gray}{\hspace{0pt}\textsf{\# enter not\_pwned@gmail.com}}}
\\ 

{\color{gray}{\footnotesize\texttt{5}}} 
& \clickstatement( /html/body/.../div[1]/span/button ) 
& {\color{gray}{\hspace{0pt}\textsf{\# click search button}}}
\\ 

{\color{gray}{\footnotesize\texttt{6}}} 
& \scrapetextstatement( /html/body/div[0]/.../div/h2 ) 
& {\color{gray}{\hspace{0pt}\textsf{\# scrape ``Good news - no pwnage found''}}}
\end{array}
\]
In addition, the demonstration also includes a DOM trace $\domtrace = [ \dom_1, \mydots, \dom_7]$, where $\action_i$ is performed on DOM $\dom_i$. 
Suppose $\inputdata = [ \text{``\texttt{pwned@gmail.com}''}, \text{``\texttt{not\_pwned@gmail.com}''} ]$ is the list of emails. 

Figure~\ref{fig:alg:action-trace} shows the abstract syntax tree for $\actiontrace$, and Figure~\ref{fig:alg:init-fta} gives the corresponding initial FTA $\fta$. 
Here, we use $\fullxpath_i$ as a shorthand to denote the selector in $\action_i$. 
Note that $\actiontrace$ and $\fta$ in  both figures are pretty much ``isomorphic'', except that each action in $\fta$ has multiple selectors as ``leaf transitions'' (see each dashed circle), whereas each action in $\actiontrace$ has only one selector. Section~\ref{sec:alg:fta-init} will present more details around how $\fta$ is constructed. 
\label{ex:alg:action-trace}
\end{example}

\begin{figure}
\centering
\begin{minipage}[b]{.43\linewidth}
\centering
\includegraphics[height=4.2cm]{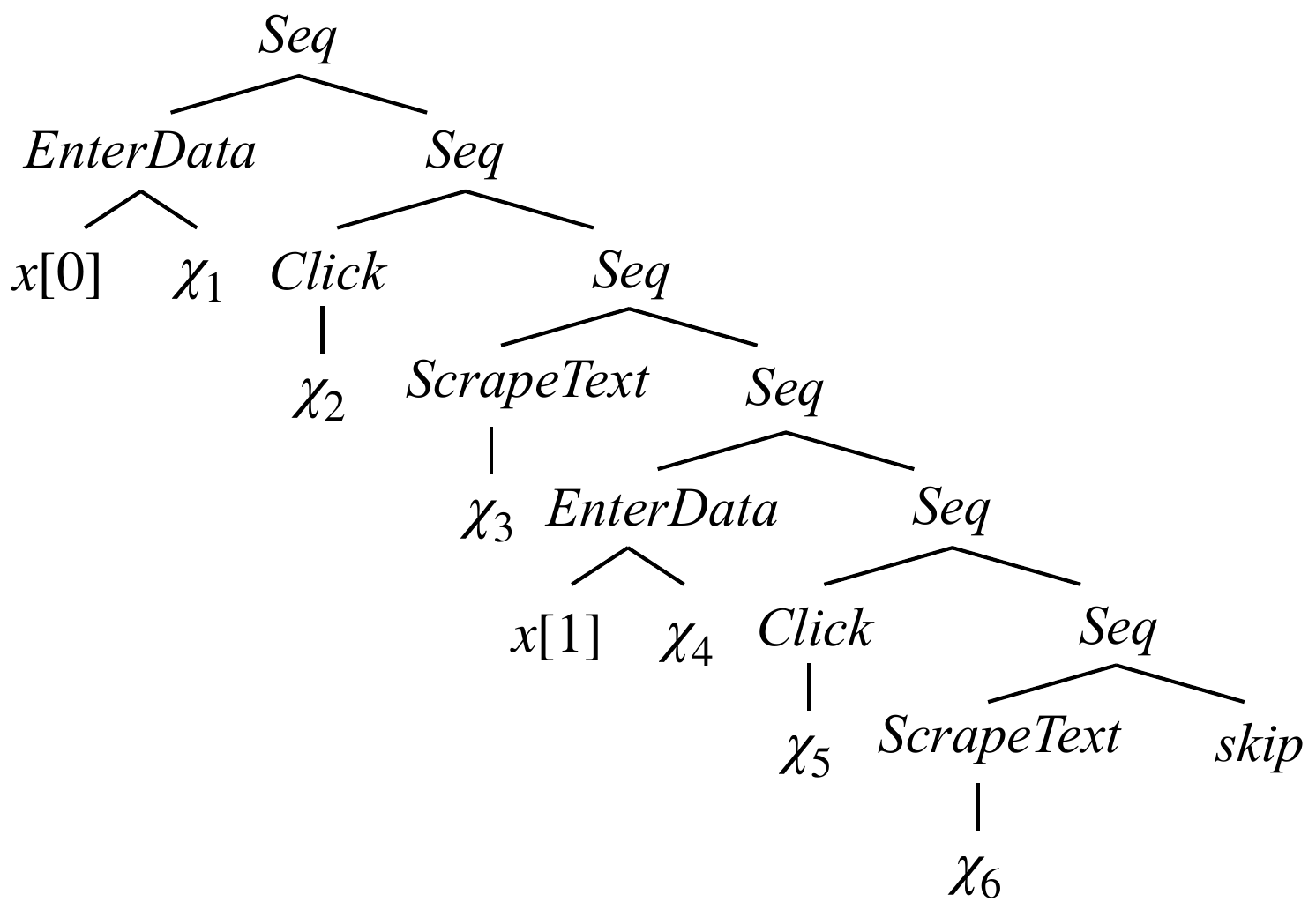}
\caption{AST for action trace from Example~\ref{ex:alg:action-trace}.}
\label{fig:alg:action-trace}
\end{minipage}
\begin{minipage}[b]{.55\linewidth}
\centering
\includegraphics[height=4.6cm]{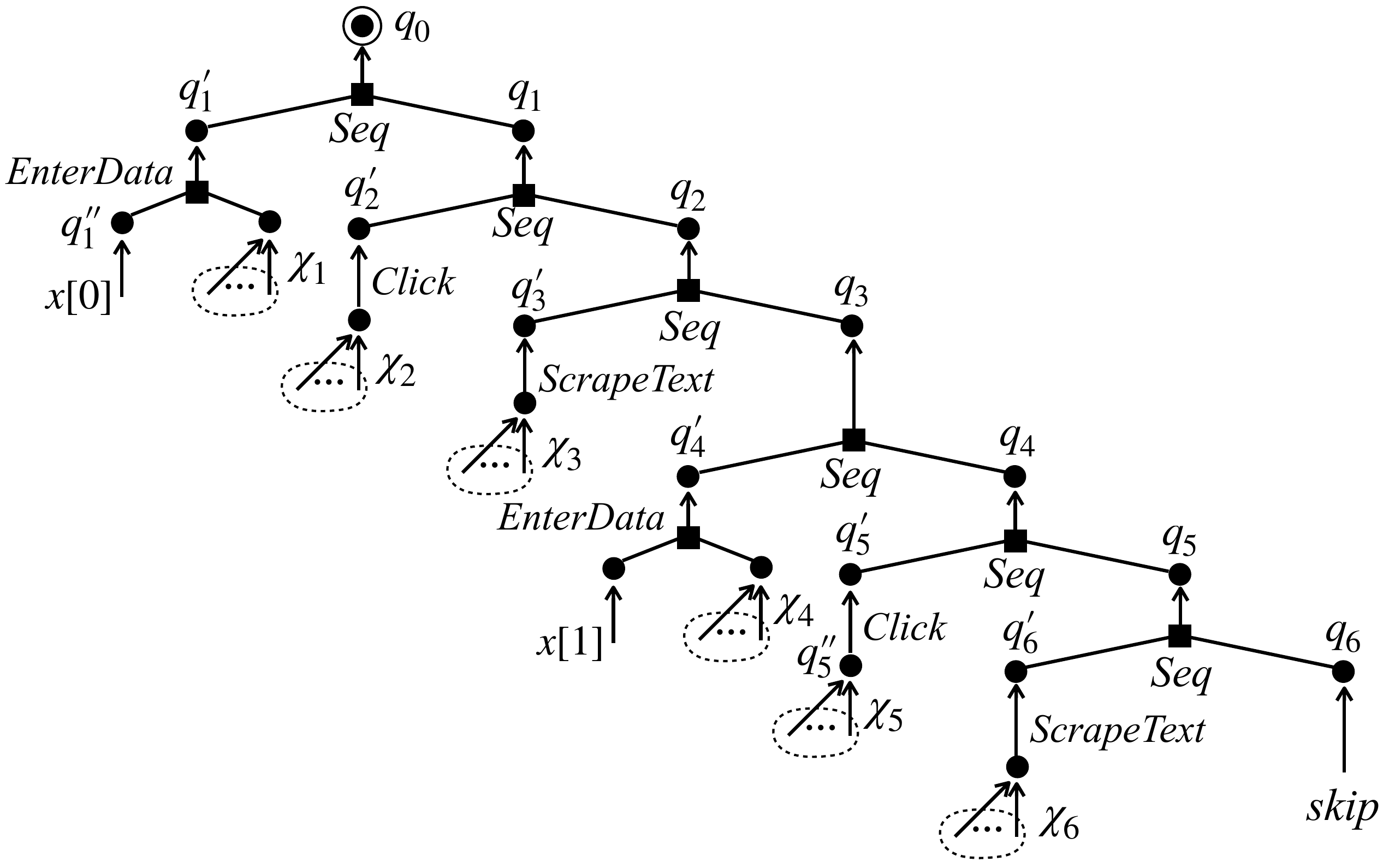}
\caption{Initial FTA for Figure~\ref{fig:alg:action-trace}.}
\label{fig:alg:init-fta}
\end{minipage}
\end{figure}

Given the initial $\fta$, we then enter a loop (lines~\ref{alg:top-level:begin-loop}-\ref{alg:top-level:end-loop}) which iteratively adds to $\fta$ \emph{loopy} programs, until no new programs can be found (i.e., $\fta$ saturates) or a predefined timeout is reached. 
In each iteration, we non-deterministically pick $2l$ \emph{consecutive} $\seqprog$ transitions $\ftatransition_1, \mydots, \ftatransition_{2l}$ (line~\ref{alg:top-level:matching-transitions}), and then perform three key steps: 
(1) \textsc{SpeculateFTA} guesses an FTA $\fta_s$ based on these $2l$ transitions, 
(2) \textsc{EvaluateFTA} performs lifted interpretation over $\fta_s$, yielding another FTA $\fta_e$, 
and (3) \textsc{MergeFTAs} merges $\fta_e$ into $\fta$. 
The resulting $\fta$ at line~\ref{alg:top-level:merge}, compared to $\fta$ before the merge, includes new loopy programs that are synthesized during this iteration, with the same guarantee that all programs in $\fta$ still reproduce $\actiontrace$. 
The generalization step takes place in (1), where we reroll a slice of statements in a program from $\fta$ to a loop that is then stored in $\fta_s$. 
This loop rerolling step, however, is speculative, meaning some rerolled loops may not be correct. 
Therefore, in step (2), we use \textsc{EvaluateFTA} to check all loops and retain only those that can indeed reproduce $\actiontrace$ --- this \textsc{EvaluateFTA} algorithm (i.e., lifted interpretation) is the key contribution of this paper. 

In what follows, we explain how each step works in more detail. 
In particular, we will begin with \textsc{EvaluateFTA} in Section~\ref{sec:alg:fta-lifted-interp}, given $\fta_s$ and its corresponding input contexts. 
Then in subsequent sections, we explain how FTA initialization, speculation, merging, and ranking work, respectively.

\subsection{Lifted Interpretation}\label{sec:alg:fta-lifted-interp}

\begin{figure}
\centering
\begin{minipage}{.9\linewidth}
\centering
\includegraphics[height=3cm]{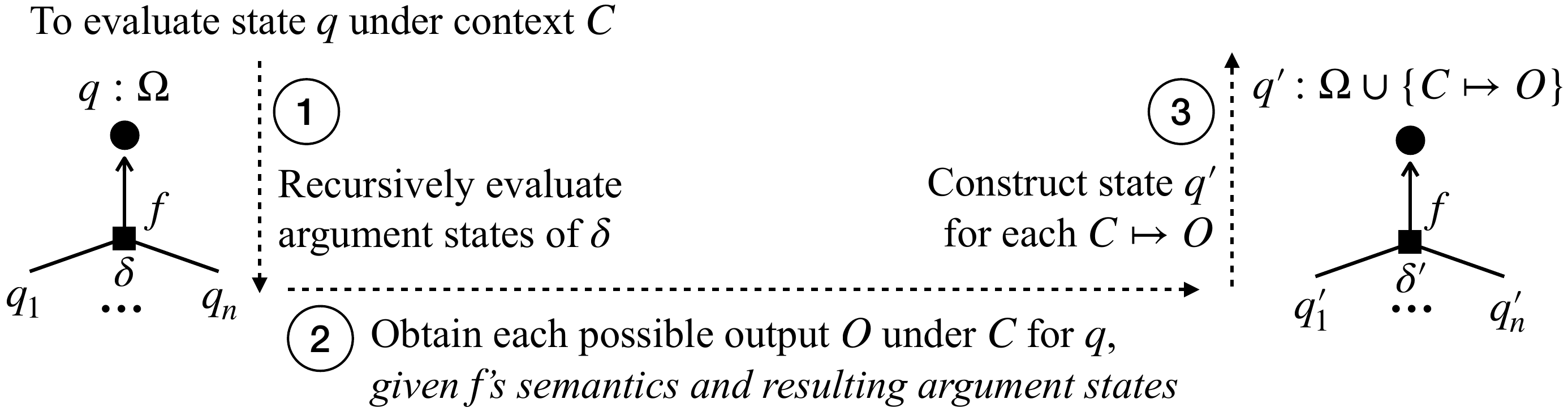}
\caption{Illustration of \textsc{EvaluateFTA} rules from Figure~\ref{fig:evaluation-rules}.}
\label{fig:illustration-evaluation-rules}
\end{minipage}
\end{figure}

As mentioned in Section~\ref{sec:overview}, our lifted interpretation uses the following key judgment. 
\[
\begin{array}{c}
\evalftastate{
\inputcontext
}{
\ftastate; \ftatransitions
}{
\ftastate', \ftatransitions'
}
\end{array}
\]

In our domain, a context $\inputcontext$ consists of a DOM trace $\domtrace$ and a binding environment $\env$. 
Figure~\ref{fig:evaluation-rules} shows the \textsc{EvaluateFTA} rules that implement lifted interpretation. 
We suggest readers looking at Figure~\ref{fig:illustration-evaluation-rules} at the same time. 
Rule (1) reduces multi-context evaluation to single-context evaluation. We note that the evaluation under $\inputcontext_i$ relies on the previous evaluation result for $\inputcontext_{i-1}$. 
To evaluate a state $\ftastate$ under a single context $\inputcontext$, Rule (2) further reduces it to evaluating each of $\ftastate$'s incoming transitions. 
The remaining rules evaluate various types of transitions; they share the same principle as illustrated in Figure~\ref{fig:illustration-evaluation-rules}. 
In particular, given transition $\ftatransition = \ftatransitionoperator(\ftastate_1, \mydots, \ftastate_n) \ftatransitionarrow \ftastate$ and context $\inputcontext$: 
\begin{enumerate}[leftmargin=*]
\item 
First, we recursively evaluate each argument state $\ftastate_i$. 
The input context under which to evaluate each $\ftastate_i$ and the order to evaluate them depends on the semantics of $\ftatransitionoperator$. 
\item 
Given each resulting state $\ftastate'_i$ for $\ftastate_i$, we obtain the output $\outputval$ for $\ftastate$ for $\inputcontext$. Note that $\outputval$ is computed \emph{compositionally}, from the corresponding outputs associated with $\ftastate'_i$ and per $\ftatransitionoperator$'s semantics. 
\item 
Finally, we construct state $\ftastate'$ that $\ftastate$ evaluates to. In particular, $\ftastate'$ includes $\inputcontext \mapsto \outputval$ as a behavior in its footprint $\footprint$, as well as all behaviors from $\ftastate$'s footprint $\footprint$. 
\end{enumerate}
Let us examine the rules in detail. 
Rule (3) is a base case for a nullary transition $\ftatransition$ with no argument states: it directly evaluates the selector expression $\selectorexpr$ and yields a state $\ftastate'$ with behavior $\domtrace, \env \inputcontexttooutput \fullxpath$. 
Rule (4) is more interesting. 
We first evaluate $\ftastate_1$ to $\ftastate'_1$ under context $\domtrace, \env$: note that here we use the context for $\ftastate$ to evaluate $\ftastate'_1$, due to $\clickstatement$'s semantics. 
Then, given $\ftastate'_1$, we obtain its output $\fullxpath$ which is later used to construct the output trace $\actiontrace'$ for $\ftastate$. 
Finally, we construct $\ftastate'$ with footprint $\annot'$, which includes the new behavior we just created based on $\ftastate'_1$ and $\clickstatement$'s semantics, as well as everything from $\annot$. 
Rule (5) considers a $\seqprog$ transition with two argument states. 
We first evaluate $\ftastate_1$ to $\ftastate'_1$. 
However, before evaluating $\ftastate_2$, we need to obtain $\domtrace'_1$ from $\ftastate'_1$ to form the context to evaluate $\ftastate_2$. 
The output action traces $\actiontrace'_1, \actiontrace'_2$ for $\ftastate'_1, \ftastate'_2$ form the output action trace for $\ftastate$. Finally, we construct $\ftastate'$ with $\annot'$, same as previous rules. 
Rule (6) is another base case for $\skipprog$. 
Rule (7) {also concerns} a base case: if {the} input DOM trace is empty, it yields the sub-FTA rooted at $\ftastate$.
\finalized{
Note that in general, $\ftastate$ would also include $\emptytrace, \env \inputcontexttooutput \emptytrace, \emptytrace$ in its footprint.
Rule (7) does not show this explicitly, because we assume all states implicitly have $\emptytrace, \_ \inputcontexttooutput \emptytrace, \emptytrace$. }

Now let us look at the most exciting rules for loops. 
Consider $\fordataloop$: its first argument is a data expression {that} yields a list $\dataexprlist$, and the second argument (i.e., loop body) is evaluated with the loop variable $\datavar$ being bound to each element from $\dataexprlist$. 
To evaluate $\ftastate$ with an incoming $\fordataloop$ transition, Rule (8) first evaluates $\ftastate_1$ to $\ftastate'_1$ and obtains $\dataexprlist$ from $\ftastate'_1$. Then, we evaluate ``loop body'' state $\ftastate_2$ under a context with $\dataexprlist$: Rule (9) presents how this evaluation works. Intuitively, Rule (9) has a series of $n$ evaluations: $\ftastate$ to $\ftastate'_1$, $\ftastate'_1$ to $\ftastate'_2$, $\mydots$, $\ftastate'_{n-1}$ to $\ftastate'_n$. Note that the output DOM trace for $\ftastate'_i$ is part of the context for evaluating $\ftastate'_{i+1}$, due to the data dependency across iterations. The output state is $\ftastate'_n$, which encapsulates information from all $n$ iterations. 
Popping up back to Rule (8): given $\ftastate'_2$, we obtain the output traces $\actiontrace'_i$ for all iterations, which are then concatenated to form the output trace in $\ftastate'$. 
We skip the discussion on other loop types as they are very similar to $\fordataloop$. 

\begin{figure}
\centering
\begin{minipage}{.99\linewidth}
\centering
\arraycolsep=2pt\def\arraystretch{1}
\begin{centermath}
\small 
\centering
\begin{array}{l}

(1) \ 
\irule{
\begin{array}{c}
\ftastate'_0 = \ftastate \ \ \ 
\ftatransitions'_0 = \ftatransitions 
\ \ \  \ 
\evalftastate{\inputcontext_i}
{\ftastate'_{i-1}; \ftatransitions'_{i-1}}
{\ftastate'_i, \ftatransitions'_i}  
\end{array}
}{
\evalftastate{\inputcontext_1, \mydots, \inputcontext_n}{\ftastate; \ftatransitions}{\ftastate'_n, \ftatransitions'_n}
}

\ \ 

(2) \ 
\irule{
\begin{array}{c}
\ftatransition = \ftatransitionoperator(\ftastate_1, \mydots, \ftastate_n) \ftatransitionarrow \ftastate  \in \ftatransitions \ \ \ 
\evalftatransition{\inputcontext}{\ftatransition; \ftatransitions}{\ftastate', \ftatransitions'}
\end{array}
}{
\evalftastate{\inputcontext}{\ftastate; \ftatransitions}{\ftastate', \ftatransitions'}
}

\\ \\ 

(3) \ 
\irule{
\begin{array}{c}
\domtrace = [\dom_1, \mydots, \dom_m] \ \ \ \ 
\eval{\env, \dom_1}{\selectorexpr}{\fullxpath} \ \ \ \ 
\annot = \getannot(\ftastate) \ \ \ \ 
\annot' = \annot \cup \{ \domtrace, \env \inputcontexttooutput \fullxpath \} \ \ \ \ 
\ftastate' = \newstate(\selectorexprsymbol, \annot')
\end{array}
}{
\evalftatransition{\domtrace, \env}{ \selectorexpr \ftatransitionarrow \ftastate; \ftatransitions}{\ftastate', \{ \selectorexpr \ftatransitionarrow \ftastate' \}}
}

\\ \\ 

(4) \  
\irule{
\begin{array}{c}
\evalftastate{\domtrace, \env}{\ftastate_1; \ftatransitions}{\ftastate'_1, \ftatransitions'_1} \ \ \ \ 
\domtrace, \env \inputcontexttooutput \fullxpath  \in \getannot(\ftastate'_1) \ \ \ \ 
\annot = \getannot(\ftastate) \ \ \ \ 
\domtrace = [\dom_1, \mydots, \dom_m] \\ 
\actiontrace' = [ \clickstatement(\fullxpath) ]
\ \ \ \ 
\annot' = \annot \cup \{ \domtrace, \env \inputcontexttooutput \actiontrace', [\dom_2, \mydots, \dom_m] \} \ \ \ \ 
\ftastate' = \newstate(\expressionsymbol, \annot')
\end{array}
}{
\evalftatransition{\domtrace, \env}{ \clickstatement(\ftastate_1) \ftatransitionarrow \ftastate; \ftatransitions}{\ftastate', \ftatransitions'_1 \cup \{ \clickstatement(\ftastate'_1) \ftatransitionarrow \ftastate' \} }
}

\\ \\ 

(5) \ 
\irule{
\begin{array}{c}
\evalftastate{\domtrace, \env}{\ftastate_1; \ftatransitions}{\ftastate'_1, \ftatransitions'_1} \ \ \ \ 
\domtrace, \env \inputcontexttooutput \actiontrace'_1, \domtrace'_1  \in \getannot(\ftastate'_1) \\ 
\evalftastate{\domtrace'_1, \env}{\ftastate_2; \ftatransitions}{\ftastate'_2, \ftatransitions'_2} \ \ \ \ 
\domtrace'_1, \env \inputcontexttooutput \actiontrace'_2, \domtrace'_2  \in \getannot(\ftastate'_2) \\ 
\annot = \getannot(\ftastate) \ \ \ 
\annot' = \annot \cup \{ \domtrace, \env \inputcontexttooutput \actiontrace'_1 \traceconcat \actiontrace'_2, \domtrace'_2 \} \ \ \ 
\ftastate' = \newstate(\programsymbol, \annot')
\end{array}
}{
\evalftatransition{\domtrace, \env}{\seqprog(\ftastate_1, \ftastate_2) \ftatransitionarrow \ftastate; \ftatransitions}{\ftastate', \ftatransitions'_1 \cup \ftatransitions'_2 \cup \{ \seqprog(\ftastate'_1, \ftastate'_2) \ftatransitionarrow \ftastate' \}}
}

\\ \\ 

(6) \ 
\irule{
\begin{array}{c}
\annot = \getannot(\ftastate) \ \ \ \ 
\annot' = \annot \cup \{ \domtrace, \env \inputcontexttooutput \emptytrace, \domtrace \} \ \ \ \ 
\ftastate' = \newstate(\programsymbol, \annot)
\end{array}
}{
\evalftatransition{\domtrace, \env}{\skipprog \ftatransitionarrow \ftastate; \ftatransitions}{\ftastate', \{ \skipprog \ftatransitionarrow \ftastate' \}}
}

\ \

(7) \ 
\irule{
\begin{array}{c}
\ \subfta(\ftastate, \ftatransitions) = 
 \big( \{ \ftastate \}, \ftatransitions' \big) \\
\end{array}
}{
\evalftatransition{\emptytrace, \env}{\ftastate; \ftatransitions}{\ftastate, \ftatransitions'}
}

\\ \\ 

(8) \ 
\irule{
\begin{array}{c}
\evalftastate{\domtrace, \env}{\ftastate_1; \ftatransitions}{ \ftastate'_1, \ftatransitions'_1 } \ \ \ \
\domtrace, \env \inputcontexttooutput \dataexprlist \in \getannot(\ftastate'_1)
\\ 
\evalftastate{\domtrace, \env, \dataexprlist}{\ftastate_2; \ftatransitions}{\ftastate'_2, \ftatransitions'_2} 
\ \ \ \ 
\domtrace, \env \big[ \datavar \bindsto L[i] \big] \inputcontexttooutput \actiontrace'_{i}, \domtrace'_{i} \in \getannot(\ftastate'_2)   
\ \ \ \ 
i \in [1, |\dataexprlist|]
\\ 
\annot = \getannot(\ftastate) 
\ \ \ \ 
\annot' = \annot \cup \{ \domtrace, \env \inputcontexttooutput \actiontrace'_1 \traceconcat \mydots \traceconcat \actiontrace'_{|\dataexprlist|},  \domtrace'_{|\dataexprlist|} \} 
\ \ \ \ 
\ftastate' = \newstate(\expressionsymbol, \annot')
\end{array}
}{
\evalftatransition{\domtrace, \env}{\fordataloop(\ftastate_1, \ftastate_2) \ftatransitionarrow \ftastate; \ftatransitions}{\ftastate', \ftatransitions'_1 \cup \ftatransitions'_2 \cup \{  \fordataloop(\ftastate'_1, \ftastate'_2) \ftatransitionarrow \ftastate'  \}}
}

\\ \\ 

(9) \ 
\irule{
\begin{array}{c}
\ftastate'_0 = \ftastate \ \ \ \ 
\ftatransitions'_0 = \ftatransitions \ \ \ \ 
\domtrace'_0 = \domtrace \\ 
\evalftastate{\domtrace'_{i-1}, \env[\datavar \bindsto \datavalue_i]}{\ftastate'_{i-1}; \ftatransitions'_{i-1}}{\ftastate'_i, \ftatransitions'_i} \ \ \ \ 
\domtrace'_{i-1}, \env[\datavar \bindsto \datavalue_i] \inputcontexttooutput \actiontrace'_i, \domtrace'_i \in \getannot(\ftastate'_i) 
\end{array}
}{
\evalftastate{\domtrace, \env, [\datavalue_1, \mydots, \datavalue_n]}{\ftastate; \ftatransitions}{\ftastate'_n, \ftatransitions'_n}
}

\end{array}
\end{centermath}
\caption{A subset of rules for $\textsc{EvaluateFTA}$; the complete set can be found in the appendix.}
\label{fig:evaluation-rules}
\end{minipage}
\end{figure}

\begin{example}
Consider the task from Example~\ref{ex:alg:action-trace}. 
Suppose $\fta_s = ( \{ \ftastatep \}, \ftatransitions)$ shown in Figure~\ref{fig:alg:speculate-fta} is the FTA speculated from the initial FTA in Figure~\ref{fig:alg:init-fta}. 
Section~\ref{sec:alg:fta-speculation} will later explain how $\textsc{SpeculateFTA}$ generates this $\fta_s$, but in brief, it contains $\fordataloop$ loops inferred from the initial FTA. 
Each state in $\fta_s$ is annotated with a grammar symbol and an \emph{empty} footprint (see a few examples in Figure~\ref{fig:alg:speculate-fta}). 

Given a context consisting of a DOM trace $[\dom_1, \mydots, \dom_6]$ and a binding environment $\env = \{ \inputvar \mapsto \inputdata \}$ (both of which are from Example~\ref{ex:alg:action-trace}), \textsc{EvaluateFTA} applies lifted interpretation to $\fta_s$: the process is very similar to that in Example~\ref{ex:overview:example}, except that now we use trace semantics for a different syntax. 
Figure~\ref{fig:alg:validate-fta} illustrates a part of the FTA $\fta_e = (\ftastates'_{f}, \ftatransitions')$ returned by \textsc{EvaluateFTA}. 
Here, we show two states $\ftastater_1 \in \ftastates'_f$ and $\ftastater_2 \in \ftastates'_f$, with different footprints, that $\ftastatep$ evaluates to: 
\[\small 
\begin{array}{c}
\evalftastate{
[ \dom_1, \mydots, \dom_6 ], 
\env
}{
\ftastatep; \ftatransitions
}{
\ftastater_1, \ftatransitions'_1
}
\ \ 
\emph{where} 
\ \ 
\ftastater_1 = 
\Big(
\programsymbol, 
\big\{  
[ \dom_1, \mydots, \dom_6 ], 
\env
\mapsto 
[ \action_1, \action_2, \action_3, \action_4, \action_5, \action_6  ]
\big\}
\Big)
\\[5pt]
\evalftastate{
[ \dom_1, \mydots, \dom_6 ], 
\env
}{
\ftastatep; \ftatransitions
}{
\ftastater_2, \ftatransitions'_2
}
\ \ 
\emph{where} 
\ \ 
\ftastater_2 = 
\Big(
\programsymbol, 
\big\{  
[ \dom_1, \mydots, \dom_6 ], 
\env
\mapsto 
[ \action_1, \action_2, \action'_3, \action_4, \action_5, \action_6  ]
\big\}
\Big)
\end{array}
\]
Here, $[\action_1, \mydots, \action_6]$ is the desired action trace (see Example~\ref{ex:alg:action-trace}) where $\action_3$ scrapes ``Oh no -- pwned'' and $\action_6$ scrapes ``Good news -- no pwnage found''. 
The action $\action'_3$, however, scrapes ``Good news -- no pwnage found'', which is undesired. 
The reason we can have $\action'_3$ is because programs in $\ftalang( \{ \ftastater_2 \}, \ftatransitions' )$ use an undesired selector expression (such as the full selector expression described in Example~\ref{ex:intro:example}) that does not generalize. 
More specifically, evaluating $\ftastatep'_3$ from $\fta_s$ (with a full selector expression in $\scrapetextstatement$) can yield state $\ftastater_6$ in $\fta_e$: 
\[\small 
\ftastater_6 : 
( \statementsymbol, 
\left\{
\begin{array}{lll}
[ \dom_3, \mydots, \dom_6 ], 
\{ \inputvar \mapsto \inputdata, \datavar \mapsto ``\texttt{pwned@gmail.com}" \} 
& \inputcontexttooutput 
& [\action'_3], [\dom_4, \dom_5, \dom_6] 
\\[3pt]
[ \dom_6 ], 
\{ \inputvar \mapsto \inputdata, \datavar \mapsto ``\texttt{not\_pwned@gmail.com}" \} 
& \inputcontexttooutput & [\action_6], \emptytrace
\end{array}
\right\}
)
\]
That is, for both emails, the full selector always scrapes ``Good news -- no pwnage found''. 
This in turn means that $\ftastatep_0$ from $\fta_s$ can evaluate to $\ftastater_4$ in $\fta_e$: 
\[\small 
\ftastater_4 : 
( \programsymbol, 
\left\{
\begin{array}{lll}
[ \dom_1, \mydots, \dom_6 ], 
\{ \inputvar \mapsto \inputdata, \datavar \mapsto ``\texttt{pwned@gmail.com}" \} 
& \inputcontexttooutput 
& [\action_1, \action_2, \action'_3], [\dom_4, \dom_5, \dom_6] 
\\[3pt]
[ \dom_4, \mydots, \dom_6 ], 
\{ \inputvar \mapsto \inputdata, \datavar \mapsto ``\texttt{not\_pwned@gmail.com}" \} 
& \inputcontexttooutput 
& [ \action_4, \action_5, \action_6 ], \emptytrace
\end{array}
\right\}
)
\]
Finally, this leads to the aforementioned state $\ftastater_2$. 
\label{ex:alg:evaluate-fta}
\end{example}

\begin{figure}[!t]
\centering
\begin{minipage}[b]{.47\linewidth}
\centering
\includegraphics[height=4.3cm]{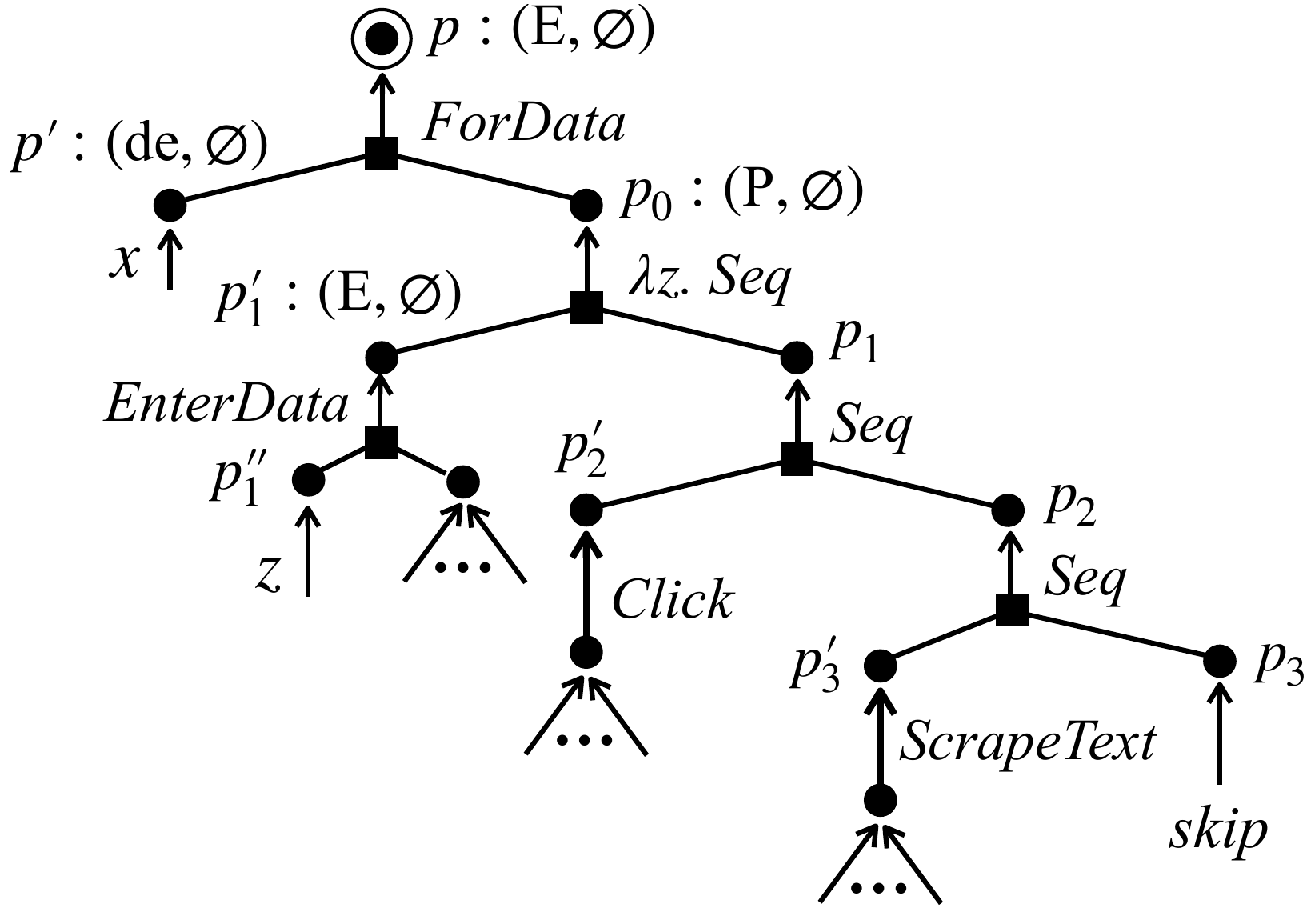}
\caption{Illustration of one potential $\fta_s$ given by \textsc{SpeculateFTA} for the initial FTA from Figure~\ref{fig:alg:init-fta}.}
\label{fig:alg:speculate-fta}
\end{minipage}
\ \ \ \ \ \ 
\begin{minipage}[b]{.47\linewidth}
\centering
\includegraphics[height=4.2cm]{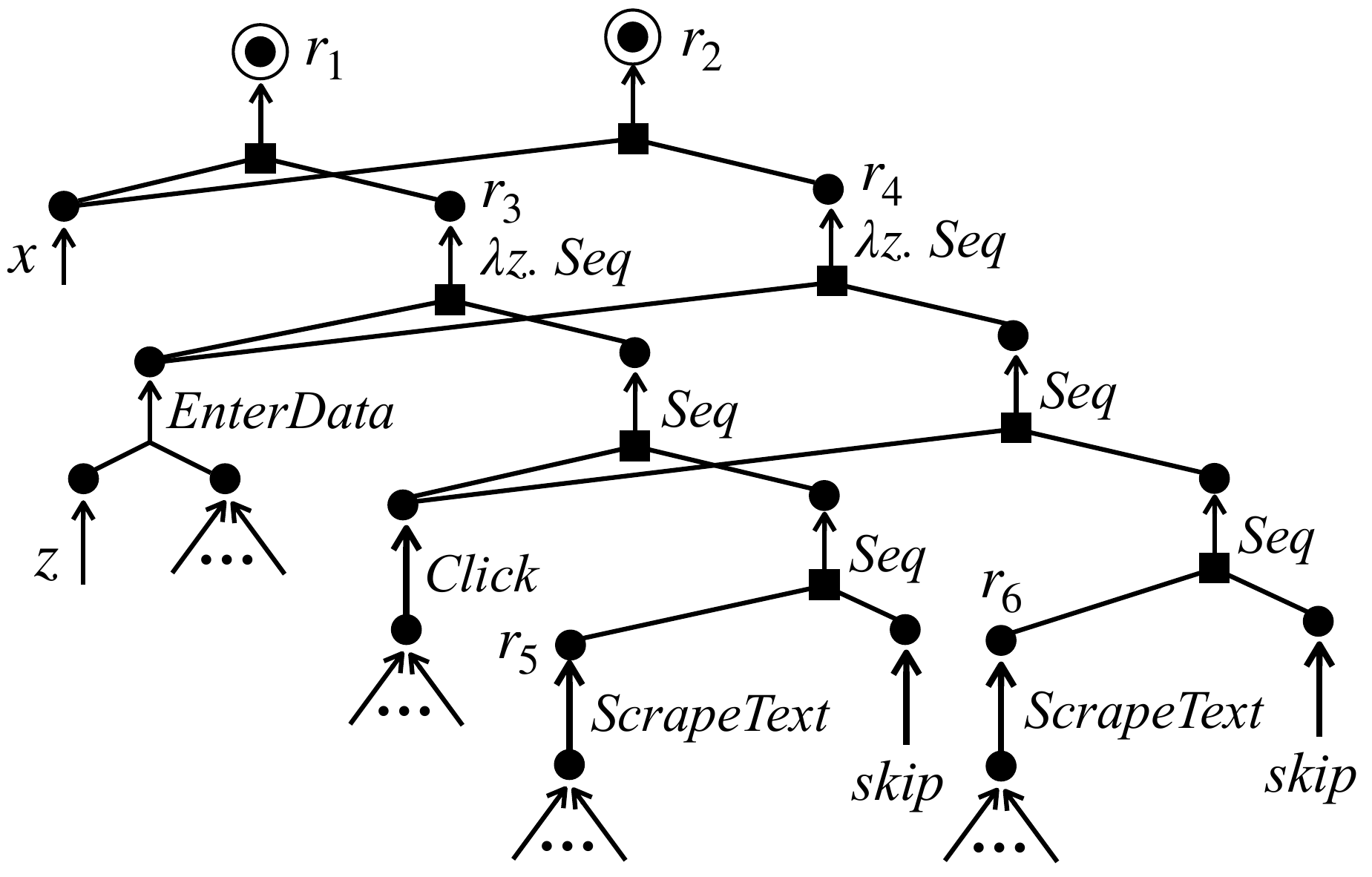}
\caption{Illustration of a part of $\fta_e$ returned by \textsc{EvaluateFTA} for $\fta_s$ from Figure~\ref{fig:alg:speculate-fta}.}
\label{fig:alg:validate-fta}
\end{minipage}
\hspace{8pt}
\end{figure}

\subsection{FTA Initialization}\label{sec:alg:fta-init}

Now let us circle back and describe the FTA initialization procedure, which we note is specific to PBD and web automation. 
Figure~\ref{fig:fta-init-rules} shows the key initialization rule that constructs, from action trace $[\action_1, \mydots, \action_m]$ and DOM trace $[\dom_1, \mydots, \dom_m]$, $m$ consecutive $\seqprog$ transitions in the initial FTA $\fta$. 
In particular, all $\ftastate_i$ ($i \in [0, m]$) states have grammar symbol $\programsymbol$, and all $\ftastate'_i$ ($i \in [1, m]$) states have $\expressionsymbol$. 
Each $\seqprog$ transition $\ftatransition_i$ ($i \in [1, m]$) connects $\ftastate'_i$ and $\ftastate_i$ to $\ftastate_{i-1}$, forming $m$ consecutive $\seqprog$ transitions. 
The footprint in $\ftastate_i$ means that all programs in $\subfta(\ftastate_i, \fta)$ yield $[\action_{i+1}, \mydots, \action_m], \emptytrace$ under context $[\dom_{i+1}, \mydots, \dom_m], \env$. 
Similarly, for each $\ftastate'_i$, it maps the context $[\dom_i, \mydots, \dom_{m}], \env$ to $[\action_i], [\dom_{i+1}, \mydots, \dom_m]$. 

In addition to $\seqprog$, we also have transitions for statements (like $\clickstatement$ and $\scrapetextstatement$) and selectors. The following rule shows how to construct transitions for $\clickstatement$ and its candidate selectors, from action trace $[\action_1, \mydots, \action_m]$, DOM trace $[\dom_1, \mydots, \dom_{m}]$, and input data $\inputdata$.  
\[
\small 
\hspace{-3pt}
\begin{array}{c}
\irule{
\begin{array}{c}
i \in [1, m] \ \ \ \ 
\action_i = \clickstatement(\fullxpath) \ \ \ \ 
\eval{\dom_i}{\concreteselector}{\fullxpath} \ \ \ \ 
\env = \{ \inputvar \mapsto \inputdata \} 
\\ 
\ftastate^{\prime}_i = \newstate
\Big( 
\expressionsymbol, 
\big\{ 
[ \dom_i, \mydots, \dom_m ], \env
\inputcontexttooutput 
[ \action_i ], [ \dom_{i+1}, \mydots, \dom_m ]
\big\} 
\Big) \ \ \ \ 
\ftastate^{\doubleprime}_i = \newstate 
\Big( 
\selectorexprsymbol, \big\{ \dom_i, \env \inputcontexttooutput \fullxpath \big\} 
\Big) \\ 
\ftatransition'_i = \clickstatement( \ftastate^{\doubleprime}_i ) \ftatransitionarrow \ftastate'_i \ \ \ \ 
\ftatransition^{\doubleprime}_i =  \concreteselector \ftatransitionarrow \ftastate^{\doubleprime}_i 
\end{array}
}{
\ftastate'_i, \ftastate^{\doubleprime}_i \in \ftastates \ \ \ \ 
\ftatransition'_i, \ftatransition^{\doubleprime}_i \in \ftatransitions
}
\end{array}
\]
Intuitively, this rule states: given a $\clickstatement$ action $\action_i$ (with a full selector expression $\fullxpath$) that is performed on DOM $\dom_i$,
for any \emph{candidate selector} $\concreteselector$ that refers to the same DOM element on $\dom_i$ as $\fullxpath$, we create a transition for $\concreteselector$. 
This is exactly why Figure~\ref{fig:alg:init-fta} has multiple selectors in the {dashed} circle around each $\fullxpath_i$. 
The construction rules for other actions are very similar; we skip the details here.

\begin{figure}
\centering
\arraycolsep=1pt\def\arraystretch{1}
\begin{centermath}
\footnotesize
\centering
\begin{array}{l}
\irule{
\begin{array}{ll}
\env = \{ \inputvar \mapsto \inputdata \} 
\\ 
\ftastate_0 = \newstate
\Big( 
\programsymbol, 
\big\{ 
[ \dom_1, \mydots, \dom_m ], \env \inputcontexttooutput [ \action_1, \mydots, \action_m ], \emptytrace 
\big\} 
\Big) 
& 
\ftastate'_1 = \newstate
\Big( 
\expressionsymbol, 
\big\{ 
[ \dom_1, \mydots, \dom_m ], \env
\inputcontexttooutput 
[ \action_1 ], [ \dom_2, \mydots, \dom_m ]
\big\} 
\Big) 
\\ 
\ftastate_1 = \newstate
\Big( 
\programsymbol, 
\big\{ 
[ \dom_2, \mydots, \dom_m ], \env \inputcontexttooutput [ \action_2, \mydots, \action_m ], \emptytrace 
\big\} 
\Big) 
& 
\ftastate'_2 = \newstate
\Big( 
\expressionsymbol, 
\big\{ 
[ \dom_2, \mydots, \dom_m ], \env
\inputcontexttooutput 
[ \action_2 ], [ \dom_3, \mydots, \dom_m ]
\big\} 
\Big) 
\\ 
\mydots
& 
\mydots 
\\ 
\ftastate_m = \newstate
\Big( 
\programsymbol, 
\big\{ 
\emptytrace, \env 
\inputcontexttooutput 
\emptytrace, \emptytrace 
\big\} 
\Big)
&  
\ftastate'_m = \newstate
\Big( 
\expressionsymbol, 
\big\{ 
[ \dom_m ], \env
\inputcontexttooutput 
[ \action_m ], \emptytrace
\big\} 
\Big)
\\ 
\ftatransition_1 = \seqprog( \ftastate'_1, \ftastate_1 ) \ftatransitionarrow \ftastate_{0} 
\ \mydots  \ 
\ftatransition_m = \seqprog( \ftastate'_m, \ftastate_m ) \ftatransitionarrow \ftastate_{m-1}
\end{array}
}{
\ftastate_0, \mydots, \ftastate_m, 
\ftastate'_1, \mydots, \ftastate'_m \in \ftastates \ \ \ \ 
\ftatransition_1, \mydots, \ftatransition_m \in \ftatransitions \ \ \ \ 
\ftastate_0 \in \ftafinalstates 
}
\end{array}
\end{centermath}
\vspace{-5pt}
\caption{Key FTA initialization rule, given action trace $[\action_1, \mydots, \action_m]$ and DOM trace $[\dom_1, \mydots, \dom_m]$.}
\label{fig:fta-init-rules}
\end{figure}

\begin{figure}[!t]
\centering
\vspace{-15pt}
\begin{minipage}{.99\linewidth}
\small 
\centering
\begin{algorithm}[H]
\begin{algorithmic}[1]
\myprocedure{SpeculateForData}{$[\ftatransition_1, \mydots, \ftatransition_{2l}], \ftatransitions$}

\Statex\Input{Each $\ftatransition_i$ is of the form $\ftatransition_i = \seqprog(\ftastate'_i, \ftastate_i) \ftatransitionarrow \ftastate_{i-1}$, $i \in [1, 2l]$. $\ftatransitions$ is a set of transitions.}

\Statex\Output{FTA with final state $\ftastatespeculated$ and transitions $\ftatransitions'$, which represents a set of speculated $\fordataloop$ loops.}

\vspace{2pt}
\State
$U_1 \assign \antiunifyfunc\textsc{ForData}(\ftastate'_{1}, \ftastate'_{l+1}, \ftatransitions)$; 
\ \ $\mycdots$; \ \ 
$U_l \assign \antiunifyfunc\textsc{ForData}(\ftastate'_{l}, \ftastate'_{2l}, \ftatransitions)$;
\label{alg:speculate-fta-fordata:anti-unify}

\vspace{3pt}

\State 
$U \assign U_1 \cup \mydots \cup \ U_l$;
\ \ 
$( \{ \ftastatespeculated'_{1} \}, \ftatransitions'_{1}) \assign \parametrizefunc(\ftastate'_1, \ftatransitions, U)$;
\ \ $\mydots$; \ \ 
$( \{ \ftastatespeculated'_{l} \}, \ftatransitions'_{l}) \assign \parametrizefunc(\ftastate'_l, \ftatransitions, U)$;
\label{alg:speculate-fta-fordata:parametrize}


\vspace{2pt}
\State 
create fresh states $\ftastatespeculated_{0}(\programsymbol, \emptyset), \mydots, \ftastatespeculated_{l}(\programsymbol, \emptyset)$, and $\ftastatespeculated'(\dataexprsymbol, \emptyset), \ftastatespeculated(\expressionsymbol, \emptyset)$;
\label{alg:speculate-fta-fordata:create-states}

\State
$\ftatransitions' \assign$ 
\hspace{7pt}
$\{ \fordataloop(\ftastatespeculated', \ftastatespeculated_0) \ftatransitionarrow \ftastatespeculated  \}$
$\cup \ \{  \seqprog(\ftastatespeculated'_{i}, \ftastatespeculated_{i}) \ftatransitionarrow \ftastatespeculated_{i-1}  \ | \  i \in [1, l] \}$ 
\label{alg:speculate-fta-fordata:create-transitions}

\Statex
\hspace{20pt}
$\cup \ \{ u \ftatransitionarrow \ftastatespeculated'  \ | \ u \in U \}$ 
\Comment{Transitions corresponding to anti-unifiers}

\Statex 
\hspace{20pt}
$\cup \ \ \ftatransitions'_{1} \cup \mydots \cup \ftatransitions'_{l}$; 
\Comment{Transitions constructed from parametrization}

\vspace{3pt}
\State
\Return $\big( \{ \ftastatespeculated \}, \ftatransitions' $ \big);
\vspace{-5pt}
\caption{Algorithm for \textsc{SpeculateForData} that speculates an FTA of $\fordataloop$ loops.}
\label{alg:speculate-fta-fordata}
\end{algorithmic}
\end{algorithm}
\end{minipage}
\end{figure}

\begin{figure}
\centering
\begin{minipage}[b]{.4\linewidth}
\centering
\includegraphics[height=4.4cm]{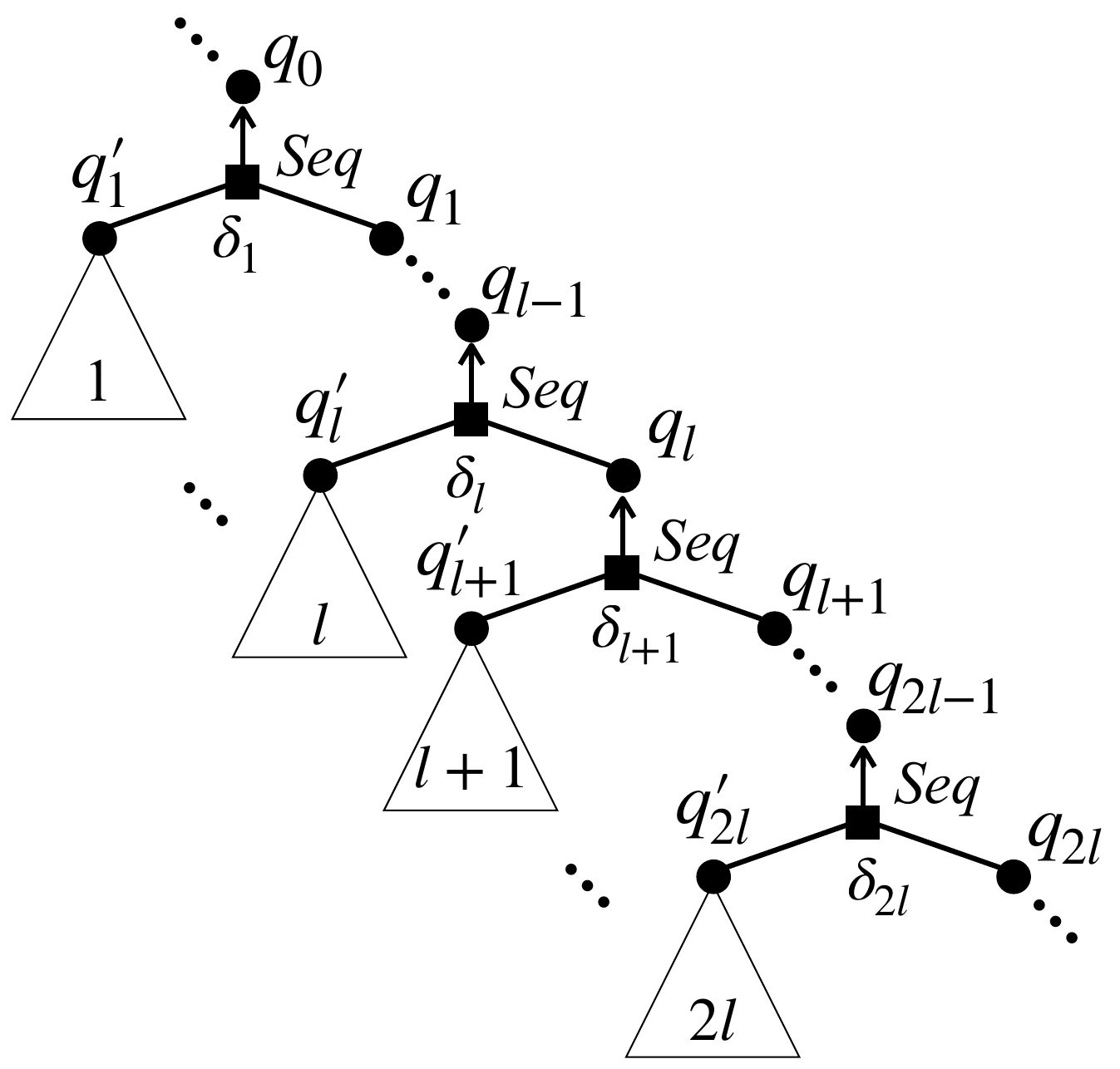}
\vspace{-5pt}
\caption{Illustration of the $2l$ consecutive transitions $\ftatransition_1, \mydots, \ftatransition_{2l}$ (line~\ref{alg:top-level:matching-transitions}, Algorithm~\ref{alg:top-level}).}
\label{fig:matched-transitions}
\end{minipage}
\hspace{5pt}
\begin{minipage}[b]{.57\linewidth}
\centering
\includegraphics[height=4.2cm]{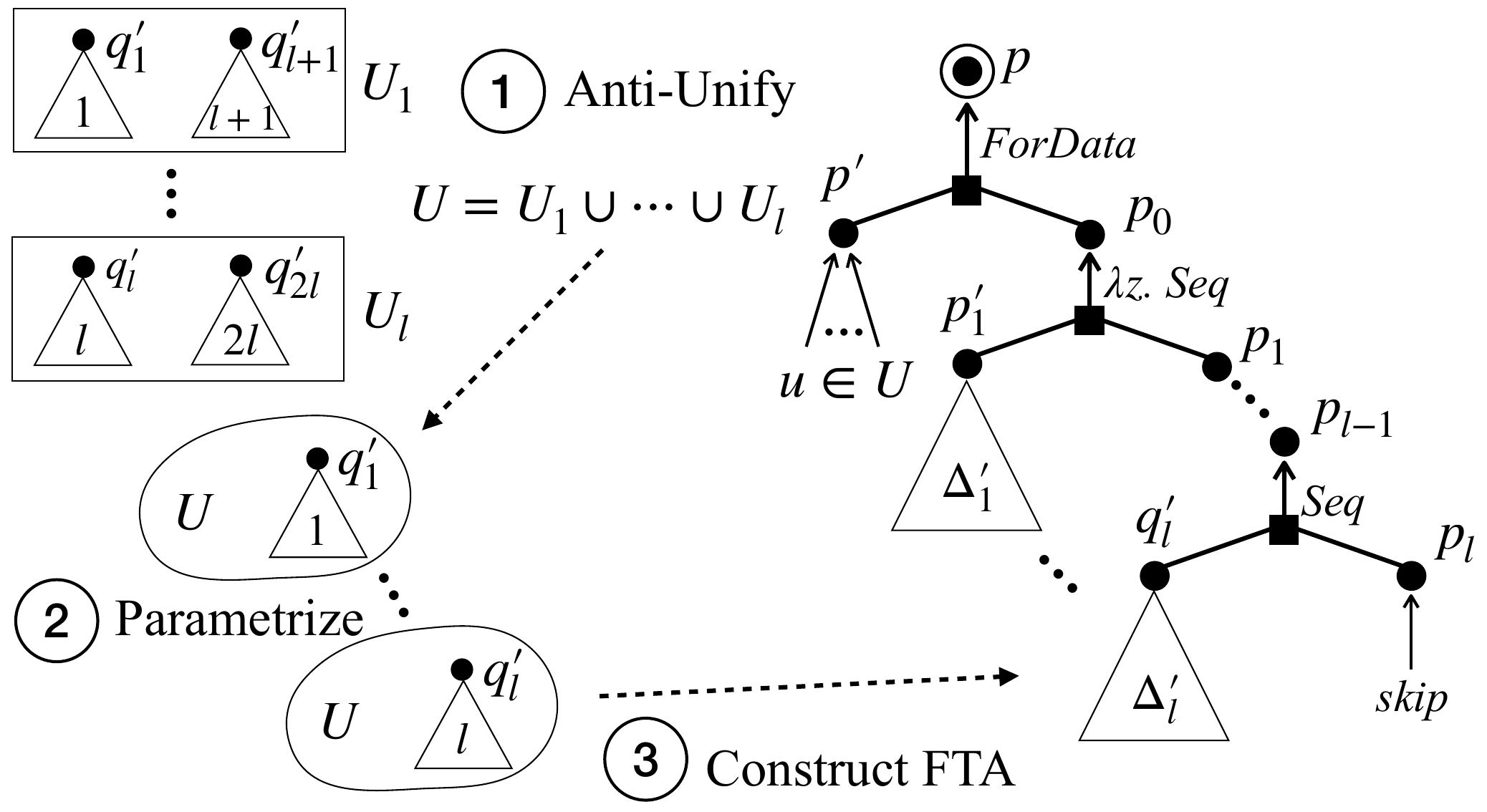}
\vspace{-5pt}
\caption{Illustration of Algorithm~\ref{alg:speculate-fta-fordata}, given the $2l$ transitions $\ftatransition_1, \mydots, \ftatransition_{2l}$ from Figure~\ref{fig:matched-transitions}.}
\label{fig:speculated-fta}
\end{minipage}
\end{figure}

\subsection{Speculating FTAs}\label{sec:alg:fta-speculation}

Section~\ref{sec:alg:fta-lifted-interp} assumed a given speculated FTA $\fta_s$. 
In this section, let us unpack the $\textsc{SpeculateFTA}$ algorithm that infers $\fta_s$. Here, $\fta_s$ is an FTA that contains candidate loops (potentially nested).  

Figure~\ref{fig:speculated-fta} illustrates how to speculate candidate $\fordataloop$ loops from the $2l$ consecutive transitions $\ftatransition_1, \mydots, \ftatransition_{2l}$ (illustrated in Figure~\ref{fig:matched-transitions}). 
Algorithm~\ref{alg:speculate-fta-fordata} describes the algorithm more formally. 
We suggest readers simultaneously looking at  Figures~\ref{fig:matched-transitions}, Figure~\ref{fig:speculated-fta} and Algorithm~\ref{alg:speculate-fta-fordata}. 

The first step is anti-unification (line~\ref{alg:speculate-fta-fordata:anti-unify}): for all $i \in [1, l]$, we synchronously traverse $\subfta(\ftastate'_i, \fta)$ and $\subfta(\ftastate'_{l+i}, \fta)$, and compute a set $U_i$ of anti-unifiers. 
For $\fordataloop$, an anti-unifier $u$ is simply a data expression that $\fordataloop$ may iterate over. 
The second step is parametrization (lines~\ref{alg:speculate-fta-fordata:parametrize}): for all anti-unifiers from $U$, we traverse $\subfta(\ftastate'_i, \fta)$ ($i \in [1,l]$) and build a new FTA with final state $\ftastatespeculated'_i$ and transitions $\ftatransitions'_i$. 
The third and last step is to construct the final speculated FTA $\fta_s$ (line~\ref{alg:speculate-fta-fordata:create-states}-\ref{alg:speculate-fta-fordata:create-transitions}). 
Each state in $\fta_s$ is annotated with an \emph{empty} footprint, because we are yet to evaluate any programs. 
In other words, $\fta_s$ is purely a syntactic compression without using any OE at all. 

We refer interested readers to the appendix for a more complete description of the speculation algorithm that handles other types of loops. 
In what follows, we explain our anti-unification and parametrization algorithms.

\begin{figure}[!t]
\small 
\centering
\begin{minipage}{.99\linewidth}
\centering
\arraycolsep=1pt\def\arraystretch{1}
\begin{centermath}
\centering
\hspace{-5pt}
\begin{array}{l}

(1) \ 
\irule{
\begin{array}{c}
\ftatransitionoperator \in \{ \clickstatement, \scrapetextstatement, \scrapelinkstatement, \downloadstatement, \enterdatastatement \}
\\ 
\ftatransitionoperator(\mydots, \ftastate'_i, \mydots) \ftatransitionarrow \ftastate_i \in \ftatransitions \ \ \ \ 
\antiunify{\ftatransitions}{\ftastate'_1, \ftastate'_2}{u}
\end{array}
}{
\antiunify{\ftatransitions}{\ftastate_1, \ftastate_2}{u}
}

\ \ \ \ \ 

(2) \ 
\irule{
\begin{array}{c}
\ftatransitionoperator \in \{ \fordataloop, \forselectorsloop \}
\\ 
\ftatransitionoperator(\ftastate'_i, \mydots) \ftatransitionarrow \ftastate_i \in \ftatransitions \ \ \ \ 
\antiunify{\ftatransitions}{\ftastate'_1, \ftastate'_2}{u}
\end{array}
}{
\antiunify{\ftatransitions}{\ftastate_1, \ftastate_2}{u}
}

\\[10pt]

(3) \ 
\irule{
\begin{array}{c}
\dataexpr_i \ftatransitionarrow \ftastate_i \in \ftatransitions 
\\ 
\antiunify{}{\dataexpr_1, \dataexpr_2}{u}
\end{array}
}{
\antiunify{\ftatransitions}{\ftastate_1, \ftastate_2}{u}
}

\ \ 

(4) \ 
\irule{
\begin{array}{c}
\\ 
\dataexpr_i = \dataexpr' \big[ \datavar \mapsto \dataexpr[i] \big] 
\end{array}
}{
\antiunify{}{\dataexpr_1, \dataexpr_2}{\dataexpr}
}

\ \ 

(5) \ 
\irule{
\begin{array}{c}
\concreteselector_i \ftatransitionarrow \ftastate_i \in \ftatransitions 
\\ 
\antiunify{}{\concreteselector_1, \concreteselector_2}{u}
\end{array}
}{
\antiunify{\ftatransitions}{\ftastate_1, \ftastate_2}{u}
}

\ \ 

(6) \ 
\irule{
\begin{array}{c}
\oplus \in \{ \slash, \sslash \} 
\\ 
\concreteselector_i = \selectorexpr' \big[\selectorvar \mapsto \concreteselector \oplus \predicate[k + i - 1]\big] 
\end{array}
}{
\antiunify{}{\concreteselector_1, \concreteselector_2}{\concreteselector \oplus \predicate [k] }
}

\end{array}
\end{centermath}
\vspace{-5pt}
\caption{Anti-unification rules.}
\label{fig:anti-unify-rules}
\end{minipage}
\end{figure}

\begin{figure}[!t]
\small 
\centering
\begin{minipage}{.99\linewidth}
\centering
\arraycolsep=1pt\def\arraystretch{1}
\begin{centermath}
\centering
\hspace{-5pt}
\begin{array}{l}

(1) \ 
\irule{
\begin{array}{c}
\begin{array}{c}
\ftatransitionoperator(\mydots, \ftastate', \mydots) \ftatransitionarrow \ftastate \in \ftatransitions 
\ \ \ \ 
\oplus \in \{ \slash, \sslash \} 
\ \ \ \ 
u \in U
\\ 
u \oplus \selectorexpr \ftatransitionarrow \ftastate' \in \ftatransitions 
\end{array}
\end{array}
}{
\selectorvar \oplus \selectorexpr \ftatransitionarrow M(\ftastate') \in \ftatransitions' 
}

\hspace{30pt}

(2) \ 
\irule{
\begin{array}{c}
\ftatransitionoperator(\mydots, \ftastate', \mydots) \ftatransitionarrow \ftastate \in \ftatransitions 
\ \ \ \ 
u \in U 
\\ 
u[1] \ftatransitionarrow \ftastate' \in \ftatransitions 
\end{array}
}{
\datavar \ftatransitionarrow M(\ftastate') \in \ftatransitions'
}

\end{array}
\end{centermath}
\vspace{-5pt}
\caption{Parametrization rules.}
\label{fig:parametrize-rules}
\end{minipage}
\end{figure}

\newpara{Anti-unification.}
Given two states $\ftastate_1, \ftastate_2$ from FTA $\fta$ with transitions $\ftatransitions$, anti-unification traverses expressions in $\subfta(\ftastate_1, \fta)$ and $\subfta(\ftastate_2, \fta)$ in a synchronized fashion, and returns a set $U$ of anti-unifiers that may be iterated over by the speculated loops.
In particular: 
\[
\begin{array}{c}
\textsc{AntiUnify}(\ftastate_1, \ftastate_2, \ftatransitions) 
= U
= 
\{ \ u \ | \ \antiunify{\ftatransitions}{\ftastate_1, \ftastate_2}{u} \ \}
\end{array}
\]
That is, $u \in U$ if there is an anti-unifier $u$ for $\ftastate_1, \ftastate_2$ given $\ftatransitions$. 
Figure~\ref{fig:anti-unify-rules} presents the detailed rules. 
Rule (1) says that, to anti-unify $\ftastate_1, \ftastate_2$ with incoming transition $\ftatransitionoperator$, we anti-unify the argument states $\ftastate'_1, \ftastate'_2$. 
Rule (2) does the same but for loops: we anti-unify expressions  (i.e., the first argument) being looped over. 
Rules (3) and (4) concern the anti-unification of data expressions --- the idea is to look for an increment pattern beginning with 1. 
Rules (5) and (6) anti-unify selector expressions, where it uses a more flexible pattern that allows starting from $k$ (not necessarily $k=1$).

\begin{example}
Consider states $\ftastate'_1$ and $\ftastate'_4$ in $\fta$ from Figure~\ref{fig:alg:init-fta}. The anti-unifier for $\ftastate'_1, \ftastate'_4$ is $\inputvar$. 
\label{ex:anti-unify}
\end{example}

\newpara{Parametrization.}
Given all anti-unifiers $U$ and a state $\ftastate$ from $\fta$ with transitions $\ftatransitions$, $\textsc{Parametrize}$  constructs from $\subfta(\ftastate, \fta)$ a fresh new $\fta'$ (with transitions $\ftatransitions'$ and final state $\ftastate'$) in two steps. 

First, we make a \emph{fresh new} copy of each state from $\subfta(\ftastate, \fta)$, keeping the same grammar symbol and but resetting the footprint to empty. 
This results in a mapping $M$ that maps every state in $\subfta(\ftastate, \fta)$ to a state in $\fta'$. 
Therefore, for each transition $\ftatransitionoperator( \ftastate_1, \mydots, \ftastate_n) \ftatransitionarrow \ftastate_0 \in \ftatransitions$, we have $\ftatransitionoperator \big( M(\ftastate_1), \mydots, M(\ftastate_n) \big) \ftatransitionarrow M(\ftastate_0) \in \ftatransitions'$.  

Then, we add new transitions labeled with \emph{parametrized} selector/data expressions to $\fta'$, given each anti-unifier $u \in U$, mapping $M$, and the state $\ftastate$ from $\fta$ with transitions $\ftatransitions$. In other words, the {actual} parametrization takes place in this step. 
Figure~\ref{fig:parametrize-rules} presents our parametrization rules. 
Rule (1) parametrizes any transition in $\ftatransitions$ that uses a selector expression of the form $u \oplus \selectorexpr$; that is, this selector uses the anti-unifier $u$ as a prefix. This rule adds a new transition with $u$ being replaced by variable $\selectorvar$ which points to state $M(\ftastate')$. 
Rule (2) parametrizes data expressions in a similar manner.

\begin{example}
Consider the initial FTA $\fta$ in Figure~\ref{fig:alg:init-fta}. 
Consider the anti-unifier $\inputvar$ from Example~\ref{ex:anti-unify}. 
Let us explain how to parametrize state $\ftastate'_1$ from $\fta$, given anti-unifier $\inputvar$, to generate state $\ftastatep'_1$ in $\fta_s$ from Figure~\ref{fig:alg:speculate-fta}. 
We first make a copy of $\subfta(\ftastate'_1, \fta)$: for example, the copy of $\ftastate'_1$ is $\ftastatep'_1$; that is, $M(\ftastate'_1) = \ftastatep'_1$. 
Then, we invoke Rule (2) from Figure~\ref{fig:parametrize-rules} with $\ftatransitionoperator$ being $\enterdatastatement$. This yields a new transition $\datavar \ftatransitionarrow M(\ftastate^{\doubleprime}_1)$ where $\ftastatep^{\doubleprime}_1 = M(\ftastate^{\doubleprime}_1)$, which indeed appears in $\fta_s$ in Figure~\ref{fig:alg:speculate-fta}. 
\label{ex:parametrize}
\end{example}

\newpara{Constructing $\fta_s$.}
Finally, lines~\ref{alg:speculate-fta-fordata:create-states}-\ref{alg:speculate-fta-fordata:create-transitions} connect the created states and transitions, as shown in Figure~\ref{fig:speculated-fta}.

\subsection{Merging FTAs}\label{sec:alg:merging-and-ranking}

Let us explain the last two procedures in Algorithm~\ref{alg:top-level}.

\begin{wrapfigure}{r}{0.25\textwidth}
\vspace{-15pt}
\begin{minipage}[t]{\linewidth}
\centering
\includegraphics[scale=0.21]{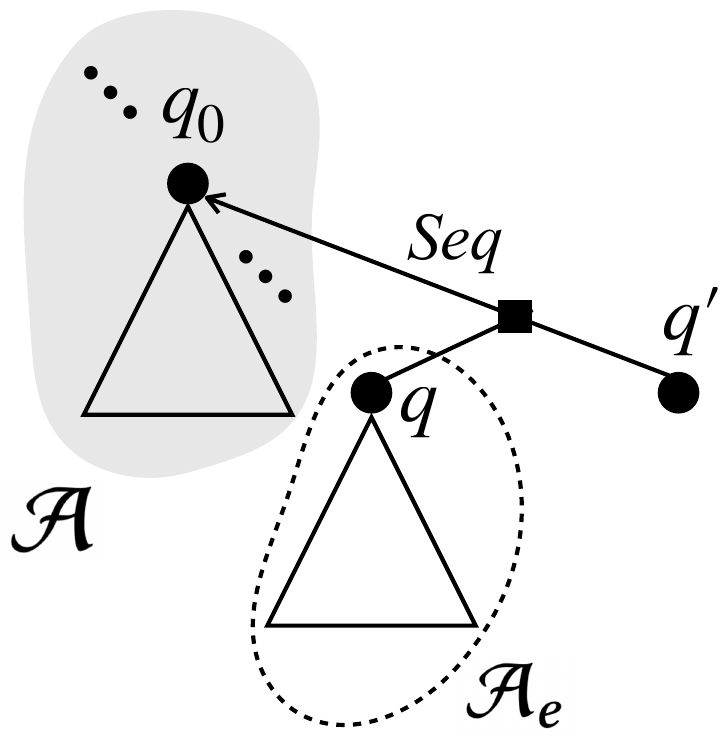}
\vspace{-5pt}
\caption{$\textsc{MergeFTAs}$.}
\label{fig:merge}
\end{minipage}
\vspace{-10pt}
\end{wrapfigure}
\vspace{3pt}
\noindent\textbf{\emph{Merging $\fta_e$ into $\fta$.}}
Intuitively, $\textsc{MergeFTAs}$ (line~\ref{alg:top-level:merge}) incorporates loops from $\fta_e$ into $\subfta(\ftastate_0, \fta)$.
Figure~\ref{fig:merge} illustrates how this works. 
We first construct a state $\ftastate'$ and a $\seqprog$ transition $\ftatransition$ that connects a final state $\ftastate$ from $\fta_e$ and $\ftastate'$ to $\ftastate_0$. 
In particular, $\ftastate'$ has grammar symbol $\programsymbol$ and is associated with the following footprint: 
\[\small 
\centering
\hspace{-105pt}
\left\{ 
\begin{array}{c}

\domtrace'_1, \env \inputcontexttooutput \actiontrace'_2, \domtrace' 

\ \Big| \ 

\begin{array}{l}
\domtrace, \env \inputcontexttooutput \actiontrace', \domtrace'  \ \in \getannot(\ftastate_0) \\ 
\domtrace, \env \inputcontexttooutput \actiontrace'_1, \domtrace'_1 \in \getannot(\ftastate) \\ 
\actiontrace'_1 \traceconcat \actiontrace'_2 = \actiontrace'
\end{array}
\end{array}
\right\}
\vspace{5pt}
\]
In other words, $\ftastate'$ is constructed based on footprints from $\ftastate_0$ and $\ftastate$, according to $\seqprog$'s semantics. 
Note the constraint $\actiontrace'_1 \traceconcat \actiontrace'_2 = \actiontrace'$: we only keep loops in $\subfta(\ftastate, \fta_e)$ that yield a \emph{prefix} of $\actiontrace'$, because otherwise $\actiontrace'_1$ is not a reachable context for $\ftastate$, at least not so in this case. 
For such final states $\ftastate$ of $\fta_e$, we create the aforementioned $\ftastate'$ and $\ftatransition$, and add $\ftatransition$ together with all transitions in $\fta_e$ to the set of transitions of $\fta$. 
Observe that, if $\ftastate'$ already exists in $\fta$, then $\textsc{MergeFTAs}$ effectively connects $\ftastate$ and an existing state $\ftastate'$ in $\fta$ to $\ftastate_0$ which is also in $\fta$.

\newpara{Ranking.}
The \textsc{Ranking} procedure (line~\ref{alg:top-level:rank}) is fairly straightforward. 
It first runs $\textsc{EvaluateFTA}$ on $\fta$ using $\domtrace, \{ \inputvar \mapsto \inputdata \}$ as the context. Note that the last DOM $\dom_{m+1}$ in $\domtrace$ does not have a demonstrated action in the input action trace $\actiontrace$, because our goal is to use the synthesized program $\program$ to automate the unseen actions.  
$\textsc{EvaluateFTA}$ gives $\fta'$ containing programs with different predicted actions $\action_{m+1}$. 
We pick a program $\program$ with the smallest size and return $\program$ as the final synthesized program.

\subsection{Soundness and Completeness}\label{sec:alg:theorems}

\begin{theorem}
Given action trace $\actiontrace$, DOM trace $\domtrace$ and input data $\inputdata$, 
our synthesis algorithm always terminates. Moreover, if there exists a program in our grammar (shown in Figure~\ref{fig:alg:dsl}) that generalizes $\actiontrace$ (given $\domtrace$ and $\inputdata$) and satisfies the condition that (1) every loop has at least two iterations exhibited in $\actiontrace$ and (2) its final expression is a loop, then our synthesis algorithm (shown in Algorithm~\ref{alg:top-level}) would return a program that generalizes $\actiontrace$ (given $\domtrace$ and $\inputdata$) upon FTA saturation. 
\label{thm:alg-soundness-completeness}
\end{theorem}

\subsection{Interactive Synthesis with Incremental FTA Construction}\label{sec:alg:incremental}

\newpara{Interactive programming-by-demonstration.}
Same as $\webrobot$~\cite{dong2022webrobot}, our synthesis technique can also be applied in an interactive setting: given an action trace $\actiontrace_m=[\action_1, \mydots, \action_m]$ and a DOM trace $\domtrace_{m} = [\dom_1, \mydots, \dom_{m+1}]$, we synthesize a program $\program_m$ that produces $\action'_{m+1}$ given $\domtrace_{m}$. 
If this prediction $\action'_{m+1}$ is not intended, the user  would manually demonstrate a correct $\action_{m+1}$. This leads to new traces $\actiontrace_{m+1}=[\action_1, \mydots, \action_{m+1}]$ and $\domtrace_{m+1} = [\dom_1, \mydots, \dom_{m+2}]$, both of which are fed to our algorithm again. 
This process repeats until the user obtains an intended program.

\newpara{Incremental FTA construction.}
We incrementalize our Algorithm~\ref{alg:top-level} based on the interactive setup. 
Given $\actiontrace_{m+1}$ and $\domtrace_{m+1}$, and given FTA $\fta_{m}$ constructed from $\actiontrace_m$ and $\domtrace_m$, we still build $\fta_{m+1}$ with the same guarantee as in Algorithm~\ref{alg:top-level} but not from scratch. 
Our key insight is that we can re-use the validated loops from $\fta_m$, but we need to evaluate them against $\domtrace_{m+1}$ as some of them may not reproduce $\actiontrace_{m+1}$. 
Essentially, our incremental algorithm is still based on guess-and-check, but we have a second type of speculation that directly takes sub-FTAs from $\fta_m$ as speculated FTAs --- this generates high-quality speculated FTAs more efficiently. 
$\fta_{m+1}$ is still initialized in the same way but this time using $\actiontrace_{m+1}$ and $\domtrace_{m+1}$. 
We still perform $\textsc{SpeculateFTA}$ and $\textsc{EvaluateFTA}$ on $\fta_{m+1}$ but in a much smaller scope this time. That is, 
$\ftatransition_{2l}$ must involve $\action_{m+1}$, because otherwise  $\fta_m$ had already considered it. 
The $\textsc{MergeFTAs}$ and $\textsc{Ranking}$ procedures remain the same.

}

\section{Evaluation}
\label{sec:eval}

This section describes a series of experiments
designed to answer the following questions: 
\begin{itemize}[leftmargin=*]
\setlength\itemsep{3pt}
\item 
\textbf{RQ1}: 
Can $\tool$ \emph{efficiently} synthesize programs for \emph{challenging} web automation tasks? How does it compare against state-of-the-art techniques? 
\item 
\textbf{RQ2}: 
How necessary is it to use an \emph{expressive} language in order to have a generalizable program? 
In particular, is it important to consider a large space of candidate selectors? 
\item 
\textbf{RQ3}: 
How does $\tool$ scale with respect to the number of candidate selectors considered? 
\item 
\textbf{RQ4}: 
How useful are various ideas proposed in this work? 
\end{itemize}

\newpara{Web automation tasks.}
To answer these questions, we construct a collection of web automation tasks from two different sources. 
First, we include \emph{all} \finalized{76 tasks} that were used to evaluate $\webrobot$: 
details about these tasks can be found in~\cite{dong2022webrobot}. 
Second, we curate \finalized{55} \emph{new} tasks, each of which has an English description of the task logic over one or more websites. 
All of our new tasks are curated based on real-life problems (e.g., those from the iMacros forum) and involve modern, popular websites (such as Amazon, UPS, Craigslist, IMDb) with complex webpages, \revision{whereas many of $\webrobot$'s benchmarks involve legacy websites. 
These new websites have deeply nested DOM structures which would require searching with a significantly larger space of candidate selectors.}
In particular, all \finalized{131} tasks involve data extraction, \finalized{45} of them involve data entry, \finalized{80} require navigation across webpages, and \finalized{48} involve pagination. 
Some of these tasks involve multiple types: for instance, \finalized{31} of them involve data entry, data extraction, and webpage navigation.

\newpara{Ground-truth Selenium programs.}
For each task, we obtain a ground-truth automation program $\program_{\emph{gt}}$ (using the Selenium WebDriver framework). 
In particular, we reuse the \finalized{76} ground-truth Selenium programs from $\webrobot$, and manually write \finalized{55} programs for our new tasks. 
On average, these Selenium programs consist of \finalized{43} lines of code, with a max of \finalized{147} lines. In general, it takes about half an hour to a few hours for us to write up an automation program, depending on the complexity of webpages and the task logic.

\newpara{Benchmarks.}
In our evaluation, a benchmark is defined as a tuple $(\actiontrace, \domtrace, \inputdata)$, where $\actiontrace$ is an action trace, $\domtrace$ is a DOM trace, and $\inputdata$ is input data. 
We obtain one benchmark for each task by running the corresponding $\program_{\emph{gt}}$ in the browser (given input data $\inputdata$ if $\program_{\emph{gt}}$ involves data entry) --- during execution, we record the trace $\actiontrace=[\action_1, \mydots, \action_m]$ of actions that $\program_{\emph{gt}}$ executes as well as $\actiontrace$'s corresponding DOM trace $\domtrace=[\dom_1, \mydots, \dom_m]$.\footnote{We terminate $\program_{\emph{gt}}$ after every loop from $\program_{\emph{gt}}$ has been executed for three full iterations or when $|\actiontrace|$ reaches 500, whichever yields a longer action trace. This essentially serves as a ``timeout'' to avoid running $\program_{\emph{gt}}$ for an unnecessarily long time.} 
Here, $\action_i$ is an action performed on $\dom_i$. 
Note that selectors in actions are recorded as full XPath expressions.

\newpara{Candidate selectors.}
For any full XPath selector $\fullxpath$ in an action $\action_i$, we also record a set $S$ of \emph{candidate selectors} for $\fullxpath$. 
In particular, a candidate selector $\selectorexpr$ is a \emph{concrete} selector expression (i.e., without variable $\selectorvar$) from our grammar (see Figure~\ref{fig:alg:dsl}) that locates the same DOM element on $\dom_i$ as $\fullxpath$. 
In other words, $\selectorexpr$ evaluates to $\fullxpath$ given $\dom_i$, or more formally, $\eval{\dom_i}{\selectorexpr}{\fullxpath}$. 
As mentioned in Example~\ref{ex:intro:example} and Section~\ref{sec:alg:top-level}, it is important to consider candidate selectors, since full XPath expressions typically do not generalize. 
Candidate selectors also affect the overall search space, as explained below.

\newpara{Program space.}
In this work, the search space of programs for a benchmark with action trace $\actiontrace$ is defined jointly by the grammar from Figure~\ref{fig:alg:dsl} and the candidate selectors for all actions in $\actiontrace$. 
This is because Figure~\ref{fig:alg:dsl} does not specify \emph{a priori} the space of selectors $\selectorexpr$ and predicates $\predicate$. 
Instead, they are defined once the traces $\actiontrace$ and $\domtrace$ are given: e.g., the space for $\selectorexpr$ includes all candidate selectors. It is very hard (if not impossible) to define this space a priori without the DOMs. 
Obviously, the overall search space of programs grows when we consider more candidate selectors.

\subsection{RQ1: Can \textsc{Arborist} Automate Challenging Web Automation Tasks?}
\label{sec:eval:main}

In this section, we evaluate $\tool$ against three metrics:
(i) how many tasks it can successfully synthesize intended programs for, 
(ii) how much synthesis time it takes, and
(iii) how many user-demonstrated actions it requires in order to synthesize intended programs. 
In other words, we evaluate $\tool$'s effectiveness, efficiency, and generalization power. 
We also compare $\tool$ against the state-of-the-art, especially on our new tasks that are more challenging.

\newpara{Setup.}
We use the same setup from the $\webrobot$ paper~\cite{dong2022webrobot}. In particular, given a benchmark with $\actiontrace$ and $\domtrace$ that have $m$ actions and DOMs respectively, we create $m-1$ \emph{tests}. 
The $k$th test consists of an action trace $\actiontrace_{k} = [ \action_1, \mydots, \action_k]$ and a DOM trace $\domtrace_{k} = [\dom_1, \mydots, \dom_k, \dom_{k+1}]$. 
For each action, we consider all candidate selectors in our grammar with {at most 3 predicates.} 
For example, {\small \texttt{a[1]/b[2]}} uses 2 predicates and therefore is included, but {{\small \texttt{c/d/a[1]/b[2]}}} is not considered even if it can also locate the same DOM element. 
We run $\tool$ in a way to simulate an interactive PBD process, same as how $\webrobot$ was evaluated. 
That is, we feed all tests to $\tool$ \emph{in sequence}: we run $\tool$ on the $k$-th test (with $\actiontrace_k$ and $\domtrace_k$), obtain a synthesized program $\program_k$, and check if $\program_k$ predicts $\action_{k+1}$ (i.e., $\program_k$ yields $\actiontrace_{k+1}$ given $\domtrace_k$). 
If not, $\action_{k+1}$ is counted as a \emph{user-demonstrated} action; otherwise, $\action_{k+1}$ can be correctly predicted and thus is not counted. 
We always count the first action $\action_1$ as user-demonstrated. 
Furthermore, for each test we use a 1-second timeout and record the time it takes to return $\program_k$. 
We run $\tool$ \emph{incrementally} (as described in Section~\ref{sec:alg:key-optimizations}) in this experiment: for each test, it resumes synthesis based on FTAs from previous tests/iterations. 
Finally, we inspect if the synthesized program $\program_{m-1}$ (in the last iteration) is an intended program: if so, the corresponding benchmark is counted as solved; otherwise, unsolved.

\newpara{Baselines.}
Among three baselines, we focus on the following two. 
\begin{itemize}[leftmargin=*]
\item 
$\webrobot$, which is the original tool from the $\webrobot$ work~\cite{dong2022webrobot}. 
We note that, while its underlying algorithm is complete in theory, the implementation is not. 
For example, it restricts the number of parametrized selectors (to five) when parametrizing actions during speculation. 
These heuristics \emph{seemed} to help avoid excessively slowing down the search, without severely hindering the completeness for those tasks considered in the $\webrobot$ work. 
\item 
$\webrobot$-extended, which  is an adapted version of $\webrobot$ that uses the extended language from Figure~\ref{fig:alg:dsl} (which allows $\forselectorsloop$ with $i \geq 1$).
In this experiment, we range $i$ from $1$ to $3$.  
This baseline still keeps all heuristics from $\webrobot$ (i.e., its search is incomplete).
\end{itemize}
We use the same space of candidate selectors as $\tool$ for these two baselines. 
\revision{The third baseline is $\helena$~\cite{chasins2019thesis}, which is also a PBD-based web automation tool.}

\evalmybox{\textbf{RQ1 take-away}:
\begin{itemize}[leftmargin=*]
\item 
$\tool$ can synthesize intended programs for \finalized{93.9\%} of our benchmarks. 
\item 
It typically takes $\tool$ subseconds to synthesize programs from demonstration. 
\item 
$\tool$ uses a median of \finalized{12} user-demonstrated actions to generalize.
\item 
$\tool$ can solve more benchmarks using less time than state-of-the-art techniques.
\end{itemize}
}

\begin{figure}
\centering
\begin{minipage}[b]{.49\linewidth}
\centering
\includegraphics[height=4.5cm]{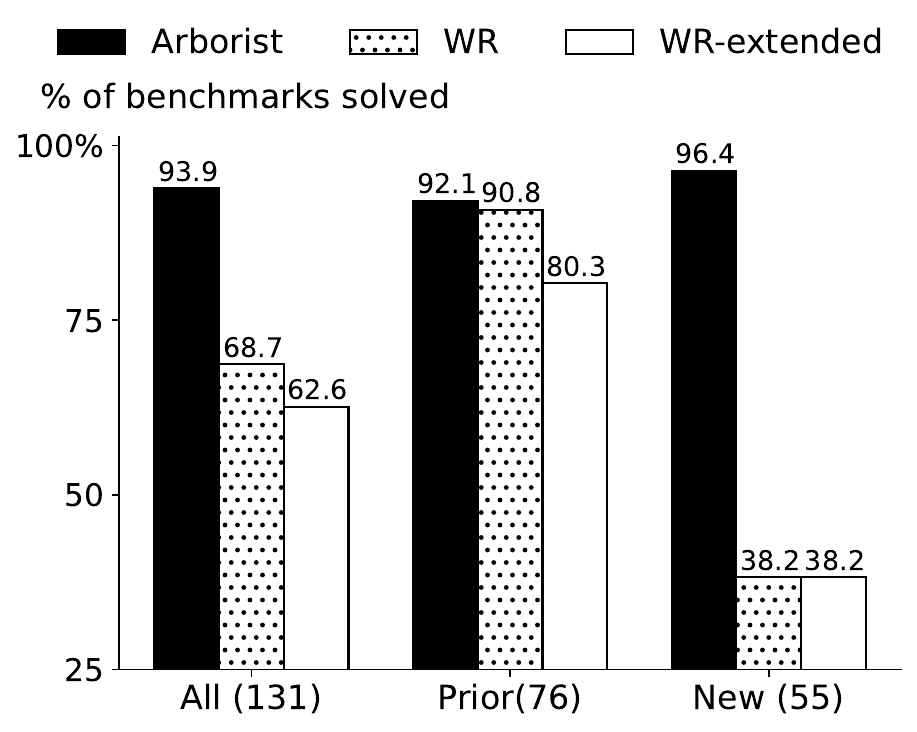}
\vspace{-5pt}
\hspace{2pt}
\caption*{(a) Percentage of benchmarks solved.}
\end{minipage}
\hspace{1pt}
\begin{minipage}[b]{.49\linewidth}
\centering
\includegraphics[height=4.45cm]{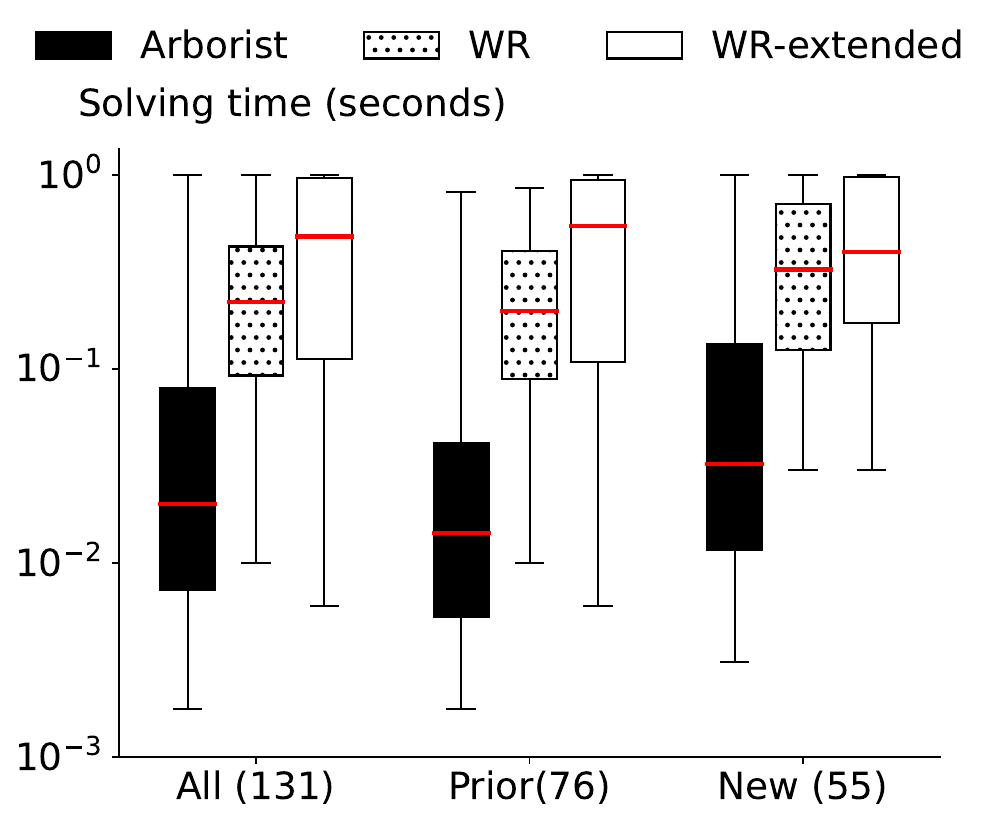}
\vspace{-5pt}
\caption*{(b) Synthesis times for solved benchmarks.}
\end{minipage}
\vspace{-5pt}
\caption{RQ1 main results. WR stands for $\webrobot$, and WR-extended means $\webrobot$-extended.}
\label{fig:eval:RQ1-main-results}
\end{figure}

\newpara{Main results.}
Figure~\ref{fig:eval:RQ1-main-results} summarizes the main results for both $\tool$ and baselines: Figure~\ref{fig:eval:RQ1-main-results}(a) shows the \emph{percentage} of benchmarks where intended programs can be synthesized, and Figure~\ref{fig:eval:RQ1-main-results}(b) reports the distribution of their corresponding synthesis times. 
Note that for both figures, we show both the \emph{aggregated} data {across} all \finalized{131} benchmarks, and \emph{separately} for prior and new tasks. 

Let us inspect Figure~\ref{fig:eval:RQ1-main-results}(a) first. 
Across all \finalized{131} benchmarks, $\tool$ can synthesize an intended program for \finalized{93.9\% (i.e., 123)} of them, whereas baselines solve at most \finalized{68.7\% (i.e., 90)}. 
This is a large gap, because baselines solve significantly fewer \emph{new tasks} (which are very challenging): in particular, $\tool$ can solve \finalized{53 out of 55 (i.e., 96.4\%), which is 2.5x more than that for baselines (i.e., 21/38.2\%).} 
These new tasks involve complex webpages and task logics, which require using nested loops and searching for selectors with more predicates in a larger space. 
This makes them significantly more challenging than prior tasks; as a result, $\webrobot$'s underlying enumeration-based algorithm fundamentally cannot scale to this level of complexity.
For prior tasks \finalized{(total 76)} in the $\webrobot$ paper (which baselines were developed and engineered on), $\tool$ still outperforms baselines by one more benchmark. 
It turns out this benchmark involves two loops that require the start index to be 2 and 3, which cannot be expressed in $\webrobot$'s original language. $\webrobot$-extended, however, solves this benchmark, since it uses a richer language. $\tool$ uses the same (extended) language, and therefore is able to synthesize an intended program as well.

In addition to solving a \emph{strict superset} of benchmarks than baselines, $\tool$ is also significantly faster. 
Figure~\ref{fig:eval:RQ1-main-results}(b) reports statistics of the synthesis times. 
In particular, for each solved benchmark, we record the maximum synthesis time across all tests, and report the distribution of these times. 
For each tool, Figure~\ref{fig:eval:RQ1-main-results}(b) \finalized{presents the quartile statistics of synthesis times.}
Across all \finalized{123 benchmarks} solved by $\tool$, \finalized{120} of them do not even use up the 1-second timeout. 
In contrast, baselines solve fewer benchmarks and time out on \finalized{15} benchmarks. Recall that both baselines have internal heuristics which unsoundly prune search space for faster search. 
We tested variants of them with these heuristics removed: the best one can solve \finalized{55} benchmarks, with \finalized{46} timeouts. 
In other words, for many benchmarks, baselines cannot exhaust the  entire search space, although they stumped upon a correct program before timeout. 
With a longer timeout of 10 seconds, baselines can only solve \finalized{9} more benchmarks. 
On the other hand, {with 1-second timeout,} $\tool$ times out on \finalized{3 benchmarks} --- they all require doubly nested loops and are among the most challenging ones. 
$\tool$ solves them all (albeit reaching timeout), while baselines solve two.

$\tool$ uses a median of \finalized{12} user-demonstrated actions, which is in line with that for $\webrobot$. 
This is reasonable, as we use the same search space for all tools in this experiment. 
While $\tool$ searches more programs than baselines (which oftentimes time out and hence search a small subset), our simple ranking heuristics seem to be quite effective at selecting generalizable programs.

\newpara{Discussion.}
$\tool$ failed to solve \finalized{8} benchmarks, including \finalized{6} from prior work (due to limitations of the web automation language, as also explained in the $\webrobot$ paper) and \finalized{2} new ones (which can be solved using a 10-second timeout).
Careful readers may wonder why $\webrobot$-extended solves {fewer} benchmarks than $\webrobot$, albeit using a richer language. This is due to the poor performance of its underlying enumeration-based algorithm: using a 1-second timeout, $\webrobot$-extend{ed} cannot even find intended programs for some benchmarks that $\webrobot$ solves. 

\newpara{Detailed results.}
Recall that $\tool$ internally has two key modules --- namely, $\textsc{SpeculateFTA}$ and $\textsc{EvaluateFTA}$ --- which typically use most of the running time. 
Among them, on average, the former takes \finalized{20\%} of the time, and the latter uses \finalized{80\%}, across all benchmarks. 
The final FTA has an average of \finalized{1714} states. 
The final synthesized programs on average have \finalized{6} expressions, and the largest one has \finalized{20}. 
Among these programs, \finalized{76} use at least one doubly-nested loop, and \finalized{12} involve at least a three-level loop.

\revision{
\newpara{$\tool$ vs. $\helena$.}
Among the 76 $\webrobot$ benchmarks, $\helena$'s PBD technique was able to synthesize intended programs from demonstrations (provided by us manually) for 13 benchmarks. 
For the 55 new tasks, $\helena$ solved 11. 
By contrast, $\tool$ solved 70 and 53 respectively. 
}

\subsection{RQ2: How Important Is It To Consider Many Candidate Selectors?}
\label{sec:eval:sparsity}

The search space considered in RQ1 uses the grammar from Figure~\ref{fig:alg:dsl} with candidate selectors of size up to {three} (measured by the number of predicates). 
This is a \emph{quite expressive} language containing at least one generalizable program for over \finalized{93\%} of our tasks (as $\tool$ was able to solve them). 
However, one might ask: is this high expressiveness really necessary? This is an important question, because if not, advanced synthesis techniques like $\tool$ may not be necessary in practice. 

\newpara{Setup.}
In this section, we investigate the impact of candidate selectors on the \emph{expressiveness} of the resulting search space: that is, given a set $S$ of candidate selectors, whether or not the corresponding program space has at least one \emph{generalizable} program.  
We choose to focus on candidate selectors in this experiment, because prior work~\cite{dong2022webrobot} has already shown the necessity of operators from the grammar in Figure~\ref{fig:alg:dsl}. 

More specifically, given a benchmark with action trace $\actiontrace$, for each $\action_i$ with full XPath selector $\fullxpath$, we use the following ways to construct the set $S$ of candidate selectors for $\action_i$. 
\begin{itemize}[leftmargin=*]
\item 
$S$ includes \emph{all} candidate selectors for $\fullxpath$ \emph{up to a certain size}. 
This is perhaps the simplest heuristic that one can design; RQ1 uses 3 as the max size which was shown to be sufficient for most of our benchmarks. 
{Therefore, in this experiment, we vary the max size from 1 to 3, and investigate how it impacts the expressiveness.}  
For each size, we run $\tool$ using the corresponding $S$, and record the number of benchmarks \emph{solved}.
Here, ``solved'' means an intended program can be found by $\tool$, which is a witness that corresponding search space is expressive enough.\footnote{Thanks to having access to a highly efficient synthesizer like $\tool$, we are able to conduct this experiment \emph{using a realistic time limit} to check if a given search space really contains a target program.}
\item 
$S$ contains candidate selectors \emph{sampled} (uniformly at random) from  all candidate selectors of size up to 3. 
\finalized{We vary $|S|$ from 1 to 1500.} 
For each $|S|$, we run $\tool$ with the corresponding $S$ and record the number of benchmarks solved; \finalized{we repeat this 8 times.}
This gives us a finer-grained, more continuous view of the impact of candidate selectors. 
\end{itemize}
While one could certainly craft other heuristics for constructing $S$, we do not consider them, because (i) it is impossible to enumerate all heuristics in the first place, and (ii) human-crafted heuristics from prior work~\cite{chasins2018rousillon,dong2022webrobot} have shown to be not as effective especially for challenging tasks. 
Another note is that, since we use $\tool$ merely as a means to check if the search space has a generalizable program, the synthesis time is not relevant in this experiment. 

For each benchmark with a given $S$, we run $\tool$ \emph{incrementally} until it reaches the end of the action trace, same as in RQ1. 
However, in RQ2, we use {10 seconds} as the timeout per iteration, rather than 1 second from RQ1 --- this allows $\tool$ to  \emph{exhaustively} search the program space such that we can more confidently conclude on its expressiveness. 
If the final synthesized program is intended, we count it as solved (meaning the corresponding search space is expressive enough).

\evalmybox{\textbf{RQ2 take-away}:

In general, an expressive search space should consider a large number of selectors (at least at the level of hundreds) for each webpage. 
}

\begin{figure}[!t]
\centering
\begin{minipage}[b]{.9\linewidth}
\centering
\includegraphics[width=\linewidth]{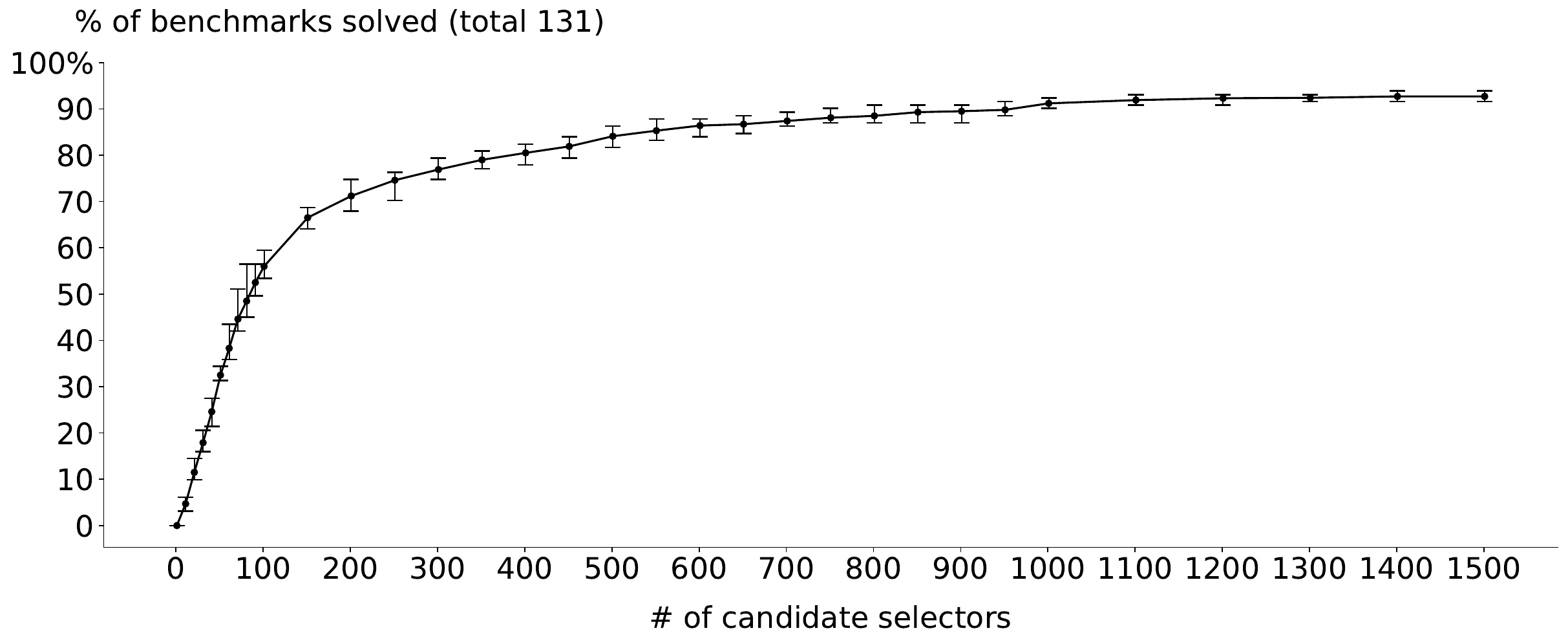}
\end{minipage}
\caption{RQ2 results. Given X candidate selectors (sampled uniformly at random from all candidate selectors of size up to 3), we have Y benchmarks whose search space is confirmed to contain an intended program.} 
\label{fig:eval:RQ2-random-results}
\end{figure}

\newpara{Results.}
If using only the full XPath expressions from the input action trace, $\tool$ manages to solve \finalized{46} benchmarks (out of \finalized{131} total), while terminating on the remaining \finalized{85} without synthesizing an intended program. 
That is, the search spaces of the \finalized{85} benchmarks are exhausted before timeout. 
This confirms that full XPath expressions typically do not generalize, and we need to consider more candidate selectors in order to solve more tasks.
If we \emph{additionally} include all candidate selectors of size 1, the number of solved benchmarks bumps up to \finalized{87}, which is \finalized{66\%} of all benchmarks. $\tool$ does not reach timeout for any of the rest, indicating the need to consider more selectors. 
Further including all candidate selectors of size 2 allows $\tool$ to solve \finalized{122} benchmarks, which is quite close to our results in RQ1. 
Again, no timeout is observed on the remaining benchmarks.
We note that the selector space at this point is already quite large: on average, we have \finalized{305} selectors of size up to two (in our grammar) per DOM across our benchmarks; the median and max are \finalized{321} and \finalized{385} respectively. 
In other words, using a simple size-based heuristic, we necessarily need to consider multiple hundreds of candidate selectors per DOM, in order to automate a decent number of tasks. 
Finally, when considering all candidate selectors of size up to 3 (same as in RQ1), \finalized{125} benchmarks are confirmed to admit intended programs; no timeouts observed. 
The remaining \finalized{6} benchmarks, upon manual inspection, cannot be solved by $\tool$'s language (to our best knowledge). 
While encouraging, this high expressiveness comes at the cost of searching among an average of \finalized{7627} selectors per DOM, with the median and max being \finalized{8066} and \finalized{9649}, across all our benchmarks.


Figure~\ref{fig:eval:RQ2-random-results} presents a finer-grained view of how the expressiveness increases when we range the number of candidate selectors from 0 to 1500 using a small increment. 
The x-axis is the number of selectors randomly sampled from the universe of all candidate selectors of size up to three. 
The $y$-axis is the percentage of benchmarks (out of all \finalized{131} benchmarks) solved by $\tool$; that is, their corresponding search spaces are confirmed to contain at least one desired program. 
For each $x$, the figure gives the \finalized{max, min, and mean} of the percentage of solved benchmarks across \finalized{8} runs. 
In total, there are only \finalized{6} benchmarks for which $\tool$ times out without returning an intended program. As also mentioned earlier, this is due to the limitation of the web automation language, rather than $\tool$ not being able to exhaust the search space. 
Therefore, we believe Figure~\ref{fig:eval:RQ2-random-results} precisely describes \emph{all} benchmarks whose program spaces contain a desired program. 


Let us inspect Figure~\ref{fig:eval:RQ2-random-results} more closely. 
First, $y$ grows pretty quickly when $x$ goes up from 0 to 100. 
At $x=100$, across \finalized{8} runs, an average of \finalized{73} benchmarks are solved, while the max and min are \finalized{78} and \finalized{70} respectively. 
There are \finalized{41} benchmarks not solved in any of the runs: notably, for all these \finalized{41} benchmarks, $\tool$ terminates before timeout \emph{without} returning a generalizable program. 
This confirms that with (only) 100 (randomly sampled) selectors, the corresponding search spaces of these \finalized{41} benchmarks do not contain any generalizable programs. 
Furthermore, \finalized{62} of the \finalized{90} solved benchmarks (in at least one run, using 100 selectors) are from prior work. 
Our observation is that these \finalized{62} benchmarks are relatively ``easier'' compared to our newly curated tasks: their task logics are relatively simpler and their solutions use relatively smaller selectors. 


On the other hand, the growth significantly slows down after $x=100$. 
We observe a pretty long tail of benchmarks that require multiple hundreds, or even more than a thousand, of selectors to be solved. 
For instance, increasing $x$ from \finalized{100 to 200} only grows $y$ from \finalized{56\% to 71\%} (in terms of average). 
In order to have another \finalized{15\%} bump, we need at least \finalized{400} selectors. 
Notably, benchmarks solved in the $[100, 1000]$ range mostly come from our new tasks: these problems have to be solved with complex selectors chosen from a larger space which are used in complex nested loops.


Finally, we seem to reach a plateau after $x=1400$, after which point increasing the number of selectors does not seem to help grow the number of solved benchmarks anymore. The maximum number of solved benchmarks we observed in this experiment (based on random sampling selectors) is \finalized{123 (i.e., 93.9\% out of 131 total).}
(In the previous experiment that includes all candidate selectors based on their size, we observed \finalized{125} benchmarks solved when using all selectors up to size 3.)

\subsection{RQ3: How Does \textsc{Arborist} Scale against Number of Candidate Selectors?}
\label{sec:eval:scalability}

Following up RQ2, one may wonder how $\tool$ would scale with more candidate selectors than those considered in RQ1 and RQ2. 
In this experiment, we \emph{stress test} $\tool$'s search algorithm given a very large number of candidate selectors.

\newpara{Setup.}
We define \emph{search efficiency} as the amount of time for a search algorithm to \emph{exhaust} a given search space. 
For $\tool$, this means: given a benchmark with action trace $\actiontrace$ and given candidate selectors for each $\action_i \in \actiontrace$, how long it takes for the FTA $\fta$ to saturate (i.e., no more new programs can be found) \emph{before timeout}. 
We choose to use the exhaustion time, rather than the time to discover a generalizable program (i.e., synthesis time), because the synthesis time oftentimes depends on the particular search order used while the exhaustion time is more stable. 
Also note that the exhaustion time does not reflect the synthesis time; the latter tends to be much shorter. 
This experiment does not evaluate $\tool$'s synthesis time; see RQ1 for its synthesis times. 

Specifically, we sample a set $S$ of candidate selectors from a universe of all candidate selectors of size up to four. 
This universe is extremely large, with an average of \finalized{139,886} selectors per DOM, with the median and max being \finalized{135,856} and \finalized{291,385} respectively. 
For each benchmark, we vary $|S|$ from \finalized{1000} to \finalized{10,000}. 
Given each $|S|$, we run $\tool$ incrementally on each benchmark, until reaching the end of its action trace. We use the same \finalized{10-second} time per iteration as in RQ2. 
However, in this experiment, we count the number of benchmarks that are \emph{exhausted}. That is, $\tool$ terminates before reaching the timeout for every iteration. 
Given $|S|$, we run $\tool$ on each benchmark \finalized{5} times, and record the \finalized{max, min, and mean} number of exhausted benchmarks across all \finalized{5} runs. 
In addition, for each exhausted benchmark, we record its max exhaustion time among all iterations. 
We also report the distribution of these times across all exhausted benchmarks and across all runs, as a way to quantify $\tool$'s search efficiency.

\evalmybox{\textbf{RQ3 take-away}:

$\tool$ can search very efficiently: in particular, it can exhaust program spaces that consider multiple thousands of selectors within at most a few seconds. 

}

\begin{figure}
\centering
\begin{minipage}[b]{.46\linewidth}
\hspace{-10pt}
\centering
\includegraphics[height=4.05cm]{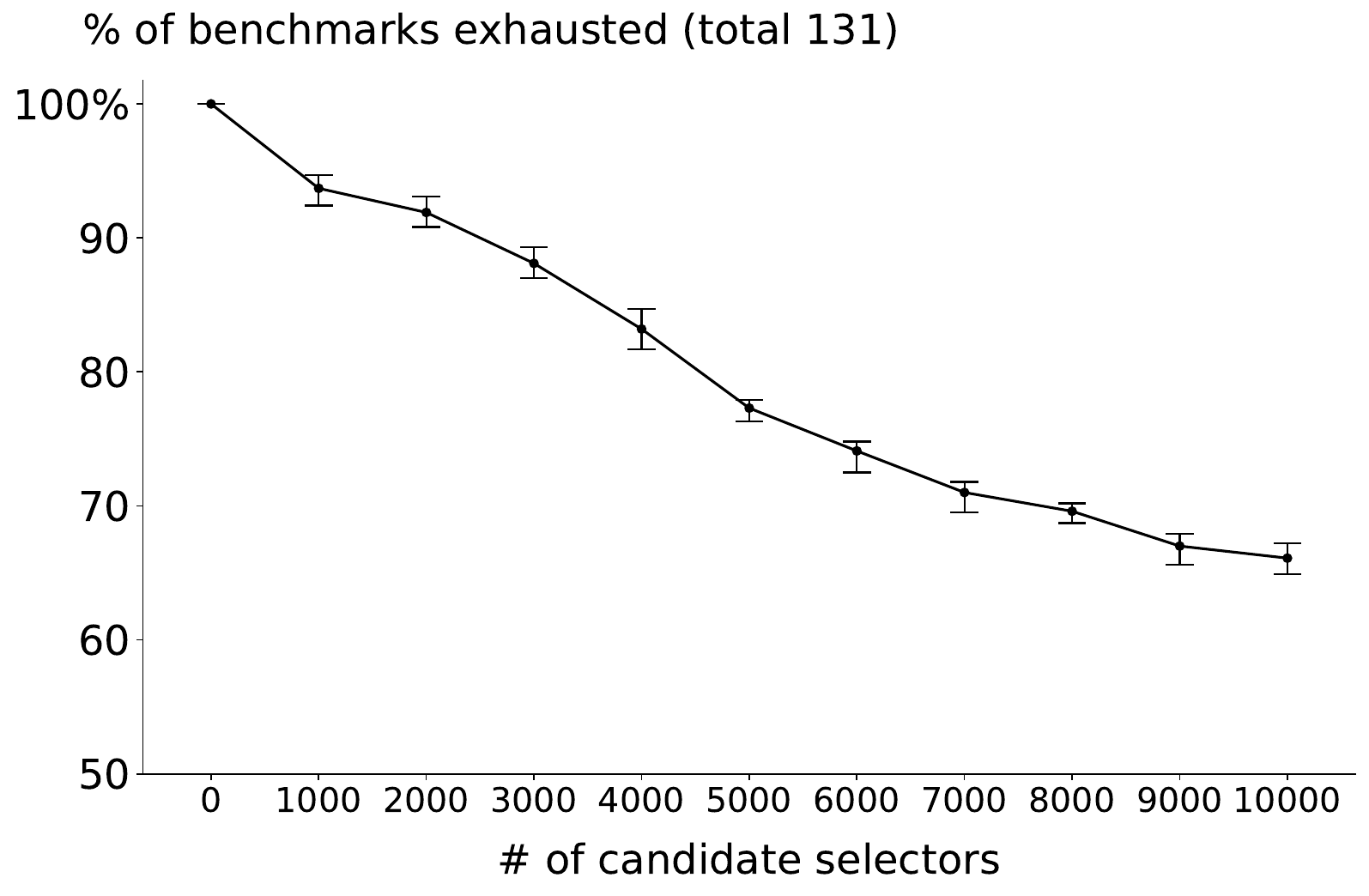}
\vspace{-5pt}
\caption*{(a) Percentage of exhausted benchmarks vs. number of sampled candidate selectors.}
\end{minipage}
\hspace{5pt}
\begin{minipage}[b]{.5\linewidth}
\centering
\hspace{-10pt}
\includegraphics[height=4cm]{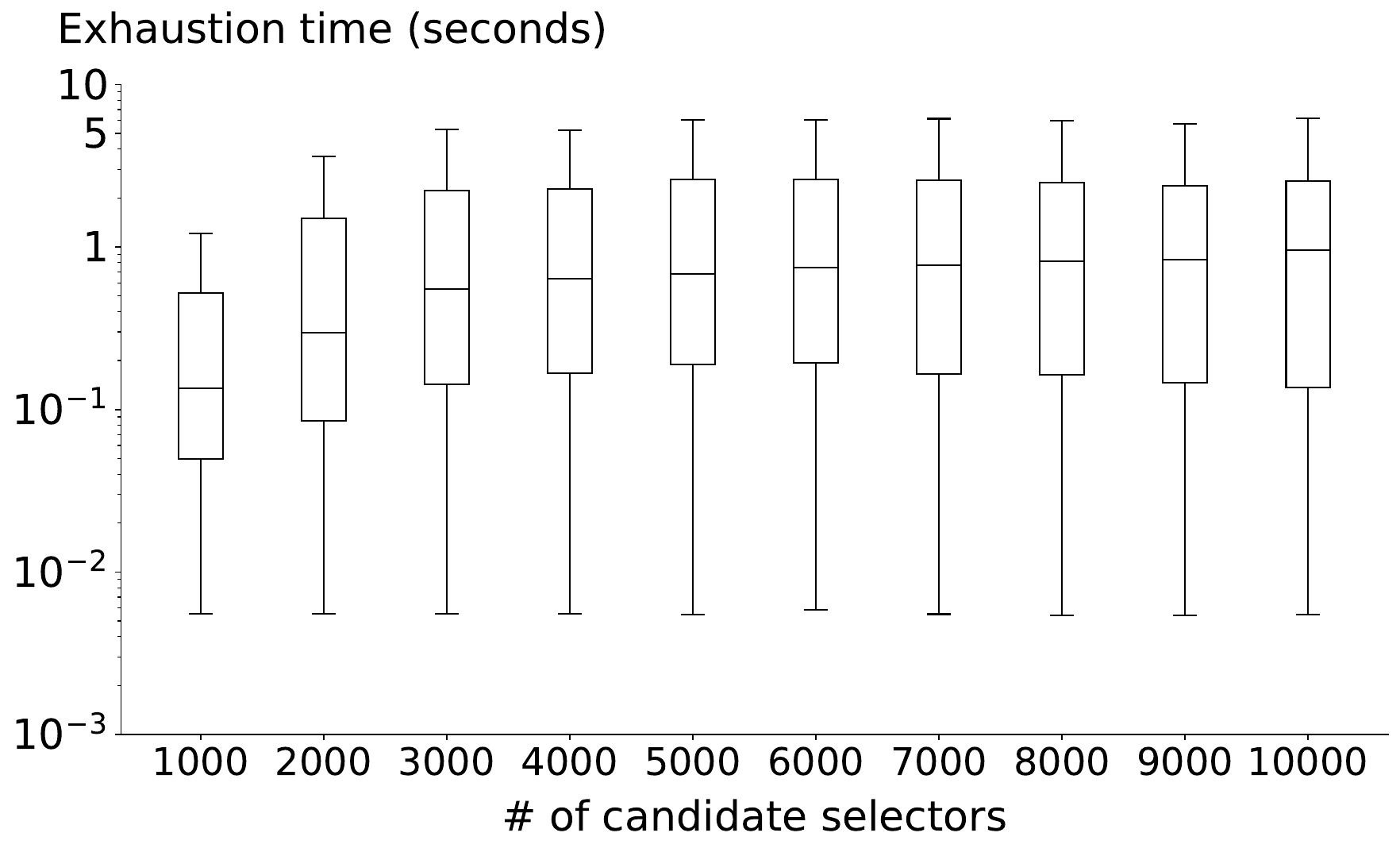}
\vspace{-5pt}
\caption*{(b) Distribution of exhaustion times against each number of sampled candidate selectors.} 
\end{minipage}
\vspace{-5pt}
\caption{RQ3 results. Selectors are sampled uniformly at random from all candidate selectors of size up to four. $\tool$ exhausts a benchmark's program space if it terminates before the timeout (i.e., 10 seconds).}
\label{fig:eval:RQ3-results}
\end{figure}

\newpara{Results.}
Figure~\ref{fig:eval:RQ3-results}(a) shows for each $|S|$, how many benchmarks $\tool$ can successfully exhaust in 10 seconds: since we have multiple runs for each $|S|$, we report \finalized{the max, median, and min} across all runs. 
Figure~\ref{fig:eval:RQ3-results}(b) presents the distribution of exhaustion times --- 
\finalized{same as in RQ1, we also report quartile statistics here in RQ3} 
--- across all exhausted benchmarks for each $|S|$. 
The key take-away message is clear: $\tool$ scales quite well as the number of selectors is increased.
For example, with 1,000 candidate selectors, $\tool$ is able to exhaust the program space for \finalized{95\%} of all \finalized{131} benchmarks, 
with a median exhaustion time of about \finalized{0.1} seconds. 
If we further increase $|S|$ from 1,000 to 5,000, we observe a small drop from \finalized{95\%} to \finalized{79\%} in terms of the percentage of benchmarks that can be exhausted, while the median exhaustion time is under 1 second. 
Finally, looking at the extreme of 10,000 selectors: \finalized{68\%} benchmarks exhausted with a median of \finalized{1}-second exhaustion time. 

While exhaustion is in general fast, it takes even less time to discover a generalizable program, as also mentioned earlier. 
\finalized{For example, among those 68\% (i.e., 89) exhausted benchmarks, $\tool$ can discover an intended program for 61 within 1 second (and 43 under 0.5 seconds).}
In contrast, $\webrobot$'s enumeration-based algorithm cannot exhaust more than \finalized{10} benchmarks (using the same 10-second timeout), even if fed with only \finalized{100} selectors. 
Furthermore, using 10,000 selectors, $\webrobot$'s median solving time (i.e., returning an intended program) is \finalized{10} seconds. 
These data points again highlight that $\tool$'s underlying search algorithm is highly efficient. 


\subsection{RQ4: Ablation Studies}
\label{sec:eval:ablation}


\newpara{Impact of observational equivalence.}
We consider a variant of $\tool$ with the OE capability disabled.
In other words, this ablation has to enumerate loop bodies (which use loop variables), and does not allow sharing across FTA states that correspond to loop bodies. 
We evaluate this variant under the RQ1 setup: it solves (i.e., generates an intended program for) \finalized{56} benchmarks, among which the median synthesis time is \finalized{1} second.
In contrast, $\tool$ solves \finalized{123} benchmarks with a median running time of \finalized{0.02 seconds} across those solved.
This again highlights the importance of OE for speeding up the search.

\newpara{Impact of incremental FTA construction.}
This ablation builds the FTA $\fta$ \emph{from scratch} given new input traces, without reusing previous FTAs. That is, the incremental FTA construction optimization in Section~\ref{sec:alg:incremental} is disabled. 
We run this ablation using the RQ1 setup. 
It solves (i.e., returns an intended program for) \finalized{105} benchmarks --- out of these solved benchmarks, \finalized{40} reach the 1-second timeout and the median running time is \finalized{0.3} seconds. 
$\tool$ solves \finalized{18} more benchmarks with only \finalized{3} timeouts in total and using a significantly less median time of \finalized{0.02} seconds.

\subsection{Case Study: Large Language Models}
\label{sec:llm}

\revision{Given the recent advances in large language models (LLMs) and exploding interests in applying them for program synthesis, we conduct a case study where we use LLMs to generate web automation programs from demonstrations. This is mainly a sanity check, and we refer interested readers to the appendix for more details. In summary, our key take-away is that LLMs (in particular, GPT-3.5~\cite{chatgpt-link}) fail to generate semantically correct programs, even for some of the simplest benchmarks. 
The model can produce unstable results, claiming a benchmark is unsolvable in one run while outputting programs in another trial. 
While these results are poor, it is well-known that LLMs are sensitive to the prompting strategy~\cite{si2022prompting}, and there might be a better method to prompt the model which we have not tried. Nevertheless, we believe these results indicate that our benchmarks are quite hard for state-of-the-art LLMs and require further research in relevant areas in order to better solve these problems. }

\section{Related Work}
\label{sec:related}

In this section, we briefly discuss some closely related work.

\newpara{Observational equivalence (OE).}
OE is a very general concept, which states the \emph{{indistinguishability}} between multiple entities based on their \emph{observed implications}. 
\citet{hennessy1980observing} proposed OE to define the semantics of concurrent programs, where two terms are  observationally equivalent whenever they are interchangeable in \emph{all observable contexts.} 
The idea of OE has also been adopted by programming-by-example (PBE)~\cite{albarghouthi2013recursive,udupa2013transit,peleg2020oopsla} to reduce a large search space of programs thereby boosting the synthesis efficiency. 
\revision{Building upon the concept of OE, our work extends OE-based reduction to also programs with local variables.} 

\newpara{Synthesis of programs with local variables.}
\revision{Programs with local variables are evaluated under \emph{non-static} contexts.} 
To our best knowledge, there are no principled approaches to effectively reduce the space of such programs. 
Prior work~\cite{wang2017synthesis,chen2021web,feser2015synthesizing,peleg2020oopsla,chen2020multi,smith2016mapreduce} typically falls back to some form of brute-force enumeration or utilizes domain-specific reasoning to prune the search space of programs with local variables (such as lambda bodies). 
We propose a principled approach --- i.e., lifted interpretation --- to reduce the space of such programs, thereby speeding up the search of them. 

\revision{RESL~\cite{peleg2020oopsla} is especially related to our work: it uses an \emph{extended context} (same as ours) when searching lambdas; however, it does not present a general approach that can reduce such programs. 
It clearly articulated the key problem of applying OE in general: computing reachable contexts and evaluating programs depend on each other. 
Our work presents a new algorithm that computes contexts and evaluates programs \emph{simultaneously}, by constructing the equivalence relation of programs \emph{while} evaluating programs which facilitates the computation of reachable contexts. 
In contrast, RESL utilizes rules (manually provided) to infer reachable contexts for lambda bodies, for a given higher-order sketch. 
For data-dependent functions (like reduce and fold), RESL falls back to enumeration. 
Our paper addresses the ``infeasible hypothesis'' from RESL (see D.3 in its appendix). 
We believe our work also opens up new ways to further study such program synthesis problems.} 

\revision{
\newpara{Lifted Interpretation.}
Our lifted interpretation idea can be viewed as a bidirectional approach: it traverses the grammar \emph{top-down} to generate reachable contexts, during which it builds up programs \emph{bottom-up} given contexts. 
Different from prior work~\cite{phothilimthana2016scaling,gulwani2011synthesizing,lee2021combining} that enumerates programs bidirectionally, we intertwines the enumeration of contexts and programs. 
Rosette~\cite{torlak2014lightweight} is related, in that they also lift the interpretation from concrete programs to symbolic programs (e.g., defined by a program sketch). 
A key distinction is that our work directly performs program synthesis and uses finite tree automata to succinctly encode the program space, instead of reducing the search problem to SMT solving. 
}

\newpara{Finite tree automata (FTAs) for program synthesis.}
During lifted interpretation, programs are  clustered into equivalence classes, succinctly compressed in an FTA. 
Compared to prior work~\cite{wang2017synthesis,wang2017program,wang2018relational,yaghmazadeh2018automated,miltner2022bottom,wang2018learning,handa2020inductive,koppel2022searching}, states in our FTAs encode \emph{context}-output behaviors, rather than \emph{input}-output behaviors of programs. 
The lifted interpretation idea is not tied to FTAs: our algorithm is developed using FTAs, but we believe other data structures (such as VSAs~\cite{gulwani2011automating} or e-graphs~\cite{willsey2021egg}), or an enumeration-based approach~\cite{peleg2020oopsla} can also be leveraged.

\newpara{Program synthesis for web automation.}
Our instantiation presents a new program synthesis algorithm for web automation. This is an important domain with a long line of work~\cite{chasins2018rousillon,dong2022webrobot,chen2022semanticon,miwa,dilogics,barman2016ringer,chasins2015browser,leshed2008coscripter,lin2009end,fischer2021diy,little2007koala} in both human-computer interaction and programming languages. 
The most related work is $\webrobot$~\cite{dong2022webrobot}: building upon its trace semantics, we further develop a novel synthesis algorithm that can automate a significantly broader range of more challenging tasks much more efficiently.

\newpara{Programming-by-demonstration (PBD).}
In particular, our algorithm is a form of programming-by-demonstration that synthesizes programs from a user-demonstrated trace of actions. 
Different from prior PBD work~\cite{chasins2018rousillon,dong2022webrobot,mo1990learning,lieberman1993tinker,lau2003programming} that is based on either brute-force enumeration or heuristic search of programs, $\tool$ uses observational equivalence to reduce the search space and leverages finite tree automata to succinctly represent all equivalence classes of programs.



\revision{\section{Conclusion}
\label{sec:conc}

We proposed \emph{lifted interpretation}, which is a general approach to reduce the space of programs with local variables, thereby accelerating program synthesis. 
We illustrated how lifted interpretation works on a simple functional language, and presented a full-fledged instantiation of it to perform programming-by-demonstration for web automation. 
Evaluation results in the web automation domain show that lifted interpretation allows us to build a synthesizer that significantly outperforms state-of-the-art techniques. 

}

\revision{\begin{acks}                            
We thank our shepherds, Hila Peleg and Nadia Polikarpova, for their extremely valuable feedback. 
We thank the POPL anonymous reviewers for their constructive comments. 
We also thank Anders Miltner, Chenglong Wang, Kasra Ferdowsi, Ningning Xie, Shankara Pailoor, Yuepeng Wang for their feedback on earlier drafts of this work. 
This work was supported by the National Science Foundation under Grant Numbers CCF-2123654 and CCF-2236233. 
\end{acks}

}

\bibliography{main}

\newpage
\appendix
\section{Lifted Interpretation Rules for Example in Section~\ref{sec:overview}}

\begin{figure}[!h]
\small
\centering
\begin{minipage}{.99\linewidth}
\centering
\arraycolsep=3pt\def\arraystretch{1.2}
\begin{centermath}
\centering
\begin{array}{cc}

(1) &  
\irule{
\begin{array}{c}
\ftatransition = \ftatransitionoperator(\ftastate_1, \mydots, \ftastate_n) \ftatransitionarrow \ftastate  \in \ftatransitions \ \ \ 
\evalftatransition{\Gamma}{\ftatransition; \ftatransitions}{\ftastate', \ftatransitions'}
\end{array}
}{
\evalftastate{\Gamma}{\ftastate; \ftatransitions}{\ftastate', \ftatransitions'}
}

\\ \\

(2) & 
\irule{
\begin{array}{c}
\evalftatransition{\Gamma}{\ftastate_1; \ftatransitions}{\ftastate'_1, \ftatransitions'_1} \ \ \ \ 
\Gamma \inputtooutput lst \in \getannot(\ftastate'_1) \ \ \ \ 
\evalftatransition{\Gamma, lst}{\ftastate_2; \ftatransitions}{\ftastate'_2, \ftatransitions'_2} \\ 
v_0 = 0 \ \ \ \ 
\Gamma \big[ acc \mapsto v_{i-1}, elem \mapsto lst[i] \big] \inputtooutput v_i \in \getannot(\ftastate'_2) \ \ \ \ 
i \in [ 1, |lst| ] \\ 
\annot = \getannot(\ftastate) \ \ \ \ 
\annot' = \annot \cup \{ \Gamma \inputtooutput v_{|lst|} \} \ \ \ \ 
\ftastate' = \newstate(\exprogsymbol, \annot')
\end{array}
}{
\evalftatransition{\Gamma}{\foldop(\ftastate_1, \ftastate_2) \ftatransitionarrow \ftastate; \ftatransitions}{\ftastate', \ftatransitions'_1 \cup \ftatransitions'_2 \cup \{ \foldop(\ftastate'_1, \ftastate'_2) \ftatransitionarrow \ftastate' \} }
}

\\  \\ 

(3) & 
\irule{
\begin{array}{c}
v_0 = 0 \ \ \ \ 
\ftastate'_0 = \ftastate  \ \ \ \ 
\ftatransitions'_0 = \ftatransitions \ \ \ \ 
i \in [1, |lst| ] \\ 
\Gamma_{i-1} = \Gamma \big[ acc \mapsto v_{i-1}, elem \mapsto lst[i] \big] \ \ \ \ 
\evalftastate{\Gamma_{i-1}}{\ftastate'_{i-1}; \ftatransitions'_{i-1}}{\ftastate'_i, \ftatransitions'_i} \ \ \ \ 
\Gamma_{i-1} \inputtooutput v_i \in \getannot(\ftastate'_i)
\end{array}
}{
\evalftatransition{\Gamma, lst}{\ftastate; \ftatransitions}{ 
\ftastate'_{|lst|}, \ftatransitions'_{|lst|} }
}

\\ \\ 

(4) & 
\irule{
\Gamma[x] = lst \ \ \ \ 
\annot = \getannot(\ftastate) \ \ \ \ 
\annot' = \annot \cup \{ \Gamma \inputtooutput lst \} \ \ \ \ 
\ftastate' = \newstate( \exlistsymbol, \annot' )
}{
\evalftatransition{\Gamma}{x \ftatransitionarrow \ftastate; \ftatransitions}{ \ftastate', \{ x \ftatransitionarrow \ftastate' \} }
}

\\ \\ 

(5) & 
\irule{
\begin{array}{c}
\evalftastate{\Gamma}{\ftastate_1; \ftatransitions}{\ftastate'_1, \ftatransitions'_1} \ \ \ \ 
\Gamma \inputtooutput v_1 \in \getannot(\ftastate'_1) \ \ \ \ 
\evalftastate{\Gamma}{\ftastate_2; \ftatransitions}{\ftastate'_2, \ftatransitions'_2} \ \ \ \ 
\Gamma \inputtooutput v_2 \in \getannot(\ftastate'_2) \\ 
\annot = \getannot(\ftastate) \ \ \ \ 
\annot' = \annot \cup \{ \Gamma \mapsto v1 + v2 \}  \ \ \ \ 
\ftastate' = \newstate( \exEsymbol, \annot') \\ 
\end{array}
}{
\evalftatransition{\Gamma}{\addop(\ftastate_1, \ftastate_2) \ftatransitionarrow \ftastate; \ftatransitions}{\ftastate', \ftatransitions'_1 \cup \ftatransitions'_2 \cup \{ \addop(\ftastate'_1, \ftastate'_2) \ftatransitionarrow \ftastate' \} }
}

\\ \\ 

(6) & 
\irule{
\begin{array}{c}
\evalftastate{\Gamma}{\ftastate_1; \ftatransitions}{\ftastate'_1, \ftatransitions'_1} \ \ \ \ 
\Gamma \inputtooutput v_1 \in \getannot(\ftastate'_1) \ \ \ \ 
\evalftastate{\Gamma}{\ftastate_2; \ftatransitions}{\ftastate'_2, \ftatransitions'_2} \ \ \ \ 
\Gamma \inputtooutput v_2 \in \getannot(\ftastate'_2) \\ 
\annot = \getannot(\ftastate) \ \ \ \ 
\annot' = \annot \cup \{ \Gamma \mapsto v1 \times v2 \}  \ \ \ \ 
\ftastate' = \newstate( \exEsymbol, \annot') \\ 
\end{array}
}{
\evalftatransition{\Gamma}{\mulop(\ftastate_1, \ftastate_2) \ftatransitionarrow \ftastate; \ftatransitions}{\ftastate', \ftatransitions'_1 \cup \ftatransitions'_2 \cup \{ \mulop(\ftastate'_1, \ftastate'_2) \ftatransitionarrow \ftastate' \} }
}

\\ \\ 

(7) & 
\irule{
\Gamma[acc] = v \ \ \ \ 
\annot = \getannot(\ftastate) \ \ \ \ 
\annot' = \annot \cup \{ \Gamma \inputtooutput v \} \ \ \ \ 
\ftastate' = \newstate( \exEsymbol, \annot' )
}{
\evalftatransition{\Gamma}{acc \ftatransitionarrow \ftastate}{ \ftastate', \{ acc \ftatransitionarrow \ftastate' \} }
}

\\ \\ 

(8) & 
\irule{
\Gamma[elem] = v \ \ \ \ 
\annot = \getannot(\ftastate) \ \ \ \ 
\annot' = \annot \cup \{ \Gamma \inputtooutput v \} \ \ \ \ 
\ftastate' = \newstate( \exEsymbol, \annot' )
}{
\evalftatransition{\Gamma}{elem \ftatransitionarrow \ftastate}{ \ftastate', \{ elem \ftatransitionarrow \ftastate' \} }
}

\\ \\ 

(9) & 
\irule{
\annot = \getannot(\ftastate) \ \ \ \ 
\annot' = \annot \cup \{ \Gamma \inputtooutput 1 \} \ \ \ \ 
\ftastate' = \newstate( \exEsymbol, \annot' )
}{
\evalftatransition{\Gamma}{1 \ftatransitionarrow \ftastate}{ \ftastate', \{ 1 \ftatransitionarrow \ftastate' \} }
}

\\ \\ 

(10) & 
\irule{
\annot = \getannot(\ftastate) \ \ \ \ 
\annot' = \annot \cup \{ \Gamma \inputtooutput 2 \} \ \ \ \ 
\ftastate' = \newstate( \exEsymbol, \annot' )
}{
\evalftatransition{\Gamma}{2 \ftatransitionarrow \ftastate}{ \ftastate', \{ 2 \ftatransitionarrow \ftastate' \} }
}

\end{array}
\end{centermath}
\caption{Lifted interpretation rules for the example from Section~\ref{sec:overview}.}
\label{fig:rebuttal:evaluate-fta-rules}
\end{minipage}
\end{figure}

\newpage
\section{Rules for FTA Evaluation in Section~\ref{sec:alg}}
\label{sec:complete-rules}

\begin{figure}[!h]
\centering
\begin{centermath}
\small 
\centering
\hspace{-5pt}
\begin{array}{cc}

(S1) &  
\irule{
\begin{array}{c}
\ftastate'_0 = \ftastate \ \ \ 
\ftatransitions'_0 = \ftatransitions \ \ \ 
\evalftastate{\inputcontext_i}
{\ftastate'_{i-1}; \ftatransitions'_{i-1}}
{\ftastate'_i, \ftatransitions'_i}  \ \ \ 
i \in [1, n]
\end{array}
}{
\evalftastate{\inputcontext_1, \mydots, \inputcontext_n}{\ftastate; \ftatransitions}{\ftastate'_n, \ftatransitions'_n}
}

\\ \\

(S2) &  
\irule{
\begin{array}{c}
\ftatransition = \ftatransitionoperator(\ftastate_1, \mydots, \ftastate_n) \ftatransitionarrow \ftastate  \in \ftatransitions \ \ \ 
\evalftatransition{\inputcontext}{\ftatransition; \ftatransitions}{\ftastate', \ftatransitions'}
\end{array}
}{
\evalftastate{\inputcontext}{\ftastate; \ftatransitions}{\ftastate', \ftatransitions'}
}

\\ \\ 

(S3) &  
\irule{
\begin{array}{c}
\domtrace = [\dom_1, \mydots, \dom_m] \ \ \ \ 
\eval{\env, \dom_1}{\selectorexpr}{\fullxpath} \\ 
\annot = \getannot(\ftastate) \\  
\annot' = \annot \cup \{ \domtrace, \env \inputcontexttooutput \fullxpath \} \ \ \ \ 
\ftastate' = \newstate(\selectorexprsymbol, \annot')
\end{array}
}{
\evalftatransition{\domtrace, \env}{ \selectorexpr \ftatransitionarrow \ftastate; \ftatransitions}{\ftastate', \{ \selectorexpr \ftatransitionarrow \ftastate' \}}
}

\\  \\ 

(S4) &  
\irule{
\begin{array}{c}
\eval{\env} {\dataexpr}{L} \ \ \ \
\annot = \getannot(\ftastate) \\ 
\annot' = \annot \cup \{ \domtrace, \env \inputcontexttooutput L \} \ \ \ \ 
\ftastate' = \newstate(\dataexprsymbol, \annot')
\end{array}
}{
\evalftatransition{\domtrace, \env}{ \dataexpr \ftatransitionarrow \ftastate; \ftatransitions}{\ftastate', \{ \dataexpr \ftatransitionarrow \ftastate' \}}
}

\\ \\

(S5) &  
\irule{
\begin{array}{c}
\evalftastate{\domtrace, \env}{\ftastate_1; \ftatransitions}{\ftastate'_1, \ftatransitions'_1} \ \ \ \ 
\domtrace, \env \inputcontexttooutput \fullxpath  \in \getannot(\ftastate'_1) \\ 
\annot = \getannot(\ftastate) \ \ \ \ 
\domtrace = [\dom_1, \mydots, \dom_m] \ \ \ \ 
\actiontrace' = [ \clickstatement(\fullxpath) ]
\\ 
\annot' = \annot \cup \{ \domtrace, \env \inputcontexttooutput \actiontrace', [\dom_2, \mydots, \dom_m] \} \ \ \ \ 
\ftastate' = \newstate(\expressionsymbol, \annot')
\end{array}
}{
\evalftatransition{\domtrace, \env}{ \clickstatement(\ftastate_1) \ftatransitionarrow \ftastate; \ftatransitions}{\ftastate', \ftatransitions'_1 \cup \{ \clickstatement(\ftastate'_1) \ftatransitionarrow \ftastate' \} }
}

\\ \\ 

(S6) &  
\irule{
\begin{array}{c}
\evalftastate{\domtrace, \env}{\ftastate_1; \ftatransitions}{\ftastate'_1, \ftatransitions'_1} \ \ \ \ 
\domtrace, \env \inputcontexttooutput \actiontrace'_1, \domtrace'_1  \in \getannot(\ftastate'_1) \\ 
\evalftastate{\domtrace'_1, \env}{\ftastate_2; \ftatransitions}{\ftastate'_2, \ftatransitions'_2} \ \ \ \ 
\domtrace'_1, \env \inputcontexttooutput \actiontrace'_2, \domtrace'_2  \in \getannot(\ftastate'_2) \\ 
\annot = \getannot(\ftastate) \ \ \ 
\annot' = \annot \cup \{ \domtrace, \env \inputcontexttooutput \actiontrace'_1 \traceconcat \actiontrace'_2, \domtrace'_2 \} \ \ \ 
\ftastate' = \newstate(\programsymbol, \annot')
\end{array}
}{
\evalftatransition{\domtrace, \env}{\seqprog(\ftastate_1, \ftastate_2) \ftatransitionarrow \ftastate; \ftatransitions}{\ftastate', \ftatransitions'_1 \cup \ftatransitions'_2 \cup \{ \seqprog(\ftastate'_1, \ftastate'_2) \ftatransitionarrow \ftastate' \}}
}

\\ \\

(S7) &  
\irule{
\begin{array}{c}
\annot = \getannot(\ftastate) \\ 
\annot' = \annot \cup \{ \domtrace, \env \inputcontexttooutput \emptytrace, \domtrace \} \\  
\ftastate' = \newstate(\programsymbol, \annot)
\end{array}
}{
\evalftatransition{\domtrace, \env}{\skipprog \ftatransitionarrow \ftastate; \ftatransitions}{\ftastate', \{ \skipprog \ftatransitionarrow \ftastate' \}}
}

\\ \\


(S8) &  
\irule{
\begin{array}{c}
\ \subfta(\ftastate, \ftatransitions) = 
 \big( \{ \ftastate \}, \ftatransitions' \big) \\
\end{array}
}{
\evalftatransition{\emptytrace, \env}{\ftastate; \ftatransitions}{\ftastate, \ftatransitions'}
}

\end{array}
\end{centermath}
\caption{$\textsc{EvaluateFTA}$ rules for loop-free statements in Section~\ref{sec:alg}.}
\end{figure}

\begin{figure}[!h]
\centering
\begin{centermath}
\small 
\centering
\begin{array}{cc}

(S9) & 
\irule{
\begin{array}{c}
\evalftastate{\domtrace, \env}{\ftastate_1; \ftatransitions}{ \ftastate'_1, \ftatransitions'_1 } \ \ \ \
\domtrace, \env \inputcontexttooutput L \in \getannot(\ftastate'_1)
\\ 
\evalftastate{\domtrace, \env, L}{\ftastate_2; \ftatransitions}{\ftastate'_2, \ftatransitions'_2} 
\ \ \ \ 
\domtrace, \env \big[ \datavar \bindsto L[i] \big] \inputcontexttooutput \actiontrace'_{i}, \domtrace'_{i} \in \getannot(\ftastate'_2)   
\ \ \ \ 
i \in [1, |L|]
\\ 
\annot = \getannot(\ftastate) 
\ \ \ \ 
\annot' = \annot \cup \{ \domtrace, \env \inputcontexttooutput \actiontrace'_1 \traceconcat \mydots \traceconcat \actiontrace'_{|L|},  \domtrace'_{|L|} \} 
\ \ \ \ 
\ftastate' = \newstate(\programsymbol, \annot')
\end{array}
}{
\evalftatransition{\domtrace, \env}{\fordataloop(\ftastate_1, \ftastate_2) \ftatransitionarrow \ftastate; \ftatransitions}{\ftastate', \ftatransitions'_1 \cup \ftatransitions'_2 \cup \{  \fordataloop(\ftastate'_1, \ftastate'_2) \ftatransitionarrow \ftastate'  \}}
}

\\ \\ 

(S10) & 
\irule{
\begin{array}{c}
\ftastate'_0 = \ftastate \ \ \ \ 
\ftatransitions'_0 = \ftatransitions \ \ \ \ 
\domtrace'_0 = \domtrace 
\\ 
\evalftastate{\domtrace'_{i-1}, \env[\datavar \bindsto \datavalue_i]}{\ftastate'_{i-1}; \ftatransitions'_{i-1}}{\ftastate'_i, \ftatransitions'_i} \ \ \ \ 
\domtrace'_{i-1}, \env[\datavar \bindsto \datavalue_i] \inputcontexttooutput \actiontrace'_i, \domtrace'_i \in \getannot(\ftastate'_i) 
\end{array}
}{
\evalftastate{\domtrace, \env, [\datavalue_1, \mydots, \datavalue_n]}{\ftastate; \ftatransitions}{\ftastate'_n, \ftatransitions'_n}
}

\\ \\

(S11) & 
\irule{
\begin{array}{c}
\evalftastate{\domtrace, \env}{\ftastate_1; \ftatransitions}{ \ftastate'_1, \ftatransitions'_1 } \ \ \ \
\domtrace, \env \inputcontexttooutput \selectorexpr/\predicate[i] \in \getannot(\ftastate'_1)
\\ 
\evalftastate{\domtrace, \env, \selectorexpr/\predicate[i]}{\ftastate_2; \ftatransitions}{\ftastate'_2, \ftatransitions'_2}  \ \ \ \ 
\domtrace'_i = \domtrace
\\ 
\domtrace'_{i+k}, \env \big[ \selectorvar \bindsto \selectorexpr/\predicate[i + k] \big] \inputcontexttooutput \actiontrace'_{i + k}, \domtrace'_{i + k} \in \getannot(\ftastate'_2)  
\ \ \ \ 
k \in [0, l-1]
\\ 
\domtrace'_{i+l}, \env \big[ \selectorvar \bindsto \selectorexpr/\predicate[i + l] \big] \inputcontexttooutput \actiontrace'_{i + l}, \domtrace'_{i + l} \not\in \getannot(\ftastate'_2)
\\ 
\annot = \getannot(\ftastate) 
\ \ \ \ 
\annot' = \annot \cup \{ \domtrace, \env \inputcontexttooutput \actiontrace'_{i} \traceconcat \mydots \traceconcat \actiontrace'_{i + l-1},  \domtrace'_{i+l-1} \} 
\ \ \ \ 
\ftastate' = \newstate(\programsymbol, \annot')
\end{array}
}{
\evalftatransition{\domtrace, \env}{\forselectorsloop(\ftastate_1, \ftastate_2) \ftatransitionarrow \ftastate; \ftatransitions}{\ftastate', \ftatransitions'_1 \cup \ftatransitions'_2 \cup \{  \forselectorsloop(\ftastate'_1, \ftastate'_2) \ftatransitionarrow \ftastate'  \}}
}

\\ \\

(S12) & 
\irule{
\begin{array}{c}
\evalftastate{\domtrace, \env\big[ \selectorvar \bindsto \selectorexpr[i] \big]}{\ftastate; \ftatransitions}{\ftastate', \ftatransitions'} \ \ \ \ 
\domtrace, \env \big[ \selectorvar \bindsto \selectorexpr[i] \big] \inputcontexttooutput \actiontrace', \domtrace' \in \getannot(\ftastate')
\\
\evalftastate{\domtrace', \env, \fullxpath/\predicate[i+1]}{\ftastate'; \ftatransitions'}{\ftastate', \ftatransitions'}

\end{array}
}{
\evalftastate{\domtrace, \env, \selectorexpr/\predicate[i]}{\ftastate; \ftatransitions}{\ftastate', \ftatransitions'}
}

\\ \\

(S13) & 
\irule{
\subfta(\ftastate, \ftatransitions) = ( \{ \ftastate \}, \ftatransitions'  )
}{
\evalftastate{\emptytrace, \env, \selectorexpr/\predicate[i]}{\ftastate; \ftatransitions}{\ftastate, \ftatransitions'}

}

\\ \\

(S14) & 
\irule{
\domtrace = [\dom_1, \mydots, \dom_m] \ \ \ \
\selectorexpr/\predicate[i] \textnormal{ is not a valid selector on } \dom_1
\ \ \ \ 
\subfta(\ftastate, \ftatransitions) = ( \{ \ftastate \}, \ftatransitions') 
}{
\evalftastate{\domtrace, \env, \selectorexpr/\predicate[i]}{\ftastate; \ftatransitions}{\ftastate, \ftatransitions'}
}

\\ \\

(S15) & 
\irule{
\begin{array}{c}
\evalftastate {\emptytrace, \domtrace, \env} {(\ftastate_1, \ftastate_2)} {\ftastate'_1, \ftastate'_2, \ftatransitions', \actiontrace', \domtrace'}
\\
\annot = \getannot(\ftastate) \ \ \ 
\annot' = \annot \cup \{ \domtrace, \env \inputcontexttooutput \actiontrace', \domtrace' \} \ \ \ 
\ftastate' = \newstate(\programsymbol, \annot')
\end{array}
}{
\evalftatransition{\domtrace, \env}{\whileloop(true, \ftastate_1, \ftastate_2) \ftatransitionarrow \ftastate; \ftatransitions}{\ftastate', \ftatransitions' \cup \{  \whileloop(true, \ftastate'_1, \ftastate'_2) \ftatransitionarrow \ftastate'  \}}
}

\\ \\

(S16) & 
\irule{
\begin{array}{c}
\evalftastate{\domtrace, \env}{\ftastate_1; \ftatransitions}{ \ftastate'_1, \ftatransitions'_1 } \ \ \ \
\domtrace, \env \inputcontexttooutput \actiontrace'_{1},  \domtrace'_{1} \in \getannot(\ftastate'_1)
\\
\evalftastate{\domtrace'_1, \env}{\ftastate_2; \ftatransitions}{ \ftastate'_2, \ftatransitions'_2 } \ \ \ \
\domtrace'_1, \env \inputcontexttooutput \fullxpath \in \getannot(\ftastate'_2)\ \ \ \
\actiontrace'_2 = [ \clickstatement(\fullxpath) ]
\\
\domtrace'_1 = [\dom_1, \dom_2, \mydots, \dom_m] \ \ \ \
\domtrace' = [\dom_2, \mydots, \dom_m] \ \ \ \
\actiontrace' = \actiontrace \traceconcat \actiontrace'_1 \traceconcat \actiontrace'_2
\\
\ftatransitions' = \ftatransitions'_1 \cup \ftatransitions'_2\ \ \ \
\evalftastate{\actiontrace', \domtrace', \env}{ (\ftastate'_1, \ftastate'_2); \ftatransitions' } { \ftastate''_1, \ftastate''_2, \ftatransitions'', \actiontrace'', \domtrace'' }
\end{array}
}{
\evalftatransition{\actiontrace, \domtrace, \env}{(\ftastate_1, \ftastate_2); \ftatransitions}{\ftastate''_1, \ftastate''_2, \ftatransitions'', \actiontrace'', \domtrace''}
}

\\ \\

(S17) & 
\irule{}{
\evalftatransition{\actiontrace, \emptytrace, \env}{(\ftastate_1, \ftastate_2); \ftatransitions}{\ftastate_1, \ftastate_2, \ftatransitions, \actiontrace, \emptytrace}
}

\\ \\

(S18) & 
\irule{
\begin{array}{c}
\evalftastate{\domtrace, \env}{\ftastate_1; \ftatransitions}{ \ftastate'_1, \ftatransitions'_1 }\ \ \ \
\domtrace, \env \inputcontexttooutput \actiontrace'_{1},  \domtrace'_{1} \in \getannot(\ftastate'_1)
\\
\domtrace'_1 = \emptytrace \ \ \ \
\subfta(\ftastate_2, \ftatransitions) =  \big( \{ \ftastate_2 \}, \ftatransitions_2' \big)
\end{array}
}{
\evalftatransition{\actiontrace, \domtrace, \env}{(\ftastate_1, \ftastate_2); \ftatransitions}{\ftastate_1, \ftastate_2, \ftatransitions'_1 \cup \ftatransitions'_2, \actiontrace \traceconcat \actiontrace'_1, \emptytrace}
}

\\ \\

(S19) & 
\irule{
\begin{array}{c}
\evalftastate{\domtrace, \env}{\ftastate_1; \ftatransitions}{ \ftastate'_1, \ftatransitions'_1 }\ \ \ \
\domtrace, \env \inputcontexttooutput \actiontrace'_{1},  \domtrace'_{1} \in \getannot(\ftastate'_1)
\\
\evalftastate{\domtrace'_1, \env}{\ftastate_2; \ftatransitions}{ \ftastate'_2, \ftatransitions'_2 } \ \ \ \
\domtrace'_1, \env \inputcontexttooutput \fullxpath \in \getannot(\ftastate'_2)
\\
\domtrace'_1 = [\dom_1, \mydots, \dom_m] \ \ \ \
\fullxpath \textnormal{ is not a valid selector on } \dom_1 \ \ \ \
\subfta(\ftastate_2, \ftatransitions) =  \big( \{ \ftastate_2 \}, \ftatransitions_2' \big)
\end{array}
}{
\evalftatransition{\actiontrace, \domtrace, \env}{(\ftastate_1, \ftastate_2); \ftatransitions}{\ftastate_1, \ftastate_2, \ftatransitions'_1 \cup \ftatransitions'_2, \actiontrace \traceconcat \actiontrace'_1, \domtrace'_1}
}

\end{array}
\end{centermath}
\caption{$\textsc{EvaluateFTA}$ rules for loopy statements in Section~\ref{sec:alg}.}
\end{figure}

\newpage
\section{Theorems and Proofs}
In this section, we prove the soundness and completeness theorems of our algorithm given a fixed DOM trace $\domtrace$ and an action trace $\actiontrace$. 

For simplicity, we use $\seqprog(S_1, \mydots, S_k)$ to denote $\seqprog(S_1, \mydots, \seqprog(S_k, \emph{skip}))$ and use $\domtrace_{[i,j]}$ to denote a slice of the DOM trace $[\dom_i,\mydots, \dom_j]$ (and same for action traces $\actiontrace$).
We also define $P_1\listconcat P_2$, the concatenation of two programs $P_1=\seqprog(S_1, \mydots, S_{k_1})$ and $P_2=\seqprog(S'_1, \mydots, S'_{k_2})$,
 to be $\seqprog(S_1, \mydots, S_{k_1},S'_1,\mydots, S_{k_2})$.

\subsection{Soundness and completeness of components}

In this section, we prove several lemmas about \textsc{EvaluateFTA} and \textsc{MeregFTAs}, which are used in both the soundness theorem and the completeness theorem.

First, we show that the footprints correctly capture the semantics of programs represented by $\fta$. To do this, we define a concept called ``well-annotatedness''.

\begin{definition}
    We call an FTA $\fta$ \textit{well-annotated} if $\fta$ satisfies the condition that, for any state $\ftastate = (s, \{C_1\bindsto O_1, \mydots, C_n\bindsto O_n\})$ in $\fta$ that represents a program $\program$ (that is, $\program \in \ftalang(\subfta(\ftastate, \ftatransitions))$ where $\ftatransitions$ is $\fta$'s transitions), we have $\eval{C_i}{\program }{O_i}$.
\end{definition}

Both \textsc{EvaluateFTA} and \textsc{MergeFTAs} are sound with respect to well-annotatedness; we prove their soundness as lemmas below. These lemmas are later used to establish the soundness of our main algorithm, since they imply $\fta$, which is constructed using \textsc{EvaluateFTA} and \textsc{MergeFTAs}, is well-annotated.

\begin{lemma}[Soundness of \textsc{EvaluateFTA}]
    \label{proof:eval}
    Given an unannotated FTA $\fta_s$ and a list of contexts $[C_1, \mydots, C_n]$, $\textup{\textsc{EvaluateFTA}}(\fta_s, [C_1, \mydots, C_n])$ returns a well-annotated FTA.
\end{lemma}

\begin{proof}
    The judgement $\evalftastate{C_1,\mydots,C_n}{q;\transitions}{q',\transitions'}$
    ``folds'' over the judgement over a single context,
    so we can inductively prove the lemma by inducting on the list $C_1,\mydots,C_n$.
    The base case holds since any unannotated FTA is a well-annotated FTA.
    Moreover, it can be proved by induction on the rules
    that the judgement over a single context, $\evalftastate{C}{q;\transitions}{q',\transitions'}$, only adds footprints consistent with the trace semantics. Therefore,
    if $( \{ q \},\transitions)$ is well-annotated,
    $( \{ q' \},\transitions')$ is 
    well-annotated as well, which justifies the inductive case. 
\end{proof}

\begin{lemma}[Soundness and monotonicity of \textsc{MergeFTAs}]
    \label{proof:merge}
    Given two well-annotated FTAs $\fta$ and $\fta_e$ and a state $q$ in $\fta$, $\fta_{\textit{new}}=\textup{\textsc{MergeFTAs}}(\fta, q, \fta_e)$ is a well-annotated FTA satisfying $\ftalang(\fta)\subseteq \ftalang(\fta_{\textit{new}})$.
\end{lemma}

\begin{proof}
    Well-annotatedness of $\fta_\textit{new}$ can be shown by noticing the semantics of \textsc{MergeFTAs} respects the footprints of existing states and the semantics of $\seqprog$. We also have $\ftalang(\fta)\subseteq \ftalang(\fta_\textit{new})$ since \textsc{MergeFTAs} only adds new transitions to $\fta$.
\end{proof}

However, soundness itself does not fully characterize the behavior of \textsc{EvaluateFTA}. For example, both an unannotated FTA and an empty FTA are always well-annotated, but they are not the intended output of \textsc{EvaluateFTA}. Therefore, we additionally prove a completeness theorem of \textsc{EvaluateFTA}, which precisely describes what programs the evaluated FTA should represent and what annotations it should have:

\begin{lemma}[Completeness of \textsc{EvaluateFTA}]
    \label{proof:eval}
    Given an unannotated FTA $\fta_s$ and a list of contexts $[C_1, \mydots, C_n]$, let $\fta_e=\textup{\textsc{EvaluateFTA}}(\fta_s, [C_1, \mydots, C_n])$. We have $\ftalang(\fta_s)=\ftalang(\fta_e)$.
    Moreover, each final state $q\in Q_f$ in $\fta_e$ has contexts $[C_1,\mydots,C_n]$
\end{lemma}

\begin{proof}
We know that 
    $\fta_e$ is the union FTA of each $\fta'_i=({q'_i}, \transitions'_i)$ where $\evalftastate{C}{q_i;\transitions}{q_i',\transitions'}$ and $q_i$ is a final state of $\fta_s$.
    It can be proved that rules over the judgement $\evalftastate{C}{q_i;\transitions}{q_i',\transitions'}$ for a single context $C$ preserves the invariant that $\ftalang(\fta_i)=\ftalang(\fta'_i)$ where $\fta_i=( \{ q_i \},\transitions)$.
    By induction the invariant is also preserved for the judgement over lists of contexts $C_1,\mydots,C_n$,
    so $\ftalang(\fta) = \ftalang(\fta_e)$.

    Moreover, it can be shown that 
    $\evalftastate{C}{q_i;\transitions}{q_i',\transitions'}$ maintains the invariant that 
    \[
    \inputcontexts(q'_i)=\inputcontexts(q_i)\cup\{C\}
    \]
    
    and therefore by induction $\evalftastate{C_1,\mydots,C_n}{q_i;\transitions}{q_i',\transitions'}$
    satisfy 
    \[\inputcontexts(q'_i)=\inputcontexts(q_i)\cup\{C_1,\mydots, C_n\}\] 
    
    which proves the claim as $\fta_s$ is unannotated.
\end{proof}

\subsection{Soundness of the synthesis algorithm}

Now, we are ready prove the soundness of our main algorithm.

\begin{theorem}[Soundness]\label{proof:main:soundness}
    Given input data $I$, input action trace $\actiontrace$, and input DOM trace $\domtrace$,
    $\fta$ contains only programs that reproduce $A$; that is,
    for all $\program\in \ftalang(\fta)$, we have  
    $\eval{\domtrace_{[1,m]}, \env}{ \program }{\actiontrace, \emptytrace}$, where $\env$ is a binding context with only $\inputvar \mapsto \inputdata$, 
\end{theorem}

\begin{proof}
    We show this by showing that $\fta$ is well-annotated, 
    which directly implies the theorem, 
    since the only final state of $\fta$ is associated with the footprint $\eval{\domtrace_{[1,m]}, \env}{ \program }{\actiontrace, \emptytrace}$. 
    
    We next prove $\fta$ is well-annotated by inducting on the while loop.
    \begin{itemize}
        \item Base case: It is straightforward to see that the initial FTA is well-annotated.
        For example, consider the set of programs represented by the final state: $\seqprog(\action_1[\sigma_1], \ldots, \action_m[\sigma_m])$ where $\sigma_i$ is any substitution that maps full XPaths $\absolutexpath$ in $\action_i$ to selectors $\selectorexpr$ satisfying $\eval{\dom_i, \env}{\selectorexpr}{\fullxpath}$. Any such program $\program$ reproduces the trace (i.e., $\eval{\domtrace_{[1,m]}, \env}{ \program }{\actiontrace, []}$).
        \item Inductive steps: given that $\fta$ currently satisfies the invariant, we show that after one iteration of the main loop body (line~\ref{alg:top-level:speculate}-\ref{alg:top-level:merge} of Algorithm~\ref{alg:top-level}), the new $\fta$ (denoted as $\fta_{\textit{new}}$) still satisfies the invariant. To show this, we notice that by the soundness of \textsc{EvaluateFTA} (Lemma~\ref{proof:eval}), $\fta_e$ is well-annotated.
        Moreover, by the inductive hypothesis, $\fta$ is well-annotated as well.
        According to the soundness of \textsc{MergeFTAs} (Lemma~\ref{proof:merge}), 
        $\fta_{\textit{new}}$ is well-annotated.
    \end{itemize}
\end{proof}

\subsection{Completeness of the synthesis algorithm}

The completeness theorem of our synthesis algorithm is slightly trickier to state
Note that our synthesis algorithm does not consider \textit{every} program that generalizes the trace. For example, our algorithm only synthesizes loops that have two full iterations exhibited in trace. 
Moreover, our algorithm only consider programs whose generalizability comes from synthesized loops. 
For example, program $\seqprog(\action_1,\seqprog(\action_1, \emph{skip}))$ generalizes the action trace $[\action_1]$, but our algorithm does not consider it.

To formally define the program space our synthesis algorithm covers, we introduce a concept called {\it rerollability} of an given evaluation.
Specifically, we use two judgements to define the set of rerollable programs, $\eval{\domtrace,\env}{^R \program}{\actiontrace,\domtrace'\ \textsc{Reroll}}$ and 
$\eval{\domtrace,\env}{^G \program}{\actiontrace,\domtrace'\ \textsc{Reroll}}$. The only difference between these two judgements is that the first one requires programs to strictly reproduce the trace, and the second one allows the program to also possibly generalize the trace . 
Specifically, the generalizability judgement allows ``skipping evaluations'' (similar to the standard semantics of our DSL)
\[(\textsc{Emp-Reroll})\ \ \ \ \ \ 
\irule{
}{
\eval{\emptytrace, \env}{^G \program }{\emptytrace, \emptytrace\ \textsc{Reroll}}
}
\]
that the reproducibility judgement does not allow.

Inference rules for rerollability generally follow the semantics of our DSL.
However, there are two key differences:
First, we require any evaluation of loops must 
fully reproduce itself in the first two iterations.
Second, we track whether a subprogram should strictly reproduce or can potentially generalize the given trace.
We show a few representative rules here ($X$ ranges over $\{R,G\}$ in the following rules):

\[
\small 
\begin{array}{cc}
\\
\\
\textsc{(ForData-Reroll)} &
\irule{
\begin{array}{c}
\eval{\Sigma}{\dataexpr}{L} \ \ \ \
|L| > 1 \\
\eval{\domtrace, \env[z\bindsto L[1]]}{^R \program}{\actiontrace_1, \domtrace_1\ \textsc{Reroll}} \\
\eval{\domtrace_1, \env[z\bindsto L[2]]}{^R \program}{\actiontrace_2, \domtrace_2\ \textsc{Reroll}} \\
\eval{\domtrace_{i-1}, \env[z\bindsto L[i]]}{^X\program}{\actiontrace_i, \domtrace_i\  \textsc{Reroll}}\ \ \ \ \forall\; i\in [3,|L|]

\end{array}
}{
\eval{\domtrace, \env}{^X \fordataloop( \dataexpr, \lambda \datavar. \program )}{\bigoplus_{i= 1}^{ |L|}\actiontrace_i, \domtrace_{|L|}\ \textsc{Reroll}}
}

\\
\\
\textsc{(Click-Reroll)}&
\irule{
\domtrace = [ \dom_1, \mydots, \dom_m ] \ \ \ \ 
\eval{\dom_1, \env}{\selectorexpr}{\fullxpath} \ \ \ \ 
\domtrace' = [ \dom_2, \mydots, \dom_m ] 
}{
\eval{\domtrace, \env}{^X\clickexpr(\selectorexpr)}[ {\clickexpr(\fullxpath) ], \domtrace' \ \textsc{Reroll}}
}
\\
\\

(\textsc{Seq-Non-Final-Reroll}) &
\irule{
\begin{array}{c}
P \neq \emph{skip}\\
\eval{\domtrace, \env}{^R \expression}{\actiontrace', \domtrace'\ \textsc{Reroll}} \ \ \ \ 
\eval{\domtrace', \env}{^X\program}{\actiontrace^{\doubleprime}, \domtrace^{\doubleprime}\ \textsc{Reroll}}
\end{array}
}{
\eval{\domtrace, \env}{^X \seqprog( \expression, \program )}{\actiontrace' \listconcat \actiontrace^{\doubleprime}, \domtrace^{\doubleprime}\ \textsc{Reroll}}
}
\\
\\

(\textsc{Seq-Final-Reroll}) &
\irule{
\begin{array}{c}
\eval{\domtrace, \env}{^X \expression}{\actiontrace', \domtrace'\ \textsc{Reroll}}
\end{array}
}{
\eval{\domtrace, \env}{^X \seqprog( \expression, \emph{skip} )}{\actiontrace', \domtrace'\ \textsc{Reroll}}
}

\end{array}
\]

The following lemmas give a sense of the relative restrictiveness of these judgements.

\begin{lemma}
\label{theorem:repro-gen}
    $\eval{\domtrace,\env}{^R}{\actiontrace',\domtrace'\ \textsc{Reroll}}$
    implies
    $\eval{\domtrace,\env}{^G P}
    {\actiontrace',\domtrace'\ \textsc{Reroll}}$.
\end{lemma}

\begin{lemma}
    $\eval{\domtrace,\env}{^G}{\actiontrace',\domtrace'\ \textsc{Reroll}}$
    implies
    $\eval{\domtrace,\env}{P}
    {\actiontrace',\domtrace'}$.
\end{lemma}

\begin{proof}
    Both lemmas can be proved via a straightforward induction.
\end{proof}

Now we have defined the space of rerollable programs we can synthesize, we are ready to show the completeness of our theorem. To prune away programs like $\seqprog(\action_1, \seqprog(\action_1,\emph{skip}))$ for trace $[\action_1]$,
we additionally restrict our completeness theorem to programs whose final expression can produce at least one action on the given trace, as shown below.

\begin{theorem}[Completeness]
\label{proof:main:completeness}
    Given DOM trace $\domtrace$ and environment $\sigma$,
    if the evaluation of
    a program $P=\seqprog(S_1,\mydots, S_k)$ under $\domtrace_{[1,m]}$ and $\env$ is rerollable (and may generalize the given trace) and satisfies that the evaluation of $S_k$ produces at least one action,
    $P$ will be represented by $\fta$; In other words, if
    \[\eval{\domtrace_{[1,m]}, \env}{^R \seqprog(S_1, \mydots, S_{k-1}) }{\actiontrace_{[1,p]}, \domtrace_{[p+1,m]}\ \textup{\textsc{Reroll}}}\]
    and \[\eval{\domtrace_{[p+1,m]}, \env}{^GS_k}{\actiontrace_{[p+1,m]}, \emptytrace\ \textup{\textsc{Reroll}}}\]
    where $p\neq m$,
    then program \(\seqprog(S_1,\mydots, S_k)\in \ftalang(\fta).\)
\end{theorem}

\begin{proof}

To prove the theorem, we prove a generalization of it: If
\begin{equation}
   \eval{\domtrace_{[o,m]},\env}{^R \seqprog(S_1, \mydots, S_{k-1})}{\actiontrace_{[o,p]}, \domtrace_{[p+1,m]}\ \textsc{Reroll}}
\end{equation}
and
\begin{equation}
   \eval{\domtrace_{[p+1,m]},\env}{^P S_k}{\actiontrace_{[p+1,q]}, \domtrace_{[q+1,m]}}\ \textsc{Reroll}
   \end{equation} 
   for some indices $o,p,q$ with $p\neq q$,
then
\[\seqprog(S_1, \mydots, S_k, \action_{j+1}, \mydots \action_m)\in \ftalang(q_o, \fta)\]
for some state $q_o$ of $\fta$.

We prove this generalization by inducting on the nested levels of loops of program $P$.

\begin{itemize}
    \item Base case: All loop-free programs $\program$ satisfying (1) and (2)
    have the shape $\seqprog(a_o[\sigma_o],\mydots,a_q[\sigma_q])$ where $\sigma_i$ is any substitution that maps full XPaths $\absolutexpath$ in $\action_i$ to selectors $\selectorexpr$ satisfying $\eval{\dom_i, \env}{\selectorexpr}{\fullxpath}$.
    By our construction,
    there always exists some state $q$ in
    $\textsc{InitFTA}(\domtrace_{[1,m]},\actiontrace)$ that represents such a program with the corresponding footprints.
    By the monotonicity of \textsc{MergeFTAs}, this state is in $\fta$ and represents such a program as well.
    \item Inductive steps:
    Consider a program $P=\seqprog(S_1,\mydots, S_k)$ where all $S_i$
    has a nested level of $n+1$ with rerollable evaluation 
    $\eval{\domtrace_{[o,m]}, \env}{^R \seqprog(S_1,\mydots, S_{k_1})}{A_{[o,p]}, \domtrace_{[p+1,m]}\ \textsc{Reroll}}$ and $\eval{\domtrace_{[p,m]}, \env}{^G S_k}{A_{[p,q]}, \domtrace_{[q+1,m]}\ \textsc{Reroll}}$.
    Let us only consider the case for $\fordataloop$ and other loop forms can be proved similarly.
    
    To show this statement about $\seqprog(S_1,\mydots, S_k)$, we only need to show it for a single $S$.
    Suppose that $S=\fordataloop( \dataexpr, \lambda \datavar. \program' )$ is an expression from $S_1, \mydots, S_k$ and that $P'$ has the form $\seqprog(S'_1,\ldots, S'_l)$ and a nested level of $n$.
    We will show that if $\eval{\domtrace_{[i,m]},\env}{^G S}{A_{[i, l]}, \domtrace_{[l+1,m]}}$, then $\seqprog(S, a_{l+1},\mydots, a_m)\in\ftalang(q, \fta)$ for some $q$.
    
    By \textsc{ForData-Reroll}, 
    we have
    \[\begin{array}{c}
    \eval{\domtrace_{[i,m]}, \env[z\bindsto L[1]]}{^R \program'}{\actiontrace_{[i, j]}, \domtrace_{[j+1,m]}\ \textsc{Reroll}} \\
\eval{\domtrace_{[j+1,m]}, \env[z\bindsto L[2]]}{^R \program'}{\actiontrace_{[j+1,k]}, \domtrace_{[k+1,m]}\ \textsc{Reroll}}
    \end{array}\]
    for some $j,k$.
Alternatively, we can apply the substitution with $z$ to $P'$ directly.
    \[\begin{array}{c}
    \eval{\domtrace_{[i,m]}, \env}{^R \program'[z\bindsto L[1]]}{\actiontrace_{[i, j]}, \domtrace_{[j+1,m]}\ \textsc{Reroll}} \\
\eval{\domtrace_{[j+1,m]}, \env}{^R \program'[z\bindsto L[2]]}{\actiontrace_{[j+1,k]}, \domtrace_{[k+1,m]}\ \textsc{Reroll}}
    \end{array}\]

It can be proved via a straightforward induction that $\program'_{\listconcat}=\program'[z\bindsto L[1]]\listconcat \program'[z\bindsto L[2]]$ satisfy 
\[
\eval{\domtrace_{[i,m]}, \env}{^R \program'_{\listconcat}}{\actiontrace_{[i, k]}, \domtrace_{[k+1,m]}\ \textsc{Reroll}}.
\]
Moreover, $\program'_{\listconcat}$  has a nested level of $n$, so by Lemma~\ref{theorem:repro-gen} and inductive hypothesis, it follows 
\[\program'_{\listconcat}\listconcat \seqprog(\action_{k+1},\mydots,\action_{m})\in\ftalang(q_i,\fta)\]
for some state $q_0$.
Expanding the definition of $\program'_{\listconcat}$:
\[\begin{array}{rr}
\seqprog(S'_1[z\bindsto L[1]],\ldots, S'_k[z\bindsto L[1]],&\\
S'_1[z\bindsto L[2]],\ldots, S'_k[z\bindsto L[2]],&\\
\action_{k+1},\mydots,\action_{m})&\in\ftalang(q_0,\fta).
\end{array}\]

We just show that there will be two consecutive traces of the loop in the FTA,
so the matching procedure (line~\ref{alg:top-level:matching-transitions} of algorithm~\ref{alg:top-level}) will eventually match the root state $q_0$ and a list of states $q'_1,\mydots,q'_{2l}$ where \begin{enumerate}
    \item $q_0=(P, \{\domtrace_{[i,m]},\env\bindsto\actiontrace_{[i,m]}, \emptytrace\})$,
    \item
$S'_i[z\bindsto L[1]]\in L(q'_i,\fta)$, and
    \item $S'_i[z\bindsto L[2]]\in L(q'_{l+i},\fta)$.
\end{enumerate}
During \textsc{SpeculateFTA}, anti-unifying any pair $q'_i, q'_{l+i}$
 will yield an anti-unifier $\dataexpr$ with $\eval{\env}{\dataexpr}{L}$, because per the construction in \textsc{InitFTA}, all states containing $L[1]$ (resp.\ $L[2]$) will also contain all alternative data expressions $\dataexpr[1]$ (resp.\ $\dataexpr[2]$).
As a result, the FTA after speculation, $\fta_s$, will represent the program
$\fordataloop( \dataexpr, \lambda \datavar. \seqprog(S'_1,\mydots,S'_k))$. 
This is exactly the $S$ we want to synthesize.

Finally, by completeness of \textsc{EvaluateFTA}, $S$ is represented by a final state $q$ in $\fta_e$ as well and 
$q$ has context $\{P, \domtrace_{[i,m]} \}$.
By soundness of \textsc{EvaluateFTA}, $\fta_e$ is well-annotated, 
 so the footprint follows the evaluation and has to be $\{\domtrace_{[i, m]},\env\bindsto \actiontrace_{[i, l]}\}$ for some index $l$ that marks the end of loop evaluation (or the end of the trace)  satisfying $\eval{\domtrace_{[i, m]},\env}{S}{\actiontrace_{[i, l]}, \domtrace_{[l+1, m]}}$.

On the other hand, according to \textsc{InitFTA}, there exists a state $q'$ in $\fta$ that represents program $\seqprog(a_{l+1},\mydots, a_m)$ with footprint $\domtrace_{[l+1, m]},\env\bindsto \actiontrace_{[l+1, m]}, \emptytrace$.
Because the root state $q_0$ has footprint $\domtrace_{[i,m]},\env\bindsto\actiontrace_{[i,m]}, \emptytrace$. By definition \textsc{MergeFTAs} will add a transition $\seqprog(q, q')\rightarrow q_0$ to $\fta$,
so $q_0$ now represents $\seqprog(S, a_{l+1},\mydots,a_m)$.

\end{itemize}

\end{proof}

\subsection{The main theorem}

Combining everything together, we have the following main theorem:

\newtheorem*{theoremmain}{Theorem 4.5}

\begin{theoremmain}
Given action trace $\actiontrace$, DOM trace $\domtrace$ and input data $\inputdata$, 
our synthesis algorithm always terminates. Moreover, if there exists a program in our grammar (shown in Figure~\ref{fig:alg:dsl}) that generalizes $\actiontrace$ (given $\domtrace$ and $\inputdata$) and satisfies the condition that (1) every loop has at least two iterations exhibited in $\actiontrace$ and (2) its final expression is a loop, then our synthesis algorithm (shown in Algorithm~\ref{alg:top-level}) would return a program that generalizes $\actiontrace$ (given $\domtrace$ and $\inputdata$) upon FTA saturation. 
\end{theoremmain}

\begin{proof}
    Our algorithm always terminates since there are only finitely many ways to reroll a trace.
    Moreover, if there exists a program $P$ that generalizes $A$ and satisfies condition (1) and (2), such a program also reproduces $A$ and its last loop expression produces at least one action (since any loop expression has at least two exhibited iterations). 
    By completeness (theorem~\ref{proof:main:completeness}) of our algorithm, $P$ is represented by $\fta$. 
    As a result, the set of generalizing programs is not empty, and
    \textsc{Rank} will heuristically pick such a program with smallest size.
    By soundness (theorem~\ref{proof:main:soundness}) of our algorithm, the returned program is always correct with respect to the given $\domtrace$ and $I$.
\end{proof}

\label{sec:proofs}
\newpage

\section{Case Study: Large Language Models}

Given the recent advances in large language models (LLMs) and exploding interests in applying them for program synthesis, we perform a case study where we use LLMs to solve program synthesis problems in our domain of web automation. 

\newpara{Setup.}
Given a benchmark (consisting of an action trace $\actiontrace$ together with candidate selectors for each action in $\actiontrace$), we construct a prompt that contains (i) an English description of our web automation language's syntax and semantics, (ii) a brief explanation of the web automation program synthesis problem (i.e., reproduce the demonstration) together with a small example, and (iii) the input action trace $\actiontrace$ --- including candidate selectors --- that we aim to synthesize a program for. 
Here, (i) and (ii) are shared across benchmarks while (iii) is created in a per-benchmark manner. 
Once constructed, the prompt is fed to the model which is then queried for \finalized{3} completions. We define a benchmark {as} ``solved'' if any of the completions gives an intended program.
We conduct this evaluation using the GPT-3.5 API~\cite{chatgpt-link}, which is one of the most popular LLMs publicly available nowadays. 

A unique technical challenge in our domain is that, we have a large number of candidate selectors as part of the problem specification --- as shown in RQ2, it is necessary to consider multiple hundreds of them in order to solve a decent number of benchmarks. 
The GPT-3.5 model we use, however, supports (only) up to 16,385 tokens, which unfortunately is not large enough to hold all necessary selectors for most of our benchmarks. 
Therefore, in this experiment, we use a very small subset of our language that contains only $\forselectorsloop$ loops (with loop index starting at 1) and primitive expressions (like $\clickstatement$). 
We further focus on benchmarks that can be solved using absolute selectors. 
This leads us to a total of 14 benchmarks, where the complexity of intended programs {ranges} from trivial 1-level loops with two expressions to a nested loop with 10 expressions. 
We truncate the input action trace to {include} at most 200 actions to fit within the model's context window. We also make sure this max length is sufficient to generate all intended programs for our 14 benchmarks.

\newpara{Results.}
While the model oftentimes generates syntactically reasonable programs (which is quite impressive), it fails to generate semantically correct programs for \finalized{7} (out of 14) benchmarks.
The \finalized{7} successfully solved cases all involve single-level loops with a maximum of 6 statements --- these are the simplest benchmarks in our suite. 
Among the \finalized{7} unsolved benchmarks, the model frequently fails to parametrize expressions in {the} loop body and struggles with constructing desired loop structures. 
{The model sometimes produces unstable results, claiming that a benchmark is unsolvable in one run while outputing programs in the other two trials.}

While these results are clearly poor, it is well-known that LLMs are sensitive to the prompting strategy~\cite{si2022prompting} and there might be a better method to prompt the model. Nevertheless, we believe these results demonstrate that our benchmarks are quite hard for state-of-the-art LLMs and require further research in relevant areas in order to better solve these problems.

\end{document}